\documentclass[11pts]{article}

\usepackage{amsmath}
\usepackage{amsthm}
\usepackage{amssymb}
\usepackage{caption}
\usepackage{subcaption}
\usepackage{multirow}

\usepackage[left=25.4mm,right=25.4mm,top=30mm,bottom=25.4mm,a4paper]{geometry}

\usepackage{graphicx}
\usepackage{natbib}

\newtheorem{theorem}{Theorem}

\newtheorem{remark}[theorem]{Remark}

\newtheorem{corollary}[theorem]{Corollary}

\newtheorem{assumption}{Assumption}
\newtheorem{notation}{Notation}

\newcommand{\E}{\mathbb E}
\newcommand{\e}{\mathrm e}

\newcommand{\D}{\mathrm{d}}
\newcommand{\F}{\mathcal F}

\newcommand{\Var}{\mathrm{Var}}

\begin{document}


\title{Marked Hawkes process modeling of price dynamics and volatility estimation}

\author{Kyungsub Lee\footnote{Department of Statistics, Yeungnam University, Gyeongsan, Gyeongbuk 38541, Korea} and Byoung Ki Seo\footnote{Corresponding author, School of  Management Engineering, UNIST(Ulsan National Institute of Science and Technology), Ulsan 44919, Korea}}
\maketitle

\begin{abstract}
A simple Hawkes model have been developed for the price tick structure dynamics incorporating market microstructure noise and trade clustering.
In this paper, the model is extended with random mark to deal with more realistic price tick structures of equities.
We examine the impact of jump in price dynamics to the future movements and dependency between the jump sizes and ground intensities.
We also derive the volatility formula based on stochastic and statistical methods and compare with realized volatility in simulation and empirical studies.
The marked Hawkes model is useful to estimate the intraday volatility similarly in the case of simple Hawkes model.
\end{abstract}

\section{Introduction}

In this paper, the tick dynamics of stock prices observed at ultra-high-frequency level are modeled based on the symmetric marked Hawkes process and the empirical properties of the price dynamics are examined.
The simple self and mutually excited Hawkes model for the price dynamics with the unit jump size incorporates the stylized facts of the ultra-high-frequency financial data, such as market microstructure noise and order clustering.
On the other hand, random size jumps, i.e., not a constant jump, as in the simple Hawkes model, are observed in the tick structure of equity markets, particularly when there is a high ratio between the stock price and minimum tick size.
By combining the Hawkes model with a mark structure, which has additional information for each event, more realistic model of the tick price dynamics is proposed to deal with the random size jumps.

Recent studies on (ultra)-high-frequency data and the market microstructure have been developed in several ways.
The volume of literature on the financial theory of the market microstructure and limit order book \citep{Rosu2009}, and the role of algorithmic trading at a high-frequency rate \citep{Foucault2012,Chaboud2014,Hoffmann2014} is increasing.
A number of studies focused on the reduced form or stochastic modeling of the limit order dynamics and order executions; the reader may refer to \cite{Lo2002}, \cite{Cont2010}, \cite{Malo2012}, \cite{Cont2013}, \cite{Abergel2013}.

The statistical property of the ultra-high-frequency data is also an important subject because they exhibit the distinctive characteristics from the macro price dynamics.
For example, care should be taken when applying the typical statistical methods to ultra-high-frequency data,  and when computing the realized volatility \citep{ABDL} due to microstructure noise, which refers to the mean reverting properties of the price processes at high frequency level.
Previous studies (\cite{Ait2005}, \cite{Zhang2005}, \cite{Hansen} and \cite{Ait2011}) measured the volatility of the return in the presence of market microstructure noise.
\cite{Huth2014} examined the lead/lag relationship between asset prices and showed that there are significant cross correlations in the futures/stock at the high-frequency contrast with the daily data cases where cross correlations are negligible.
The lead/lag relationship among the international index futures of different countries were also observed by \cite{Alsayed2014}.
For the semi-Markov model with price jumps to explain the micorstructure noise, consult \cite{Fodra}.

The financial asset price time series at the ultra-high-frequency level exhibits several autocorrelations that are not observed on a daily basis.
Under the tick structure with a minimum tick size of price variation, the price dynamics is a pure jump process that consists of up jumps and down jumps.
First, the frequency of up movements tends to increase with increasing frequency of the past down movements and vice versa.
This causes a mean reverting property in the price dynamics, even though the correlations last for less than a few seconds.
Second, there are also autocorrelations between the movements of the same direction.
This causes volatility clustering that is different from the clustering on the macro level, which is typically modeled by GARCH \citep{Bollerslev1986} or the stochastic volatility model \citep{Heston1993}, 
because the clustering properties in the tick structure last for only a few seconds.
The durations of the autocorrelations are much shorter than those of the autocorrelation observed on a daily level.

These properties are well incorporated into the Hawkes model, which belongs to the class of point processes and is introduced by \cite{Hawkes1971}.
Therefore, there has been an increase in the related work of modeling price dynamics based on the Hawkes process.
The bivariate Hawkes process was introduced to model the buy and sell order arrivals and the impact of the orders on future prices was examined \citep{Hewlett2006}.
The generalized Hawkes models was used to study the dependence between the occurrence of time trades and changes to the mid quote as well as the dependences between trading days \citep{Bowsher2007}.
\cite{Large2007} examined the market resilience after trades using the limit order book data and mutually excited Hawkes models.

\cite{Bacry2012} explained the non-parametric estimation method for the symmetric Hawkes process based on high-frequency futures data.
Based on the mutually exciting Hawkes process, \cite{Bacry2013} suggested the mathematical framework that incorporates the market microstructure noise and the Epps effect, which is the correlation between the returns of two different assets at high sampling frequency. 
The trade clustering properties of the price dynamics on the micro level was well incorporated by the self-excited Hawkes process \citep{Fonseca2014Self}.
The formulas for the moments and correlation function for the self and mutually excited Hawkes process were derived by \cite{Fonseca2014}.
A multivariate Hawkes process was introduced for the up and down price movements and buy and sell orders to explain the stylized facts of the market impact and microstructure \citep{Bacry2014}.
\cite{Zheng2014} suggested a multivariate constrained Hawkes process to describe the dynamics of the bid and ask prices.
\cite{LeeSeo} focused on the daily and intraday volatility estimation based on the symmetric Hawkes process and compared the result with the realized volatility.
For more about the kernel estimation in the Hawkes model, consult \cite{Bacry2016} and for the correlation and lead-lag relationship in a multi-asset model using the Hawkes process, consult \cite{Fonseca2016}.

The previous studies focused mainly on the simple point process model, where the jump size is constant.
In the present study, the existing simple Hawkes model was extended to the marked Hawkes model to handle more realistic price movements in stock markets where the random size jumps (marks) are introduced.
A marked point process was introduced based on the ACM-ACD model, where points are the transaction time and the marks are information on the transaction \citep{Russell2005}.
The Hawkes process was adopted to explain the aftermath effects of the marks, which is more convenient for calculating the useful formula.
The future effects of the marks depend on the absolute sizes of the marks and hence a linear impact function is introduced to deal with the future impact of the mark.
Our empirical study shows that the estimates of the slope parameter of the impact function are significant positive values in stock markets.
This suggests that the larger marks tend to magnify the future intensities more than the smaller marks.
For the distribution of the mark, a specific distribution is not assumed in this paper but the empirical distribution is used for estimations and volatility calculation.
Our model is not limited to the independent mark as the empirical studies show the intensity dependent mark distribution.

The remainder of the paper is organized as follows.
In Section~\ref{Sect:point}, the marked Hawkes model is proposed to describe the tick price dynamics of equities.
Section~\ref{Sect:simul} and Section~\ref{Sect:empirical} present the simulation results and empirical studies, respectively.
Section~\ref{Sect:concl} concludes the paper.
The mathematical proofs are reported in the Appendix.

\section{Symmetric marked Hawkes model}\label{Sect:point}

\subsection{Marked point process}

This subsection introduces the basic concepts of marked point processes.
The mathematical framework is in line with \cite{Daley}.
With the given complete separable metric state space, $\mathcal X$, a point process is a measure to count the number of random events that occur in an open set that belongs to the $\sigma$-field of $\mathcal X$'s Borel set, $\mathcal B_{\mathcal X}$.
To deal with random events, a probability space $(\Omega, \mathbb P)$ is introduced.
A point process $N(\omega)$ for $\omega \in \Omega$, or simply $N$, is the counting measure on $\mathcal B_{\mathcal X}$, i.e., $N(A)=N(\omega,A)$ is a random non-negative finite integer for any bounded $A \in \mathcal B_{\mathcal X}$.
In other words, $N(A)$ is a random variable that counts the number of events in $A$. 
From now on, the term $\omega$ is omitted for the notational simplicity.

A marked point process is a more complex model that is introduced to describe not only the location of random events but also additional information, called the mark, attached to each event.
A marked point process is a point process $\{(t_\ell, k(t_\ell) )\}$ with locations $t_\ell \in \mathbb R$ and marks $k(t_\ell)$ in a mark space $\mathcal K$.
The location space is not necessarily $\mathbb R$, but in this paper, the space is defined as a real line to model the price movements over time.
In addition, the mark space $\mathcal K = \mathbb Z^+ $ is the space of price jump sizes. 
This suggests that the absolute jump sizes of the price movements are represented by positive integer multiples of a minimum jump size, $\delta.$ 

For price dynamics modeling, there are two marked point processes, $N_1$ and $N_2$, which represent the up and down movements, respectively.
This study assumes that the probabilistic properties of $k_1$ and $k_2$, the mark size of $N_1$ and $N_2$, respectively, are the same.
For each $i=1,2$, the ground process $N_{gi}(\cdot) = N_i (\cdot,\mathbb Z^+)$, which is the marginal process for locations, is itself a point process.
Each ground process, $N_{gi}$, describes the arrival times of up or down price movements.

The price process is assumed to be represented by a two dimensional marked Hawkes point process, which belongs to the marked point processes.
Let $\lambda_{gi}(t)$ be the conditional intensity of $N_{gi}$ upon filtration $\F_{t-}$ which is the
 minimal $\sigma$-algebra generated by the history of the marked point processes of $N_1$ and $N_2$ on $(-\infty, t)$.
The conditional intensities are $\F_t$-adapted stochastic processes, that heuristically speaking, satisfy
$$\lambda_{gi}(t)\D t = \E[N_{gi}(t+\D t) -N_{gi}(t)| \F_{t-}].$$
Owing to the discontinuity, the intensities may not be unique at the discontinuities and $\lambda_{gi}(t)$ are considered the left continuous modifications in general, but if the intensities are used as integrators of some stochastic integrations, then the right continuous modification of $\lambda_{gi}$ is used as the integrators.
The marked Hawkes model for the price dynamics does not need to be symmetric in general.
However, for computational ease, it is convenient to assumed that the joint distributions of $(k_1, \lambda_{g1})$ and $(k_2, \lambda_{g2})$ are the same, for example, $\E[k_1]=\E[k_2]$.

The marked Hawkes process assumption implies that, for each $i=1,2$, the intensities of the ground processes satisfy
\begin{align}
\lambda_{gi}(t) &= \mu + \sum_{j=1,2} q_{ij} \int_{(-\infty,t)\times\mathbb Z^+}  g_{ij}(k_j) \phi_{ij}(t-u) N_j(\D u \times \D k_j)\label{Eq:ground}\\
&= \mu + \sum_{j=1,2}  q_{ij} \sum_{-\infty < u_\ell < t}  g_{ij}(k_j(u_\ell)) \phi_{ij}(t-u_\ell)\label{Eq:ground2}
\end{align}
where $u_\ell$ are the event times.
The mark size is represented by $k_j$ in the integration form with the counting measure, and is represented by $k_j(u_\ell)$ in the summation form with the associated event time $u_\ell$.
With the counting measure, inside the integration, $k_j$ can be considered to be function on $\mathbb R \times \mathbb Z^+$ such that $(u, k_j) \rightarrow k_j$.
If the occurrence time of $k_j$ needs to be specified, then the mark is written as $k_j(u)$ as in Eq.~\eqref{Eq:ground2} with some time notation, $u$, indicating that $k_j(u)$ takes place at time $u$.

The Hawkes processes generated by the above intensities are defined by the ancestor-offspring argument (as long as the intensities are finite).
The immigrant ancestor of type $i$ with mark $k$ arrives at the system from outside in a Poisson process at rate $\mu$.
These ancestors generate offspring and the generated offspring become the new ancestor to generate new offspring.
Owing to the ancestor type $j$ born at time $u$ with mark $k$, whether immigrant or not, the type $i$ offspring are generated by a Poisson process with a rate $q_{ij}g_{ij}(k)\phi_{ij}(t-u)$ at time $t$.
The Poisson rate is emphasized by $g_{ij}(k)$ at time $u$ and decreases with $\phi_{ij}(t-u)$ as $t$ increases.
A normalization method to determine $q_{ij}, g_{ij}(k)$ and $\phi_{ij}(t-u)$ is used in general,
since the combination of $q_{ij}g_{ij}(k)\phi_{ij}(t-u)$ is not unique.

We assume the exponential decay kernel for $\phi_{ij}(t-u) = \phi(t-u) = \beta \e^{-\beta(t-u)}, \beta>0,$ which is normalized such that
$$\int_0^\infty \phi(\tau) \D \tau = 1.$$
The impact of mark $g_{ij}$ is also normalized in the sense that $\E [g_{ij}(k)] = 1$.
In addition, $q_{ij}>0$ are called the branching coefficients and $\mathbf Q := \{q_{ij}\}_{i,j=1,2}$ is called the branching matrix.
With an exponential decay kernel, $(\lambda_{g1}, \lambda_{g2})$ is Markovian and the calculations in this paper depend largely on this property.

This paper considers the case that the distribution of mark of type $i$ may depend on $\lambda_{gi}$ and hence the conditional distribution of the mark is represented by $f(k_i|\lambda_{gi}(t))$.
On the other hand, this paper does not assume the specific parametric distribution for the mark size except when paths are generated under the simulation study.
The estimation procedures and volatility analysis are performed without specifying the mark distribution.

The counting measure $N_i$ can be interpreted as a stochastic jump process.
To apply the stochastic integration theory later, without notational confusion, the associated jump processes is defined as
$$ N_i(t) = \int_{(0,t]\times\mathbb Z^+} k_i  N_i(\D u \times \D k_i) $$
where in the l.h.s., the stochastic process $N_i$ is represented by the sole parameter $t$ and in the r.h.s.,
the measure $N_i$ is represented by both the location $t$ and mark size $k$.
In the stochastic process representation, $N_i(t)$ counts the number of events over $(0,t]$ with weight $k_i$.
The jump processes $N_i(t)$ are considered to be right continuous with left limits.
Similarly, the ground processes also have the stochastic process representations:
$$ N_{gi}(t) = \int_{(0,t]\times\mathbb Z^+} N_i(\D u \times \D k_i) $$
which counts the number of events over $(0,t]$ without considering the jump size $k_i$.
As a jump process, the ground process is also regarded to be right continuous with left limits.

\subsection{Linear impact function}
A linear impact function of the mark is introduced.
The impact function $g_{ij}$ and the distribution of the mark should satisfy some additional criteria so that the marked Hawkes process does not blow up.

\begin{assumption}\label{Assumption}
(i) The ground intensities $\lambda_{gi}$ are assumed to be stationary.\\
(ii) The impact functions have the same formula for all $i,j = 1,2$ and are linear with a slope parameter $\eta$:
$$ g(k) := g_{ii}(k) = g_{ji}(k)= \frac{ 1+(k-1)\eta }{\E[1+(k-1)\eta]}.$$
(iii) The branching matrix is symmetric with 
$$q_s:=q_{11}=q_{22}= \frac{\alpha_s}{\beta} \E[1 + (k-1) \eta], \quad q_c:=q_{12}=q_{21}=\frac{\alpha_c}{\beta} \E[1 + (k -1) \eta].$$
(iv) For $i=1,2$, we assume
\begin{align}
\E[k_i \lambda_{gi}(t)]  = K_{i \lambda_{gi}} \E[\lambda_{gi}(t)] \label{Eq:K}
\end{align}
for some constant $K_{i \lambda_{gi}}>0$ and
\begin{align}
\{1 + ( K_{i \lambda_{gi}} -1)\eta\} \left(\frac{\alpha_s}{\beta} + \frac{\alpha_c}{\beta}\right) < 1.\label{Eq:condition}
\end{align}
In addition, the joint distributions of $(k_1, \lambda_{g1})$ and $(k_2, \lambda_{g2})$ are the same and we have $K_{1 \lambda_{g1}} =  K_{2 \lambda_{g2}}$.
\end{assumption}

The condition of Eq.~\eqref{Eq:condition} is similar to the existence condition of the simple symmetric Hawkes process except for the additional term $\{1 + ( K_{i \lambda_{gi}} -1)\eta\}$.
Indeed, Assumption~\ref{Assumption} (ii)$\sim$(iv) leads to the weak stationarity of $\lambda_{gi}$.
Under Eq.~\eqref{Eq:ground} with the exponential decay function, we have
\begin{align}
\lambda_{g1}(t) = \mu &+ q_s \int_{(-\infty,t)\times\mathbb Z^+} g(k_1) \beta \e^{-\beta(t-u)} N_1(\D u\times \D k_1) \nonumber\\
&+ q_c \int_{(-\infty,t)\times\mathbb Z^+} g(k_2) \beta \e^{-\beta(t-u)} N_2(\D u\times \D k_2), \label{Eq:lambdag1}\\
\lambda_{g2}(t) = \mu &+ q_c \int_{(-\infty,t)\times\mathbb Z^+} g(k_1) \beta \e^{-\beta(t-u)} N_1(\D u\times \D k_1) \nonumber\\
&+ q_s \int_{(-\infty,t)\times\mathbb Z^+} g(k_2) \beta \e^{-\beta(t-u)} N_2(\D u\times \D k_2).\label{Eq:lambdag2}
\end{align}
The point process $(N_1, N_2)$ defined under the above ground intensity is then a two dimensional marked self and mutually excited Hawkes process with a linear impact function.

Assuming the integrand of the following formula is integrable, a predictable finite variation process
$$ \int_\cdot^t  \E[ g(k_i) |\lambda_{gi}(u) ]\lambda_{gi}(u)  \beta \e^{-\beta(t-u)} \D u$$ 
is a compensator for 
$$\int_{(\cdot,t) \times \mathbb Z^+} g(k_i) \beta \e^{-\beta(t-u)} N_i(\D u \times \D k_i),$$
and hence 
$$ \int_{(\cdot,t) \times \mathbb Z^+} g(k_i) \beta \e^{-\beta(t-u)} N_i(\D u \times \D k_1) - \int_\cdot^t  \E[ g(k_i) |\lambda_{gi}(u) ]\lambda_{gi}(u)  \beta \e^{-\beta(t-u)} \D u$$
is a martingale.
Therefore, by taking the unconditional expectation for the ground intensity formulas in Eqs.~\eqref{Eq:lambdag1}~and~\eqref{Eq:lambdag2}, we have
\begin{align*}
\E[\lambda_{g1}(t)] = \mu &+ q_s \int_{-\infty}^{t} \E[\E[ g(k_1) |\lambda_{g1}(u) ]\lambda_{g1}(u)]\beta \e^{-\beta(t-u)}\D u  \\
&+ q_c \int_{-\infty}^{t}\E[ \E[g(k_2)| \lambda_{g2}(u) ]\lambda_{g2}(u)]\beta \e^{-\beta(t-u)}\D u,\\
\E[\lambda_{g2}(t)] = \mu &+ q_c \int_{-\infty}^{t}\E[\E[ g(k_1) |\lambda_{g1}(u) ]\lambda_{g1}(u)]\beta \e^{-\beta(t-u)}\D u \\
&+ q_s \int_{-\infty}^{t} \E[ \E[g(k_2)| \lambda_{g2}(u) ]\lambda_{g2}(u)]\beta \e^{-\beta(t-u)}\D u 
\end{align*}
and by Eq.~\eqref{Eq:K},  
$$ \E[\E[ g(k_i) |\lambda_{gi}(u) ]\lambda_{gi}(u)] = \E[g(k_i)\lambda_{gi}(u) ]  = \frac{\{ 1 + (K_{i \lambda_{gi}}-1)\eta \} \E[\lambda_{gi}(u)]}{ 1+(\E[k]-1)\eta}$$
where $\E[k] = \E[k_i]$ since $k_1$ and $k_2$ have the same distributional property.
We write
\begin{align}
\E[\lambda_{g1}(t)] = \mu &+ \frac{\alpha_s}{\beta} \int_{-\infty}^{t} \{1 + (K_{1 \lambda_{g1}}-1)\eta\}\E[\lambda_{g1}(u)] \beta \e^{-\beta(t-u)}\D u \nonumber \\
&+ \frac{\alpha_c}{\beta} \int_{-\infty}^{t} \{1 + ( K_{2 \lambda_{g2}}-1)\eta\}\E[\lambda_{g2}(u)]\beta \e^{-\beta(t-u)}\D u,\label{Eq:integ_sys1}\\
\E[\lambda_{g2}(t)] =\mu &+  \frac{\alpha_c}{\beta}  \int_{-\infty}^{t} \{1 + ( K_{1 \lambda_{g1}}-1)\eta\}\E[\lambda_{g1}(u)]\beta \e^{-\beta(t-u)}\D u \nonumber \\
&+ \frac{\alpha_s}{\beta} \int_{-\infty}^{t} \{1 + ( K_{2 \lambda_{g2}}-1)\eta\}\E[\lambda_{g2}(u)]\beta \e^{-\beta(t-u)}\D u \label{Eq:integ_sys2}
\end{align}
or, in a system of linear differential equation,  
\begin{align*}
\begin{bmatrix}
\dfrac{\D \E[\lambda_{g1}(t)]}{\D t} \\  
\dfrac{\D \E[\lambda_{g2}(t)]}{\D t}  
\end{bmatrix}
=
\begin{bmatrix}
\alpha_s \{1 + ( K_{1 \lambda_{g1}}-1)\eta\} - \beta & \alpha_c \{1 + ( K_{1 \lambda_{g1}}-1)\eta\} \\
\alpha_c \{1 + ( K_{1 \lambda_{g1}}-1)\eta\} & \alpha_s \{1 + ( K_{1 \lambda_{g1}}-1)\eta\} - \beta
\end{bmatrix}
\begin{bmatrix}
\E[\lambda_{g1}(t)] \\
\E[\lambda_{g2}(t)]
\end{bmatrix}
+
\begin{bmatrix}
\beta \mu \\
\beta \mu \\
\end{bmatrix}
\end{align*}
where $K_{1 \lambda_{g1}}= K_{2 \lambda_{g2}}$ is used.
The eigenvalues of the system are 
$$ (\xi_1, \xi_2) = \left(-\beta+ (\alpha_s -\alpha_c)\{1 + ( K_{1 \lambda_{g1}}-1)\eta\}, -\beta+ (\alpha_s + \alpha_c)\{1 + ( K_{1 \lambda_{g1}}-1)\eta\}\right)$$
and the solution of the above system is  
\begin{align*}
\begin{bmatrix}
\E[\lambda_{g1}(t)] \\
\E[\lambda_{g2}(t)]
\end{bmatrix}
=\frac{-\lambda_{g1}(0)+\lambda_{g2}(0)}{2}\e^{\xi_1 t} \begin{bmatrix}-1\\1\end{bmatrix} 
+\frac{\lambda_{g1}(0)+\lambda_{g2}(0)}{2}\e^{\xi_2 t} \begin{bmatrix}1\\1\end{bmatrix} 
-\frac{\mu \beta}{\xi_2} \left(1-\e^{\xi_2 t} \right) \begin{bmatrix}1\\1\end{bmatrix}.
\end{align*}
If Eq.~\eqref{Eq:condition} holds, then $\xi_1, \xi_2 < 0$ and the solution converges to a constant as $t \rightarrow \infty$.
A similar argument is applied to the second moments of $\lambda_{g1}(t)$ and $\lambda_{g2}(t)$
and $\lambda_{gi}(t)$ are weakly stationary at least in the long-run 
or by assuming that $\lambda_{gi}(0)$ are equal to their long-run expectations.

Note that by the symmetry and the stationarity of $\lambda_{gi}$, Eqs~\eqref{Eq:integ_sys1} and \eqref{Eq:integ_sys2} lead directly to
\begin{align*}
\E[\lambda_{g1}(t)] = \mu &+ \frac{\alpha_s}{\beta} \{1 + (K_{1 \lambda_{g1}}-1)\eta\}\E[\lambda_{g1}(t)] + \frac{\alpha_c}{\beta}  \{1 + (K_{1 \lambda_{g1}}-1)\eta\}\E[\lambda_{g2}(t)],\\
\E[\lambda_{g2}(t)] = \mu &+  \frac{\alpha_c}{\beta} \{1 + (K_{1 \lambda_{g1}}-1)\eta\}\E[\lambda_{g1}(t)] + \frac{\alpha_s}{\beta} \{1 + (K_{1 \lambda_{g1}}-1)\eta\}\E[\lambda_{g2}(t)].
\end{align*}
and
$$ \left( \mathbf{I} - \{1 + ( K_{1 \lambda_{g1}} -1)\eta\} \begin{bmatrix} \dfrac{\alpha_s}{\beta} & \dfrac{\alpha_c}{\beta} \\ \dfrac{\alpha_c}{\beta} & \dfrac{\alpha_s}{\beta} \end{bmatrix} \right) \begin{bmatrix}\E[\lambda_{g1}(t)] \\ \E[\lambda_{g2}(t)]\end{bmatrix} = \begin{bmatrix}\mu \\ \mu \end{bmatrix}. $$
By the symmetry between $\lambda_{g1}$ and $\lambda_{g2}$,
\begin{equation}
\E[\lambda_{g1}(t)] = \E[\lambda_{g2}(t)] = \frac{\mu \beta}{\beta- (\alpha_s+\alpha_c)\{1 + (K_{1 \lambda_{g1}} -1)\eta\} }.\label{Eq:Elambdag1}
\end{equation}
Therefore, if condition~\eqref{Eq:condition} is satisfied, then the ground processes are well defined, i.e., the expectation of the ground intensities are positive and finite.

\subsection{Second moment property}

In this subsection, the volatility formula of the asset return generated by the marked Hawkes processes under the symmetric model is calculated.
The symmetry of the model implies that $i$ and $j$ are interchangeable which makes the formula simple.
In the following notation, various $K$s and $\alpha$s are defined in a similar manner to those in Eq.~\eqref{Eq:K}, which simplify the notations.

\begin{notation}\label{Notation}
For a jump process $X$ such as $\lambda_{gi}$ and $N_i$, let
\begin{align*}
\E[k_i X(t)] &= K_{iX} \E[X(t)], \quad \E[k_i^2 X(t)] = K^{(2)}_{iX} \E[X(t)].
\end{align*}
(In the previous, when $X = \lambda_{gi}$, subscript is omitted for simplicity as in Eq.~\eqref{Eq:K}, i.e., $K=K_{i\lambda_{gi}}$.)
Furthermore,
\begin{align*}
\bar{\bar K} &= 1+2(K_{1\lambda_{g1}}-1)\eta + (K^{(2)}_{1\lambda_{g1}} -2K_{1\lambda_{g1}} +1 )\eta^2 = 1+2(K_{2\lambda_{g2}}-1)\eta + (K^{(2)}_{2\lambda_{g2}} -2K_{2\lambda_{g2}} +1 )\eta^2,\\
\bar K &= K_{1\lambda_{g1}}+(K_{1\lambda_{g1}}^{(2)}-K_{1\lambda_{g1}})\eta = K_{2\lambda_{g2}}+(K_{2\lambda_{g2}}^{(2)}-K_{2\lambda_{g2}})\eta,\\
\breve \alpha &= \alpha\{1+(K_{1\lambda^2_{g1}}-1)\eta\}= \alpha\{1+(K_{2\lambda^2_{g2}}-1)\eta\},\\
\tilde \alpha &= \alpha\{1+(K_{2\lambda_{g1}\lambda_{g2}}-1)\eta\} =\alpha\{1+(K_{1\lambda_{g2}\lambda_{g1}}-1)\eta\},\\
\acute \alpha &= \alpha \{ 1 + (K_{1\lambda_{g1} N_1} - 1 )\eta \} = \alpha \{ 1 + (K_{2\lambda_{g2} N_2} - 1 )\eta \},\\
\grave \alpha &= \alpha \{ 1 + (K_{2\lambda_{g2} N_1} - 1 )\eta \} = \alpha \{ 1 + (K_{1\lambda_{g1} N_2} - 1 )\eta \}, 
\end{align*}
and
\begin{align*}
&\mathbf{M} = \begin{bmatrix} \breve \alpha_s - \beta & \tilde \alpha_c \\ \breve \alpha_c & \tilde \alpha_s - \beta \end{bmatrix}, \quad  \mathbf{M}_2 = \begin{bmatrix} \acute \alpha_s - \beta & \grave \alpha_c \\ \acute \alpha_c & \grave \alpha_s - \beta \end{bmatrix}, 
\\
&\mathbf{K} = \begin{bmatrix} K_{1\lambda_{g1}} & 0 \\ 0 & K_{1\lambda_{g1}} \end{bmatrix}, \quad \mathbf{K}_2 = \begin{bmatrix} K_{1\lambda_{g1}^2} & 0 \\ 0 & K_{2\lambda_{g1}\lambda_{g2}} \end{bmatrix}, \mathbf{K}_3 = \begin{bmatrix} K_{1\lambda_{g1} N_1} & 0 \\ 0 & K_{1\lambda_{g1} N_2} \end{bmatrix}.
\end{align*}
\end{notation}

\begin{theorem}\label{Thm:var}
Let $(N_1, N_2)$ be a two dimensional marked self and mutually excited Hawkes process with a linear impact function under Assumption~\ref{Assumption} with ground intensities of Eqs.~\eqref{Eq:lambdag1}~and~\eqref{Eq:lambdag2}.
If the price process $S_t$ follows 
$$ S_t = S_0 + \delta(N_1(t) - N_2(t))$$
with a minimum jump size $\delta$, then the unconditional variance of the return over $[0,t]$ is 
$$ \Var \left(\frac{S_t - S_0}{S_0} \right) = \frac{\delta^2}{S_0^2}\E[(N_1(t) -N_2(t))^2] $$
with
\begin{align*}
&\begin{bmatrix} \E[N_1^2(t)] \\ \E[N_1(t)N_2(t)] \end{bmatrix} = -\E[\lambda_{g1}(t)]\mathbf{K}_3 \left\{  \beta \mu \mathbf{K} \mathbf{M}_2^{-1} \begin{bmatrix}1\\1 \end{bmatrix}t^2 \right. \\
{}&\left. + \left( 2\mathbf{M}_2^{-1} \begin{bmatrix} \alpha_s \bar K \\ \alpha_c \bar K\end{bmatrix} 
- \mathbf{M}_2^{-1} \mathbf{K}_2 \mathbf{M}^{-1} 
\begin{bmatrix}  (\alpha_s^2 + \alpha_c^2) \bar{\bar K}  \\ 2\alpha_s\alpha_c\bar{\bar K} \end{bmatrix} 
 + 2\mathbf{M}_2^{-1}(\mathbf{K}\mathbf{M}_2^{-1}- \mathbf{K}_2 \mathbf{M}^{-1}) \begin{bmatrix} \beta\mu \\ \beta\mu \end{bmatrix} - \begin{bmatrix}K^{(2)}_{1\lambda_g1}/K_{1\lambda_{g1} N_1} \\ 0\end{bmatrix} \right)t \right\}.
\end{align*}
\end{theorem}
\begin{proof}
See \ref{Proof:var}.
\end{proof}
The result is slightly complicated but the following remark is useful for the practice.
\begin{remark}\label{Remark:vol}
By assuming $K_{1\lambda_{g1} N_{1}} \approx K_{1\lambda_{g1} N_2}$ and $K_{1\lambda^2_{g1}} \approx K_{1\lambda_{g1}\lambda_{g2}}$, we have
\begin{align}
\E[(N_1(t) -N_2(t))^2] &= 2(\E[N_1^2(t)] - \E[N_1(t)N_2(t)]) \nonumber\\
&\approx 2 K_{1\lambda_{g1} N_1} \E[\lambda_{g1}(t)] \left(  \frac{K_{1\lambda^2_{g1}}\bar{\bar K}(\alpha_s-\alpha_c)^2}{(\beta - \breve\alpha_s + \breve\alpha_c)(\beta - \acute\alpha_s + \acute\alpha_c)} + \frac{2(\alpha_s-\alpha_c)\bar K}{\beta - \acute\alpha_s + \acute\alpha_c}  + \frac{K^{(2)}_{1\lambda_{g1}}}{K_{1\lambda_{g1} N_1}}\right) t. \label{Eq:var}
\end{align}
Under the assumption, $\mathbf{M}, \mathbf{M}_2$ and $\mathbf{K}_2$ are symmetric and the variance formula becomes simple.
In addition, if all $K$s are equal to 1, then the variance formula is reduced to
\begin{align*}
\E[(N_1(t) -N_2(t))^2] = 2 \E[\lambda_{g1}(t)] \frac{\beta^2}{(\beta-\alpha_s-\alpha_c)^2} t.
\end{align*}
which is the same formula of the variance in the simple Hawkes model.
\end{remark}

\begin{corollary}\label{Cor:iidvol}
Under the assumption in Theorem~\ref{Thm:var}, if the marks $k_i$ are i.i.d., then $K_{iX} = K := \E[k_i]$ and $K^{(2)}_{iX} = K^{(2)} := \E[k^2_i]$ for all $X \in \{\lambda_{gi}, \lambda_{gi}\lambda_{gj}, \lambda_{gi}N_j : i,j=1,2\}$ and
$\breve \alpha = \tilde \alpha = \acute \alpha = \grave \alpha$.
Therefore,
$$ \E[(N_1(t) -N_2(t))^2] =2 K \E[\lambda_{g1}(t)] \left(  \frac{K\bar{\bar K}(\alpha_s-\alpha_c)^2}{(\beta - \breve\alpha_s + \breve\alpha_c)^2} + \frac{2(\alpha_s-\alpha_c)\bar K}{\beta - \breve\alpha_s + \breve\alpha_c}  + \frac{K^{(2)}}{K}\right) t$$
where $\bar{\bar K}$ and $\bar K$ are now represented by
\begin{align*}
\bar{\bar K} = 1+2(K-1)\eta + (K^{(2)} -2K +1 )\eta^2, \qquad
\bar K = K+(K^{(2)}-K)\eta .
\end{align*}
\end{corollary}

\subsection{Likelihood function}

To estimate the parameters in the intensity processes, such as $\mu, \alpha_s, \alpha_c, \beta$ and $\eta$, the log-likelihood needs to be computed.
The joint log-likelihood function of the realized interarrival of the jumps and marks of $(N_1, N_2)$ over the period $[0,T]$ is represented by
\begin{align}
&\left( \int_{(0, T]} \log \lambda_{g1}(u)  N_{g1}(\D u) + \int_{(0, T]} \log \lambda_{g2}(u)  N_{g2}(\D u) - \int_0^T (\lambda_{g1}(u) + \lambda_{g2}(u)) \D u  \right) \nonumber \\
&+ \left( \int_{(0,T]\times\mathbb Z^+ } \log f (k_1 | \lambda_{g1} (u) ) N_{1} (\D u\times\D k_1) + \int_{(0,T]\times\mathbb Z^+} \log f (k_2 | \lambda_{g2} (u) ) N_2(\D u \times \D k_2) \right) \nonumber\\
&=: \log L_g + \log L_m \label{Eq:likelihood}
\end{align}
where $f$ denotes the conditional distribution of the mark $k_i$ with a given $\lambda_{gi}$.

In the above formula, the log-likelihood is separated into two parts, $\log L_g$ and $\log L_m$.
The first part $\log L_g$ is the log-likelihood function of the ground intensity processes,
or more precisely, the joint log-likelihood function of the jump interarrival times.
The second part $\log L_m$ is the log-likelihood function of the conditional mark distribution.
When $\log L_g$ is used solely for the maximum likelihood estimation,
then $\log L_g$ is indeed the conditional log-likelihood function of the jump interarrival times with the conditions on the realized marks,
i.e., the estimation is the maximum likelihood estimation based on the interarrival times of jumps with given realized marks.

In the estimation procedure of the simulation and empirical studies later, no specific form of the joint distribution of the ground intensities and mark sizes is assumed 
but the empirical mark distribution is used to conduct the maximum likelihood procedure to maximize $\log L_g$, the conditional log-likelihood function of the jump interarrival times.
Because the empirical mark distribution is used, the $\log L_m$ part does not affect the estimations of $\boldsymbol\theta=(\mu, \alpha_s, \alpha_c, \beta, \eta)$.
On the other hand, if one assume a specific parametric modeling on the mark distribution by specifying a conditional distribution $f (k_i | \lambda_{gi})$ as in Subsection~\ref{Subsect:geo}
and also want to estimate the parameters in $f$, then 
$\log L_m$ is also affected by $\boldsymbol\theta$ 
since when $\lambda_{gi}$s are inferred, it is calculated using $\boldsymbol\theta$.

In another aspect, the estimation procedure only on $\log L_g$ is possible since the marks $k_i$ are observable.
If the mark sizes are unobservable, it is then inevitable to assume a parametric modeling on the joint or conditional distribution between $\lambda_{gi}$ and $k_i$ to infer the mark size $k_i$.
Owing to the parametric assumption on $f (k_i | \lambda_{gi} (u) )$, as mentioned before, the parameter family of $\boldsymbol\theta$ also appears in the formula of $\log L_m$,
and the changing values of $\boldsymbol\theta$ change not only the value $\log L_g$ but also  $\log L_m$.
Therefore, the estimator that maximizes $\log L_g$ may not converge to the estimator that maximizes $(\log L_g + \log L_m)$ as the sample size increases.

Fortunately, since all $k_i$ are observable in our empirical study, the observed realized mark sizes are used to compute $\log L_g$, not the inferred $k_i$ based on some parametric assumption.
Hence, the estimator of $\boldsymbol\theta$ that maximizes $\log L_g$ are the maximum likelihood estimator associated with the conditional joint distribution of the interarrival times of jumps with the given marks.
Note that $\log L_g$ and $\log L_m$ are separated not because the marks $k_i$ are independent from the ground intensities but because $\log L_g$ can be represented as the log-likelihood function of the conditional joint distribution of the interarrival times with the given marks.

In practice, $\log L_g$ is computed as follows.
With the presumed parameter values of $\alpha_s, \alpha_c, \beta$ and $\eta$,
we compute the inferred ground intensity processes $\lambda_{gi}$s based on the realized $N_{i}$s of the stock price process and Eqs.~\eqref{Eq:lambdag1} and \eqref{Eq:lambdag2}.
With the inferred ground intensity processes, the realized values of the stochastic integration part are calculated by
$$\int_{(0,T]} \log \lambda_{gi}(u) N_{gi}(\D u) = \sum_{n=1}^{N_{gi}(T)} \log  \lambda_{gi}(u_n) $$ where $u_n$ denotes the realized jump time by $N_{gi}$.
In addition, $\int_0^T (\lambda_{g1}(u) + \lambda_{g2}(u)) \D u $ is computed using the Riemann integral which has a closed form formula.
By repeating the above procedure with the changing presumed parameters, the numerical solver for the optimization try to find the global maximum of $\log L_g$.

Consider the concavity of the log-likelihood function.
If the Hessian of the log-likelihood function is negative semi-definite for all parameters with a given realization,
then the log-likelihood function is concave.
However, the formula of the Hessian of the marked Hawkes model is complicated, we instead examine the conditional concavity when $\beta$ and $\eta$ are fixed.
Note that, with the given realized jump times $t_i$, we have
\begin{align*}
\log L_{g1}(T) :={}& \int_{(0,T]} \log \lambda_{g1}(u) N_{g1}( \D u) - \int_0^T \lambda_{g1}(u) \D u\\
={}& \sum_{t_i < T} \left( \log \lambda_{g1}(t_i) - \int_{t_{i-1}}^{t_i} \lambda_{g1}(u) \D u \right) - \int_{t_N}^{T} \lambda_{g1}(u) \D u \\
={}& \sum_{t_i < T} \left( \log \lambda_{g1}(t_i) - \frac{1-\e^{-\beta \tau_i}}{\beta}\lambda_{g1}(t_i)  \right) - \int_{t_N}^{T} \lambda_{g1}(u) \D u
\end{align*}
where $t_N$ is the last jump time.
Using the definition of $\lambda_{g1}$, each term in the summation can be rewritten as follows:
\begin{align*}
&\log \lambda_{g1}(t_i) - \frac{1-\e^{-\beta \tau_i}}{\beta}\lambda_{g1}(t_i)\\
={}& \log \left[ \mu + (\lambda_{g1}(0) - \mu)\e^{-\beta t_i } + \int_{(0, t_i] \times \mathbb Z_+} \alpha_s \{ 1+(k_1-1)\eta\}\e^{-\beta (t_i -u)} N_1 (\D u \times \D k_1) \right.\\
&\quad \left. + \int_{(0, t_i]\times \mathbb Z_+} \alpha_c \{ 1+(k_2-1)\eta\}\e^{-\beta (t_i -u)} N_2 (\D u \times \D k_2) \right]\\
&-  \frac{1-\e^{\beta \tau_i}}{\beta} \left[ \mu + (\lambda_{g1}(0) - \mu)\e^{-\beta t_i } + \int_{(0, t_i] \times \mathbb Z_+} \alpha_s \{ 1+(k_1-1)\eta\}\e^{-\beta (t_i -u)} N_1 (\D u \times \D k_1) \right.\\
&\quad \left. + \int_{(0, t_i]\times \mathbb Z_+} \alpha_c \{ 1+(k_2-1)\eta\}\e^{-\beta (t_i -u)} N_2 (\D u \times \D k_2) \right].
\end{align*}

With fixed $\beta$ and $\eta$, the term is represented as
$$ \log(m\mu + a_s \alpha_s + a_c \alpha_c + C) -  \frac{1-\e^{\beta \tau_i}}{\beta} (m\mu + a_s \alpha_s + a_c \alpha_c + C)$$
for some constants $m, a_s, a_c,$ and $C$.
Therefore, for any fixed $\beta$ and $\eta$, the Hessian matrix of the term with respect to $\mu, \alpha_s$ and $\alpha_c$ is
\begin{align*}
\frac{1}{\lambda^2_{g1}(t_i)}
\begin{bmatrix}
- m^2 & -m a_s & -m a_c \\
- m a_s & - a_s^2 & -a_s a_c \\
-m a_c & -a_s a_c & -a_c^2
\end{bmatrix}
\end{align*}
which is negative semidefinite.
The conditional Hessian of the log-likelihood function $\log L_g$ with fixed $\beta$ and $\eta$ is represented by 
\begin{align*}
\sum_{t_i<T} \left(\frac{1}{\lambda^2_{g1}(t_i)} + \frac{1}{\lambda^2_{g2}(t_i)}\right)
\begin{bmatrix}
- m^2 & -m a_s & -m a_c \\
- m a_s & - a_s^2 & -a_s a_c \\
-m a_c & -a_s a_c & -a_c^2
\end{bmatrix}
\end{align*}
which is also negative semidefinite.

Therefore, at least if the parameter values of $\beta$ and $\eta$ are fixed, the concavity of the log-likelihood $\log L_g(\mu, \alpha_s, \alpha_c | \beta, \eta)$ as a function of $\mu, \alpha_s, \alpha_c$ can be guaranteed,
and we can assume that a numerical solver for the optimization will find the global maximum.
Consider the points of $\beta$ and $\eta$ over a sufficiently large and dense grid.
For each $\beta$ and $\eta$, we can find the conditional global maximum of $\log L_g(\mu, \alpha_s, \alpha_c | \beta, \eta)$ over the parameter space of $\mu, \alpha_s, \alpha_c$ due to the conditional concavity.
Now the interest is the shape of the conditional global maximums over the grid of $\beta$ and $\eta$.
If the conditional global maximum is still concave and the maximum of the conditional maximums can be found, 
then the numerical optimizer can be checked to determine if it finds the overall global maximum.
In the numerical procedure later, the shape of the conditional log-likelihood function over a grid of $\beta$ and $\eta$ will be examined.

\section{Simulation example}\label{Sect:simul}

\subsection{Symmetric model}\label{Subsect:geo}
In this paper, the specific distribution of the mark is  generally not assumed.
For the simulation study, however, it is necessary to assume a specific conditional distribution of the mark sizes to generate paths.
Suppose that the mark $k_i$ follows the conditional geometric distribution with
$$ p(\lambda_{gi}(u)) = \frac{1}{\min(d + c \lambda_{gi}(u), U)}$$
for some constants $c, d$ and $U$, i.e.,
$$ \mathbb P( k_i = n | \lambda_{gi}(u)) = p(\lambda_{gi}(u))(1-p(\lambda_{gi}(u)))^{n-1}.$$
This suggests that the conditional expectation of the mark size $k_i$ with a given ground intensity $\lambda_{gi}$ is
$$\E[k_i |\lambda_{gi}(u)]  = \min(d + c \lambda_{gi}(u), U)$$
for some slope $c$, intercept $d$, and upper bound $U$.
It is needed to set the upper bound for the conditional mean of the mark size to prevent a blow up of the marked Hawkes process.
With this setting, the conditional expectation of the impact depends on the current intensity:
$$ \E[g(k_i) | \lambda_{gi}(u)] = \frac{ 1+\{ \min(d + c \lambda_{gi}(u), U) - 1\} \eta}{\E[1+(k_i -1)\eta]}.$$

With each differently presumed conditional distribution and parameter setting, 500 sample paths of the two dimensional marked Hawkes process and corresponding ground intensities are generated.
The time horizon for the path is set to be 5.5 hours, which equals the time horizon used in empirical studies later.
The simulation mechanism is similar to the simple Hawkes models but it needs to incorporate the mark size and its future impacts.

With the realized interarrival times of the generated path and realized mark sizes, the maximum likelihood estimation is performed on the maximized $\log L_g$ in Eq.~\eqref{Eq:likelihood} and the results are listed in Table~\ref{Table:simul}.
The table consists of three panels with different parameter settings, which are presented in `True' rows.
For the first panel, $c=0.15, d=1.0, U=2.0$, for the second panel, $c = 0.18, d=1.0, U = 2.2$;
for the third panel, $c = 0.18, d=1.0, U = 3.5$;
 and for the fourth panel, $c=0.25, d=1.0, U=9$.

Because the likelihood for the ground processes was calculated, the estimates of $\mu, \alpha_s, \alpha_c, 
\beta$ and $\eta$ were computed but not for $c, d$ and $U$.
The sample mean of the estimates with 500 sample paths are reported in the row `mean'.
The row `std.' presents the sample standard deviations of each estimate with 500 samples.
The table shows that the estimates are consistent with the true values.

\begin{table}
\caption{Simulation study for the marked Hawkes model with 500 sample paths}\label{Table:simul}
\centering
\begin{tabular}{cccccccc}
\hline
 & $\mu$ & $\alpha_s$ & $\alpha_c$ & $\beta$ & $\eta$ & TSRV & H.Vol. \\
\hline
True & 0.1000 & 0.9500 & 0.8200 & 2.2500 & 0.1900 &  \\
mean & 0.0999 & 0.9496 & 0.8199 & 2.2487 & 0.1882 & 0.2807 & 0.2798 \\
std. & 0.0021 & 0.0194 & 0.0190 & 0.0326 & 0.0170 & 0.0333 & 0.0126\\
\hline
True & 0.1500 & 0.6200 & 0.5000 & 1.9000 & 0.2200 &  \\
mean & 0.1499 & 0.6193 & 0.5008 & 1.8999 & 0.2177 & 0.1317 & 0.1312 \\
std. & 0.0027 & 0.0172 & 0.0149 & 0.0419 & 0.0502 & 0.0101 & 0.0018\\
\hline
True & 0.3000 & 1.0500 & 0.9200 & 2.3000 & 0.0100\\
mean & 0.3003 & 1.0508 & 0.9206 & 2.3014 & 0.0094 & 0.6391 & 0.6291\\
std. & 0.0051 & 0.0136 & 0.0135 & 0.0229 & 0.0051 & 0.0562 & 0.0177\\
\hline
True & 0.2000 & 1.1000 & 1.2600 & 2.5700 & 0.0100\\
mean & 0.2002 & 1.0997 & 1.2610 & 2.5702 & 0.0099 & 4.4299 & 4.3628\\
std. & 0.0039 & 0.0154 & 0.0164 & 0.0238 & 0.0007 & 1.6973 & 0.6811\\
\hline
\end{tabular}
\end{table}

Figure~\ref{Fig:max} presents the global maximums of the conditional log-likelihood functions with various $\beta$ and $\eta$ under the first simulation setting as explained in the previous section.
The numerically computed conditional maximum points shows the concavity and it is expected that the numerical optimizer will find the global maximum in the procedure.

\begin{figure}
\centering
\includegraphics[width = 0.45\textwidth]{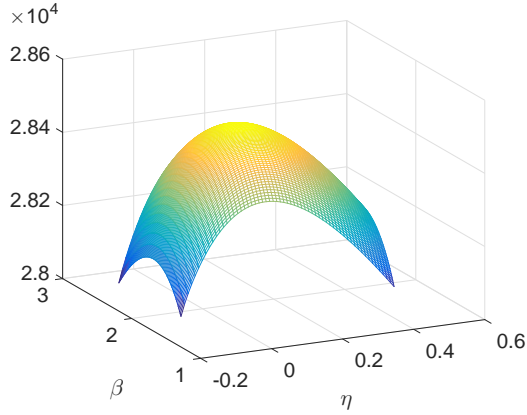}
\caption{Maximums of conditional log-likelihood function over $\beta$ and $\eta$}\label{Fig:max}
\end{figure}
 
To calculate the volatility, we need to compute $K$s in Notation~\ref{Notation}
which involves several unconditional expectations with the mark size, intensities and counting processes.
In the absence of the exact formula of the expectations due to the complicated relationship between the mark and the intensities, the following statistics are used for the expectations instead:
\begin{align}
\E[\lambda_{gi}(t)] &\approx \frac{1}{T}N_{gi}(T) \label{Eq:El}\\
\E[k_i \lambda_{gi}(t)]  &\approx \frac{1}{T}N_i(T)\label{Eq:Ekl}\\
\E[k_i^2 \lambda_{gi}(t)] &\approx \frac{1}{T}\int_{(0,T] \times \mathbb Z^+} k_i^2 N_{i}(\D u \times \D k_i)\label{Eq:Ek2l}\\
\E[\lambda^2_{gi}(t)] &\approx \frac{1}{T}\int_{(0,T] \times \mathbb Z^+} \lambda_{gi}(u)  N_{i}(\D u \times \D k_i)\label{Eq:ElN}\\
\E[k_i \lambda^2_{gi}(t)] &\approx \frac{1}{T}\int_{(0,T] \times \mathbb Z^+} k_i \lambda_{gi}(u) N_{i}(\D u \times \D k)\label{Eq:EklN}\\
\frac{1}{t}\E[\lambda_{gi}(t)N_i(t)] &\approx \frac{2}{T^2}\int_{(0,T]\times \mathbb Z^+} N_{i}(u-) N_{i}(\D u \times \D k_i)\label{Eq:ENN}\\
\frac{1}{t}\E[k_i \lambda_{gi}(t)N_i(t)] &\approx \frac{2}{T^2}\int_{(0,T] \times \mathbb Z^+} k_i N_{i}(u-) N_{i}(\D u \times \D k_i)\label{Eq:kNN}
\end{align}
where $[0, T]$ is an observation time interval.
To calculate the right hand sides, the realized $k_i$, $N_i$ and $N_{gi}$ of the generated paths and inferred $\lambda_{gi}$ from the estimates of $\mu, \alpha_s, \alpha_c, \beta$ and $\eta$ are used.
The inferred intensities $\lambda_{gi}$ are computed using Eqs.~\eqref{Eq:lambdag1}~\eqref{Eq:lambdag2}, once $\mu, \alpha_s, \alpha_c, \beta$ and $\eta$ are estimated.

The expectations of the ground intensities are approximated by the sample average of the total number of corresponding up or down moves per unit time in Eq.~\eqref{Eq:El}.
Similarly for $\E[k_i \lambda_{gi}(t)]$ where the counting process $N_i$ is used instead to compute the sample average.

The right hand side of Eq.~\eqref{Eq:Ek2l} is the sample average of the total number of jumps per unit time with weight $k_i^2$ for each jump and this approximates the left hand side.
For Eqs.~\eqref{Eq:ElN}~and~\eqref{Eq:EklN}, consider 
\begin{align*}
&\E\left[ \int_{(0,T] \times \mathbb Z^+}  \lambda_{gi}(u) N_{i}(\D u \times \D k) \right] = \int_{0}^T \E [\lambda^2_{gi}(t) ] \D t = T \E [ \lambda^2_{gi}(t) ]\\
&\E\left[ \int_{(0,T] \times \mathbb Z^+} k\lambda_{gi}(u) N_{i}(\D u \times \D k) \right] = \int_{0}^T \E [k_i \lambda^2_{gi}(t) ] \D t = T \E [k_i \lambda^2_{gi}(t) ].
\end{align*}
Figure~\ref{Fig:convergence} presents the convergences of the computed $K$ and $K^{(2)}$ with the above method as the sample size increases with the first parameter set in the simulation.

\begin{figure}
\centering
\includegraphics[width=0.4\textwidth]{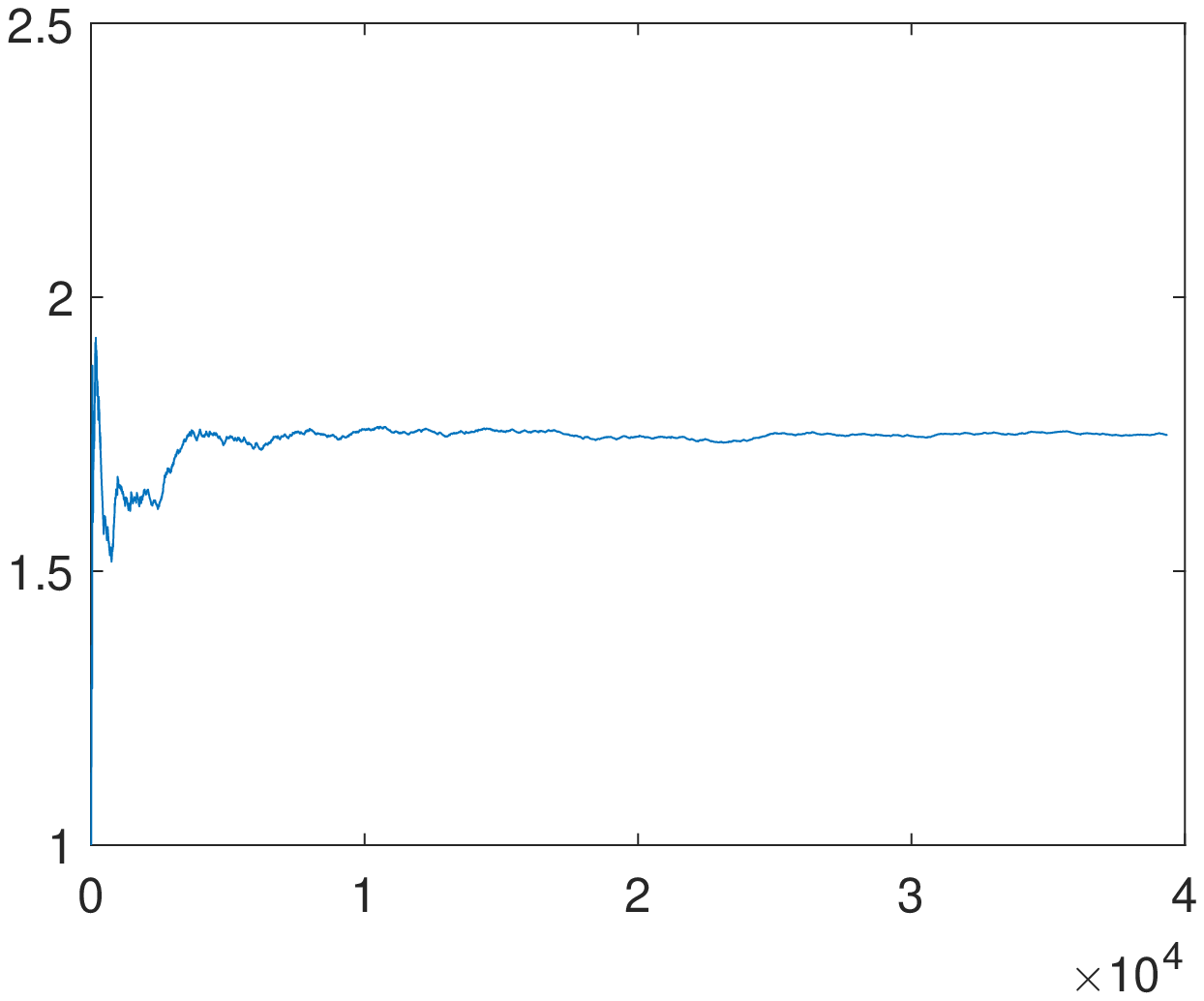}
\includegraphics[width=0.4\textwidth]{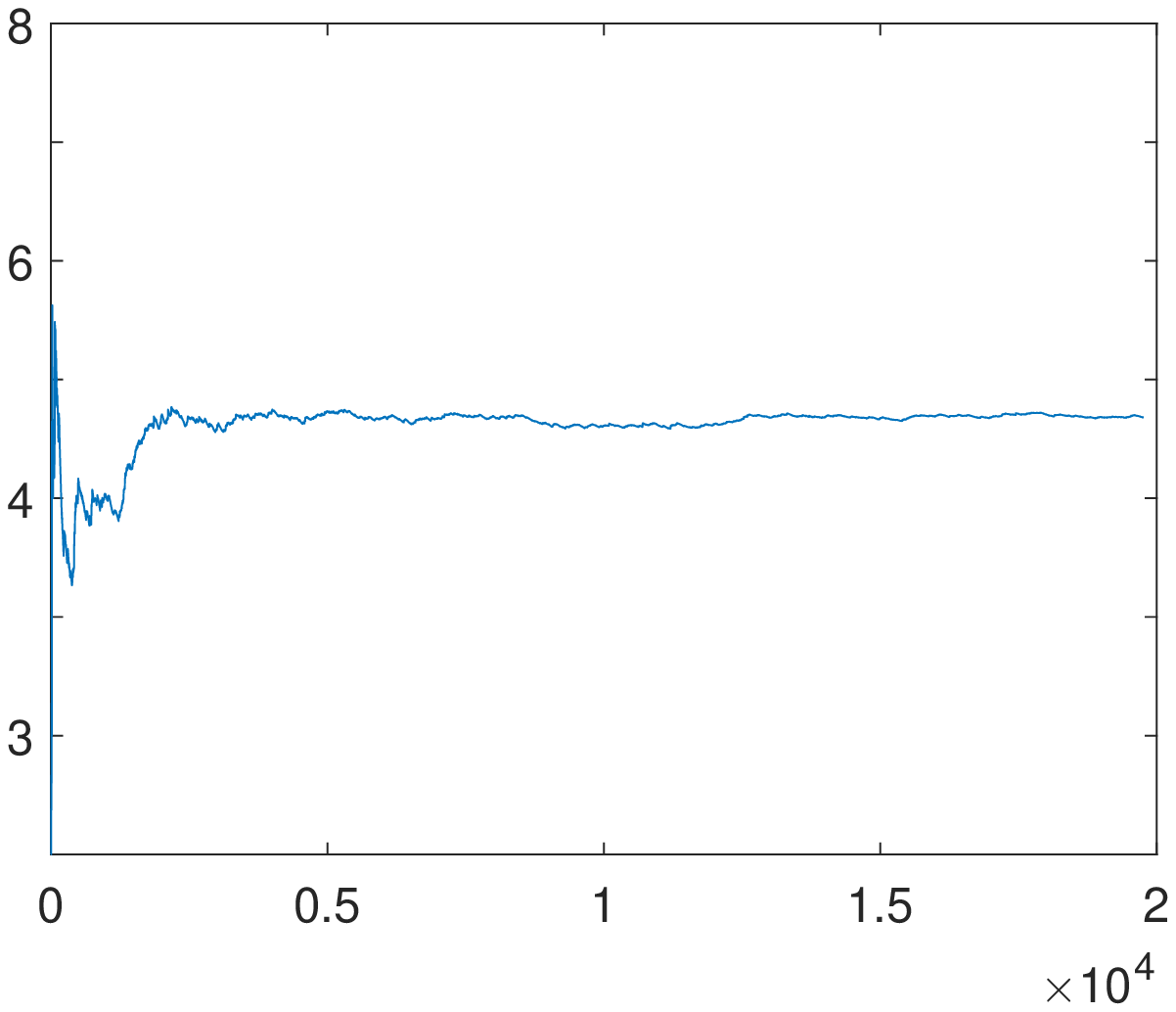}
\caption{Convergences of $K_{1\lambda_{g1}}$ and $K_{1\lambda_{g1}}^{(2)}$ as sample size increases}\label{Fig:convergence}
\end{figure}

Furthermore, $\E[\lambda_{gi}(t)N_i(t)] = c_1 t + c_2 $ for some constants $c_1 $ and $c_2$, according to \ref{Proof:var} and  $\E[\lambda_{gi}(t)N_i(t)]/t $ converges to $c_1$ as $t$ increases.
Note that
\begin{align*}
\frac{2}{T^2} \E\left[ \int_{(0,T] \times \mathbb Z^+} N_{i}(u-) N_{gi}(\D u \times \D k) \right] &= \frac{2}{T^2}\int_0^T \E[\lambda_{gi}(t)N_i(t) ] \D t \\
&=  c_1 + \frac{2c_2}{T} \approx  c_1 \approx \frac{1}{2t} \E[\lambda_{gi}(t)N_i(t)]
\end{align*}
with approximations for a large enough $t$ and $T$.
A similar argument was applied to Eq.~\eqref{Eq:kNN}.

The column `H. vol' is for the mean of the volatility estimates computed by the likelihood estimates of $\mu, \alpha_s, \alpha_c, \beta, \eta$ and $K$s using Remark~\ref{Remark:vol}.
This was compared with the two scale realized volatility (TSRV) in the column `TSRV' proposed by \cite{Zhang2005}.
The small time scale was set to be one second and the large time scale was set to be five minutes for the TSRV computation.
The results show that the Hawkes volatility and TSRV are similar.
The standard deviations of the Hawkes volatility are smaller than those of the TSRV for all simulation cases.

\subsection{Other examples}\label{Subsect:full}

This subsection examines the cases where there is a discrepancy between the Hawkes volatility and the realized volatility.
First, the fully characterized Hawkes model is examined, i.e., the coefficients of the branching matrix is represented by 
$$q_{ij} = \frac{\alpha_{ij}}{\beta_{ij}} \E[1+(k_j -1)\eta]$$
with the linear impact function of Assumption~\ref{Assumption}~(ii).
Under this setting, no symmetry is guaranteed.
Recall that in the symmetric model, $\alpha_s = \alpha_{11} = \alpha_{22}$, $\alpha_c = \alpha_{12} = \alpha_{21}$, and $\beta = \beta_{11} = \beta_{12} = \beta_{21} = \beta_{22}$.

Table~\ref{Table:full} lists the estimation results of the fully characterized Hawkes model with simulated paths with the presumed parameters.
The presumed parameters are presented in the `true' column and 500 sample paths are generated over a one day time horizon, more precisely, 5.5 hours as in the previous example.
The columns `full' report the means and standard deviations of the estimates under the maximum likelihood estimation with the fully characterized Hawkes model.
The likelihood estimations were also performed under the symmetric Hawkes model, even though the paths are generated by the fully characterized Hawkes model.
The results are presented in the columns `symmetric' at the centers of the rows of corresponding parameters.
For example, $\mu$ is presented at the center of two rows of $\mu_1$ and $\mu_2$,
$\alpha_s$ is presented at the center of two rows of $\alpha_{11}$ and $\alpha_{22}$, etc..

The `S.Vol.' represents the sample volatility of the return computed by the sample standard deviation of the closing stock prices generated by the 500 sample paths.
The TSRV and marked Hawkes volatility with corresponding standard deviations are reported in column `TSRV' and `H.Vol', respectively.
The Hawkes volatility is calculated using the estimates of the symmetric Hawkes model.
Two volatilities are biased around 4\% compared to the sample volatility.
The TSRV are larger than the sample volatilities and the Hawkes volatilities are smaller in these cases.

\begin{table}
\caption{Fully characterized Hawkes model with 500 sample paths}\label{Table:full}
\centering
\begin{tabular}{ccccccc|ccccc}
\hline
 & & \multicolumn{2}{c}{full} & \multicolumn{2}{c}{symmetric}  & & &\multicolumn{2}{c}{full} & \multicolumn{2}{c}{symmetric} \\
 & true & mean & std. & mean & std. & & true & mean & std. & mean & std.\\
\hline
$\mu_1$ & 0.1461 & 0.1467 & 0.0038 & \multirow{2}{*}{0.1345} & \multirow{2}{*}{0.0026} & 
 & 0.1130 & 0.1131 & 0.0030 & \multirow{2}{*}{0.1152} & \multirow{2}{*}{0.0024}\\
$\mu_2$ & 0.1155 & 0.1159 & 0.0032 & & &
 & 0.1149 & 0.1153 & 0.0033 &\\
$\alpha_{11}$ & 0.3185 & 0.3204 & 0.0150 & \multirow{2}{*}{0.4102} & \multirow{2}{*}{0.0148} &
 & 0.4994 & 0.5031 & 0.0242 & \multirow{2}{*}{0.5252} & \multirow{2}{*}{0.0159}\\
$\alpha_{22}$ & 0.3821 & 0.3865 & 0.0219 & & &
 & 0.4682 & 0.4799 & 0.0210 &\\
$\alpha_{12}$ & 0.9812 & 0.9848 & 0.0282  & \multirow{2}{*}{1.1512} & \multirow{2}{*}{0.0223} &
 & 0.5937 & 0.5992 & 0.0232 & \multirow{2}{*}{0.7012} & \multirow{2}{*}{0.0199}\\
$\alpha_{21}$ & 1.4949 & 1.5000 & 0.0334 & & & 
 & 0.9754 & 0.9854 & 0.0368  &\\
$\beta_{11}$ & 1.1799 & 1.1893 & 0.0567 & \multirow{4}{*}{2.0547} & \multirow{4}{*}{0.0315} & 
 & 1.8305 & 1.8512 & 0.0948 & \multirow{4}{*}{1.8744} & \multirow{4}{*}{0.0364} \\
$\beta_{22}$ & 1.9553 & 1.9840 & 0.1195 & & & 
 & 1.4706 & 1.5142 & 0.0666 &\\
$\beta_{12}$ & 2.0952 & 2.1077 & 0.0697 & & & 
 & 1.5963 & 1.6110 & 0.0624 &\\
$\beta_{21}$ & 2.5030 & 2.5132 & 0.0587 & & &
 & 2.7850 & 2.8064 & 0.1036 &\\
$\eta$ & 0.1488 & 0.1501 & 0.0235 & 0.1424 & 0.0255 &
 & 0.1761 & 0.1768 & 0.0216 & 0.1756 & 0.0225\\
\hline
&  \multicolumn{5}{c}{ $c=0.1, d= 1.0, U =2.0$} & & \multicolumn{5}{c}{ $c=0.08, d= 1.5, U = 3.0$} \\
\hline
 & S.Vol. & TSRV & std. & H.Vol. & std.  &  & S.Vol. & TSRV & std. & H.Vol. & std.  \\
 & 0.1405 & 0.1463 & 0.0146 & 0.1346 & 0.0051 & & 0.1853 & 0.1897 & 0.0161 & 0.1795 & 0.0044\\
\hline
\end{tabular}
\end{table}

Second, the symmetric marked Hawkes models were examined, where the model parameters change during the sample period.
Table~\ref{Table:timevarying} lists the estimation results with the symmetric Hawkes models of the 5.5 hour's time horizon but the model parameter of the first one hour of the period is according to the row of `True 1' and in the rest of the period, the model follows `True 2'.
In the first panel, the varying part is the upper bound of the conditional mean of the mark distribution.
In other words, during the first part of the sample period, the price process is quite volatile due to the possible large size of the mark, and the remaining part is rather stable.
This mimics the case of the 2010 Flash Crash and the empirical analysis will be performed later.
The result shows the discrepancy between the TSRV and the Hawkes volatility which are both less than the sample volatility and the TSRV is even less than the Hawkes volatility.

\begin{table}
\caption{Simulation study for the marked Hawkes model with 500 sample paths with time varying parameters}\label{Table:timevarying}
\centering
\begin{tabular}{cccccccccccc}
\hline
 & $\mu$ & $\alpha_s$ & $\alpha_c$ & $\beta$ & $\eta$ & $c$ & $d$ & $U$ & S.Vol. & TSRV & H.Vol.  \\
\hline
True 1 & 0.1000 & 1.1000 & 1.2600 & 2.5700 & 0.0100 & 0.2500 & 1 & 7   \\
True 2 & 0.1000 & 1.1000 & 1.2600 & 2.5700 & 0.0100 & 0.2500 & 1 & 1.5 \\
mean & 0.1017 & 1.1017 & 1.2662 & 2.5667 & 0.0085 &  & &  & 0.6431  & 0.5801 & 0.6288  \\
std. & 0.0022 & 0.0205 & 0.0265 & 0.0319 & 0.0039 &  & & & & 0.1990 & 0.1381\\ 
\hline
True 1 & 0.1000 & 1.1000 & 1.2600 & 2.5700 & 0.1000 & 0.1000 & 1 & 7 \\
True 2 & 0.0500 & 0.5000 & 0.5000 & 2.0000 & 0.1000 & 0.1000 & 1 & 1.5 \\
mean & 0.0375 & 1.0162 & 1.1190 & 2.3567 & 0.0233 & & & & 0.2496 & 0.2204 & 0.2300 \\
std. & 0.0010 & 0.0302 & 0.0347 & 0.0438 & 0.0142 & & & & & 0.0378 & 0.0205 \\
\hline
\end{tabular}
\end{table}

\section{Empirical study}\label{Sect:empirical}
\subsection{Data}

The empirical studies used the ultra high-frequency tick-by-tick data of some major stock prices consisting of several years with the best bid and ask quotes reported in the New York Stock Exchange (NYSE).
The time horizon of the sample for each day is set to be from 10:00 to 15:30.
The data of 30 minutes immediately after the opening and before the closing time were not used to reduce the seasonality effects.
The price movement patterns are usually different at the near opening and closing from the rest of the day.

The jump sizes of the price movements of equities in the S\&P 500 are not constant over time particularly when the price of the equity is high and hence the ratio between the price and the minimum tick size in transaction on the NYSE, \$0.01, is high.
The tick size of the NYSE was reduced from \$1/8 to \$1/16 in 1997 and from \$1/16 to \$0.01 in 2001.
In this paper, the mid-price movements is considered for the marked Hawkes modeling to remove the bid-ask bounce and hence the minimum jump size is the half tick size, \$0.005.

In the original data, the time resolution of the record is one second. 
If more than one timestamps of the price changes are reported for one second, 
then the reported events are distributed over a one second interval to equidistant finer partitions. 

\subsection{Unconditional distribution of mark}

Table~\ref{Table:tick} compares the percentage of the mark size of IBM, GE and CVX from 2008 to 2011, i.e., the unconditional distribution of mark sizes are reported in the table.
IBM and CVX have a range of mark sizes over the years but GE's mark size distributions concentrates on the minimum mark size.
This is because the price of IBM and CVX is relatively high (IBM is around \$150 and CVX is around \$100), whereas the price of GE is around \$25.
The unconditional distributions of the marks have exponentially decreasing shapes which are similar to the geometric distributions.
The empirical distribution of the marks of IBM, 2010 and 2011, was compared with the geometric distributions in Figure~\ref{Fig:mark}.
The solid lines are for the empirical distribution and the dashed lines are for the geometric distribution fitted by matching the first moments of the empirical and geometric distributions.

\begin{table}
\caption{Mark size distribution (\%) of IBM (left) and GE (center) and CVX (right) from 2008 to 2011}\label{Table:tick}
\centering
\begin{tabular}{ccccc|cccc|cccc}
\hline
mark & 2008 & 2009 & 2010 & 2011 & 2008 & 2009 & 2010 & 2011 & 2008 & 2009 & 2010 & 2011\\
\hline
1 & 51.80 & 59.98 & 80.88 & 57.04 & 89.81 & 98.39 & 99.61 & 99.68 & 60.69 & 68.30 & 91.55 & 77.21\\
2 & 21.53 & 20.57 & 13.96 & 18.41 & 7.61 & 1.51 & 0.34 & 0.30 & 20.02 & 19.75 & 7.32 & 16.20\\
3 & 11.22 & 10.89 & 3.53 & 9.54 & 1.55 & 0.00 & 0.00 & 0.00 & 9.25 & 8.43 & 0.77 & 4.77\\
4 & 6.36 & 5.50 & 1.01 & 5.80 & 0.52 & 0.00 & 0.00 & 0.00 & 4.78 & 2.70 & 0.15 & 1.28\\
5 & 3.61 & 1.98 & 0.35 & 3.68 & 0.19 & 0.00 & 0.00 & 0.00 & 2.49 & 0.62 & 0.06 & 0.32\\
6 & 2.04 & 0.66 & 0.10 & 2.22 & 0.00 & 0.00 & 0.00 & 0.00 & 1.29 & 0.13 & 0.03 & 0.12\\
7 & 1.17 & 0.24 & 0.00 & 1.37 & 0.00 & 0.00 & 0.00 & 0.00 & 0.66 & 0.04 & 0.02 & 0.05\\
8 & 0.71 & 0.01 & 0.00 & 0.81 & 0.00 & 0.00 & 0.00 & 0.00 & 0.35 & 0.01 & 0.01 & 0.02\\
\hline
\end{tabular}
\end{table}

\begin{figure}
	\centering
	\begin{subfigure}[b]{0.45\textwidth}
		\includegraphics[width=\textwidth]{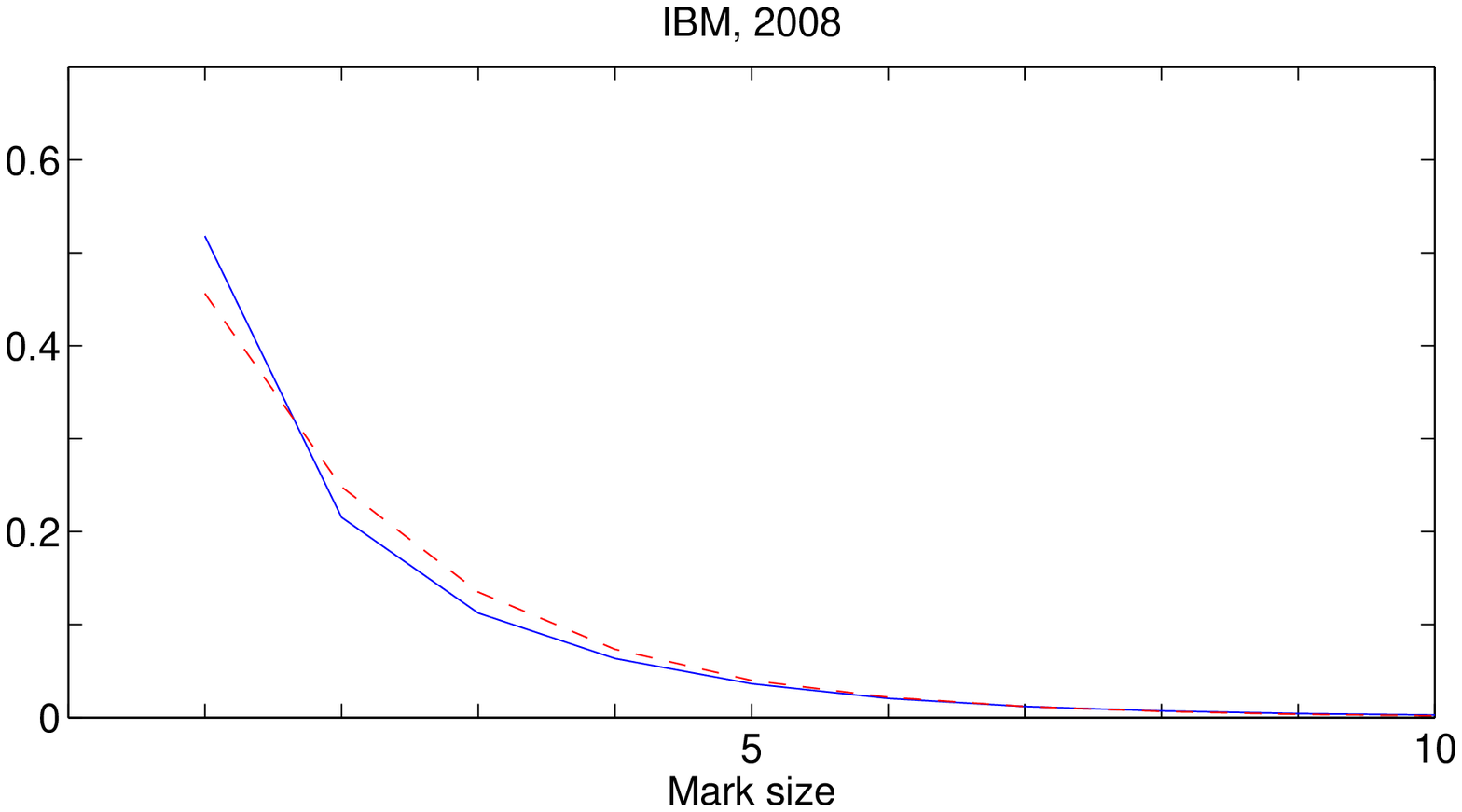}
		\caption{IBM, 2008}
	\end{subfigure}
	\begin{subfigure}[b]{0.45\textwidth}
		\includegraphics[width=\textwidth]{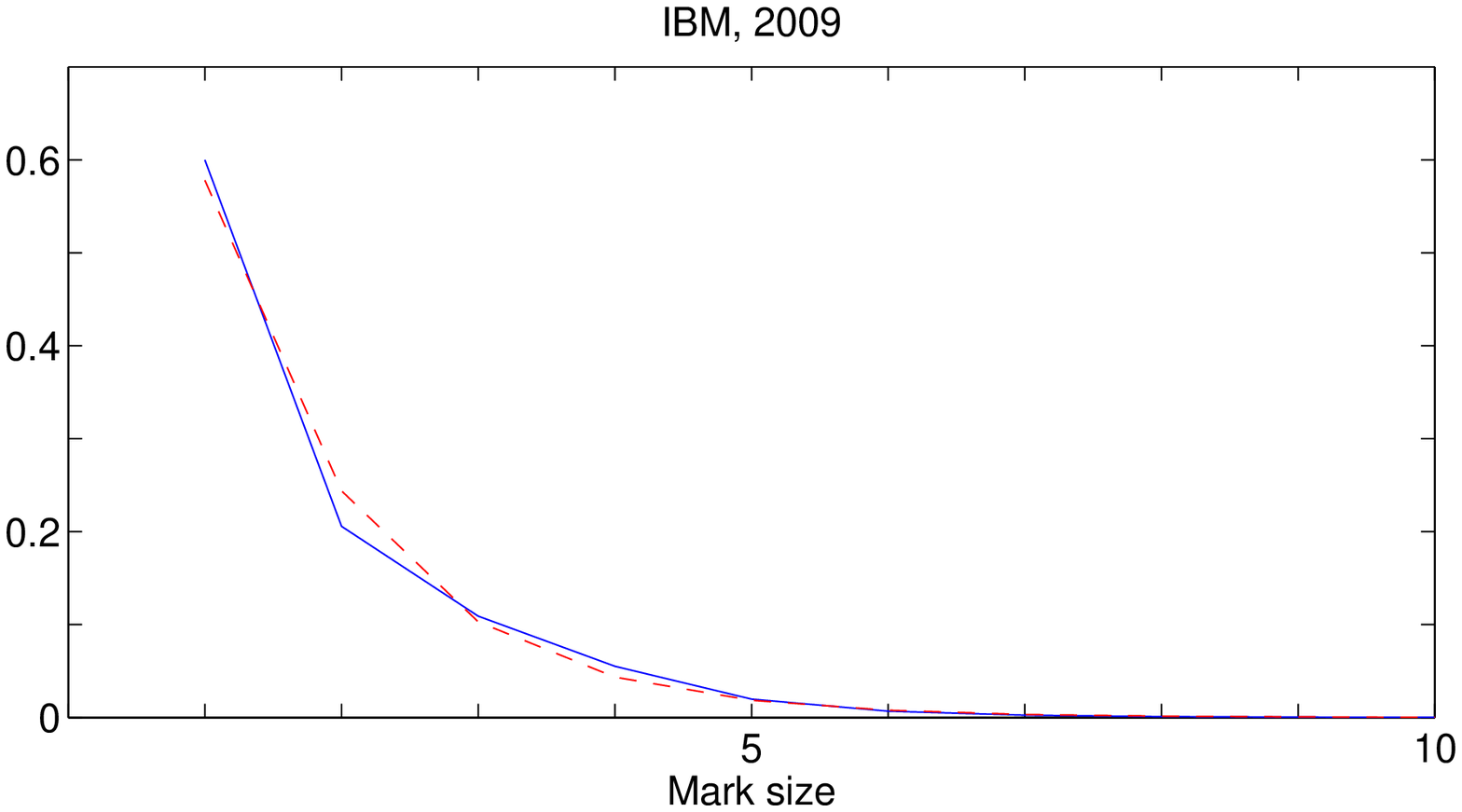}
		\caption{IBM, 2009}
	\end{subfigure}
	\begin{subfigure}[b]{0.45\textwidth}
		\includegraphics[width=\textwidth]{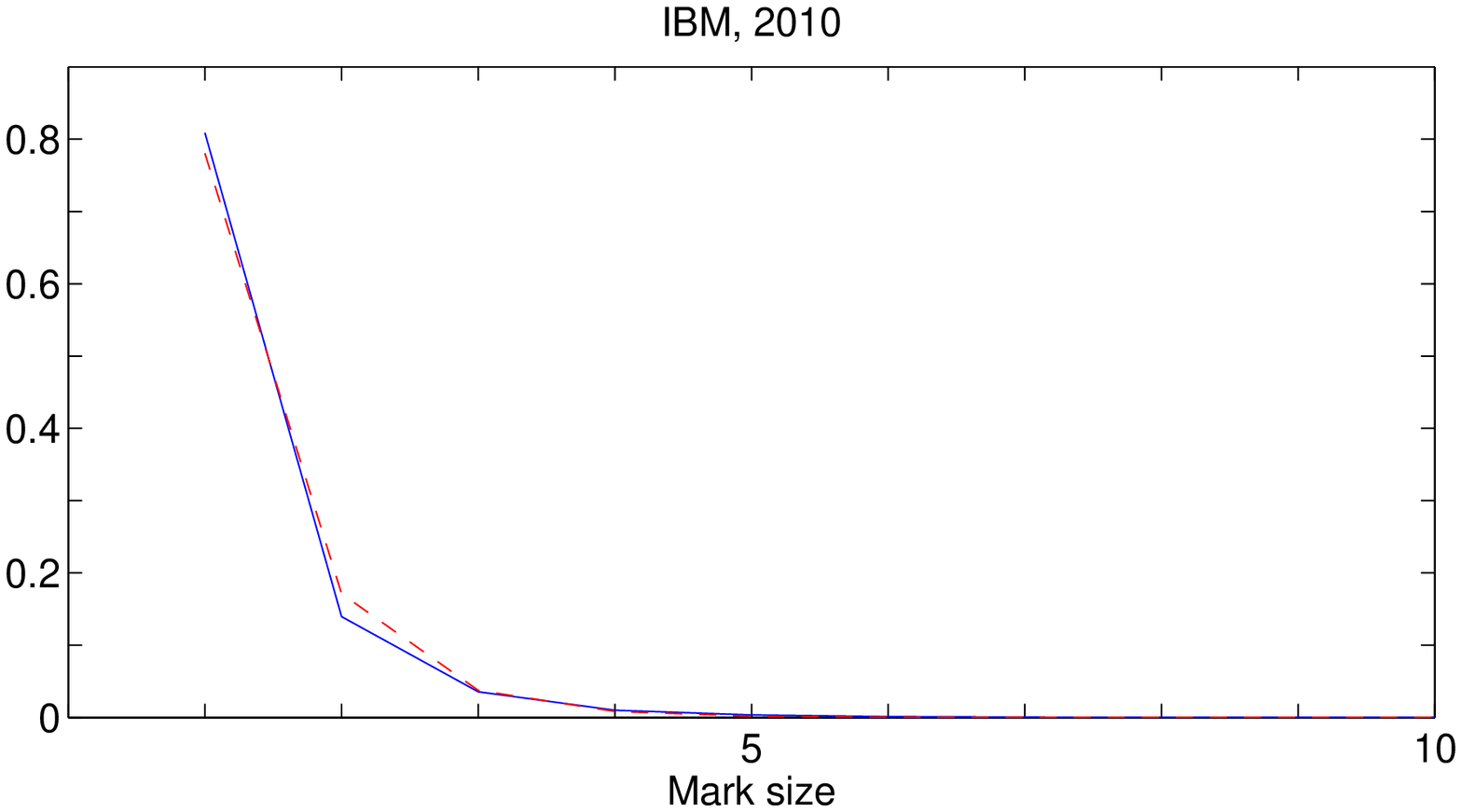}
		\caption{IBM, 2010}
	\end{subfigure}
	\begin{subfigure}[b]{0.45\textwidth}
		\includegraphics[width=\textwidth]{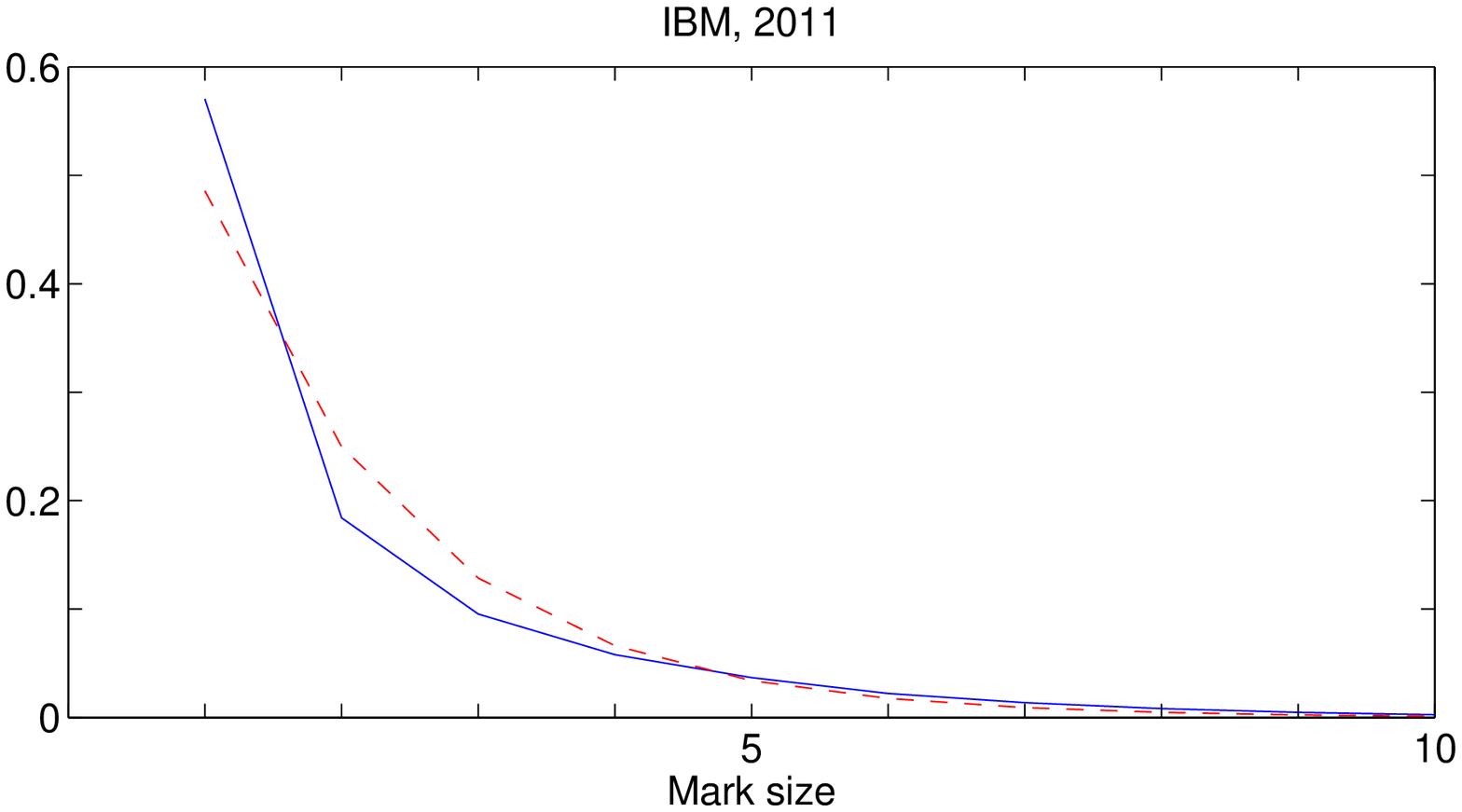}
		\caption{IBM, 2011}
	\end{subfigure}
	\caption{Empirical unconditional distribution of mark}\label{Fig:mark}
\end{figure}

\subsection{Mark size and intensity}

This subsection examines the dependence between the mark size and the ground intensity, and the number of expected events over unit interval.
The empirical evidence shows that the mark size and the current ground intensity are significantly related to each other.
First, the empirical conditional expectation of the intensities with given mark size, $\E[\lambda_{gi}(t) | k_i ]$, were calculated.
The proxy intensities are introduced because the ground intensities are unobservable.
The proxy intensities for the up, down and total jumps are defined by the numbers of up, down and total jumps, respectively, over a fixed time period, which is ended just before the time of the jump, divided by the length of the period.
The period for the proxy intensities was chosen as ten seconds.
Mathematically, the up proxy intensity for the mark $k_1$, which takes place at time $t$ is represented by
$ N_{g1}(t+\tau) - N_{g1}(t) $,
where $\tau$ is the length of the period.

Table~\ref{Table:mark_size} presents the calculated sample mean and standard error of the proxy intensities for each mark size and for each year of 2010 and 2011, IBM.
For example, for the mark size 6, there are 86,738 up jumps reported and the sample mean of the up proxy intensity is 3.4784 and the sample standard error is 0.0145.
Note that the intensity of 3.4784 implies that the expected number of movements over unit time, which was set to 1 second, is approximately 3.4784.
The table shows that the proxy intensities increase with increasing given mark size.
The negative integers in the column of the mark size represent the down jump of the price.
The proxy intensities were also calculated for five seconds time horizon.
The results are similar to the previous case of 10 seconds time horizon so the results are not shown.

\begin{table}
\caption{Relationship between the mark size and the mean of proxy intensity (10 seconds), IBM}\label{Table:mark_size}
\centering
\begin{tabular}{c|ccc|ccccccc}
\hline
&  & 2011 & & & 2010 & &\\
mark size & up & down & total & up & down & total\\
\hline
6 & 3.4784 & 3.6400 & 7.1184 & 6.4059 & 6.5618 & 12.9676\\
 & (0.0145) & (0.0149) & (0.0292) & (0.2306) & (0.1171) & (0.1150)\\
5 & 3.1191 & 3.2461 & 6.3652 & 4.3289 & 4.5448 & 8.8738\\
 & (0.0105) & (0.0106) & (0.0209) & (0.0890) & (0.0449) & (0.0448)\\
4 & 2.8206 & 2.9237 & 5.7443 & 3.4963 & 3.6881 & 7.1844\\
 & (0.0078) & (0.0078) & (0.0155) & (0.0439) & (0.0221) & (0.0221)\\
3 & 2.5683 & 2.6551 & 5.2233 & 2.7119 & 2.8717 & 5.5836\\
 & (0.0061) & (0.0060) & (0.0120) & (0.0182) & (0.0092) & (0.0092)\\
2 & 2.3051 & 2.3795 & 4.6846 & 2.1799 & 2.2892 & 4.4691\\
 & (0.0042) & (0.0041) & (0.0083) & (0.0071) & (0.0036) & (0.0036)\\
1 & 2.3131 & 2.3501 & 4.6632 & 1.7682 & 1.8481 & 3.6163\\
 & (0.0028) & (0.0028) & (0.0056) & (0.0024) & (0.0012) & (0.0012)\\
$-1$ & 2.2403 & 2.4140 & 4.6543 & 1.7488 & 1.8883 & 3.6371\\
 & (0.0028) & (0.0028) & (0.0056) & (0.0024) & (0.0012) & (0.0012)\\
$-2$ & 2.3004 & 2.4376 & 4.7381 & 2.1655 & 2.3022 & 4.4687\\
 & (0.0042) & (0.0042) & (0.0084) & (0.0072) & (0.0036) & (0.0036)\\
$-3$ & 2.5924 & 2.7167 & 5.3090 & 2.7264 & 2.8552 & 5.5816\\
 & (0.0062) & (0.0062) & (0.0123) & (0.0185) & (0.0094) & (0.0093)\\
$-4$ & 2.8174 & 2.9552 & 5.7726 & 3.4806 & 3.6163 & 7.0969\\
 & (0.0078) & (0.0079) & (0.0156) & (0.0432) & (0.0218) & (0.0217)\\
$-5$ & 3.1078 & 3.2399 & 6.3477 & 4.4487 & 4.5968 & 9.0454\\
 & (0.0104) & (0.0106) & (0.0209) & (0.0902) & (0.0466) & (0.0469)\\
$-6$ & 3.4606 & 3.5845 & 7.0451 & 6.1019 & 6.2334 & 12.3352\\
 & (0.0145) & (0.0149) & (0.0292) & (0.2141) & (0.1083) & (0.1074)\\
\hline
\end{tabular}
\end{table}

Second, Table~\ref{Table:mark_size2} presents the relationship between the mark sizes and the inferred ground intensities with the linear impact function using the IBM tick data.
Prior to calculating the inferred ground process, the parameters $\omega, \alpha_s, \alpha_c, \beta,$ and $\eta$ were estimated by maximizing $\log L_g$ defined in Eq.~\eqref{Eq:likelihood}.
The estimations were performed on a daily basis and the detailed estimation results will be demonstrated later.
Subsequently, the inferred ground intensities were computed with the estimates of $\omega, \alpha_s, \alpha_c, \beta$, and $\eta$ using the definition of the ground intensities in Eqs.~\eqref{Eq:lambdag1}~and~\eqref{Eq:lambdag2}.
The sample mean and sample standard errors of the inferred ground intensities for each mark size is reported.
Similarly with the case of the proxy intensities, the inferred ground intensities increase with increasing given mark size.
This implies that if a large size of the mark is observed, it is probably based on the large ground intensities.

\begin{table}
\caption{Relationship between the mark size and the mean of inferred ground intensity with the linear impact function, IBM}\label{Table:mark_size2}
\centering
\begin{tabular}{c|ccc|ccccccc}
\hline
&  & 2011 & & & 2010 & &\\
mark size & $\lambda_{g1}$ & $\lambda_{g2}$ & $\lambda_g$ & $\lambda_{g1}$ & $\lambda_{g2}$ & $\lambda_g$\\
\hline
6 & 5.1666 & 5.1108 & 10.2774 & 8.1033 & 8.0330 & 16.1364\\
 & (0.0222) & (0.0221) & (0.0442) & (0.1366) & (0.1370) & (0.2734)\\
5 & 4.7423 & 4.6834 & 9.4257 & 6.5586 & 6.4855 & 13.0443\\
 & (0.0165) & (0.0164) & (0.0328) & (0.0621) & (0.0620) & (0.1240)\\
4 & 4.4092 & 4.3565 & 8.7656 & 5.4497 & 5.3975 & 10.8472\\
 & (0.0126) & (0.0125) & (0.0251) & (0.0338) & (0.0337) & (0.0674)\\
3 & 4.0264 & 3.9842 & 8.0106 & 4.2182 & 4.1806 & 8.3988\\
 & (0.0094) & (0.0094) & (0.0188) & (0.0145) & (0.0145) & (0.0290)\\
2 & 3.5870 & 3.5505 & 7.1375 & 3.4705 & 3.4476 & 6.9181\\
 & (0.0066) & (0.0066) & (0.0131)& (0.0064) & (0.0064) & (0.0128)\\
1 & 3.5924 & 3.5580 & 7.1504 & 2.5428 & 2.5321 & 5.0749\\
 & (0.0057) & (0.0057) & (0.0114) & (0.0021) & (0.0021) & (0.0041)\\
$-1$ & 3.5568 & 3.5942 & 7.1509 & 2.5702 & 2.5836 & 5.1538\\
 & (0.0057) & (0.0057) & (0.0113) & (0.0021) & (0.0021) & (0.0041)\\
$-2$ & 3.5858 & 3.6270 & 7.2128 & 3.4730 & 3.4978 & 6.9708\\
 & (0.0066) & (0.0066) & (0.0132) & (0.0064) & (0.0064) & (0.0129)\\
$-3$ & 4.0451 & 4.0927 & 8.1378 & 4.2369 & 4.2781 & 8.5150\\
 & (0.0096) & (0.0096) & (0.0192) & (0.0148) & (0.0148) & (0.0296)\\
$-4$ & 4.3580 & 4.4177 & 8.7757 & 5.3775 & 5.4380 & 10.8154\\
 & (0.0126)& (0.0127) & (0.0253) & (0.0333) & (0.0334) & (0.0667)\\
$-5$ & 4.6872 & 4.7587 & 9.4459 & 6.4356 & 6.5047 & 12.9402\\
 & (0.0163) & (0.0164) & (0.0327) & (0.0624) & (0.0626) & (0.1250)\\
$-6$ & 5.0360 & 5.0884 & 10.1244 & 7.9426 & 8.0193 & 15.9620 \\
 & (0.0216) & (0.0217) & (0.0433) & (0.1343) & (0.1344) & (0.2685)\\
\hline
\end{tabular}
\end{table}

Third, Figure~\ref{Fig:CEk} illustrates the empirical expectations of the mark size conditionally upon given inferred ground intensities, $\E[k_i |\lambda_{gi}(t)]$, using the tick data of IBM, 2008-2011.
For each year, the empirical conditional expectation with a given $\lambda_{gi} = n$ for an integer $n$ was computed using the sample mean of the mark sizes, whose associated inferred ground process is falling into $(n-1, n]$.
The conditional expectations of the marks were plotted where the total observed numbers of the mark were larger than 100 for each year, i.e., the samples with a small number of observations are dropped out.
In the figure, the ground intensities vary by more as the years pass, suggesting that the overall number of activities increases.
Most of intensities were less than 15 in 2008 but the inferred intensity in 2011 were more widely distributed as one can see that the large portion of the observed intensities were larger than 15. 

The shape of the conditional expectation of the mark changes over time.
The conditional expectation tends to increase with increasing intensity in 2008 and 2010.
In 2009 and 2011, the conditional expectations showed humped shapes. 
Figure~\ref{Fig:CEk_M} shows the empirical conditional expectations of the mark given ground intensity computed monthly basis from January to June, 2011, of IBM.
In the monthly basis empirical conditional expectation, irregular patterns were observed over time.
The changing shape of the conditional distribution of marks over time is the reason why the mark distribution is not specified, and the estimation was performed in a non-parametric manner for the part of the mark distribution. 

\begin{figure}
	\centering
	\begin{subfigure}[b]{0.45\textwidth}
		\includegraphics[width=\textwidth]{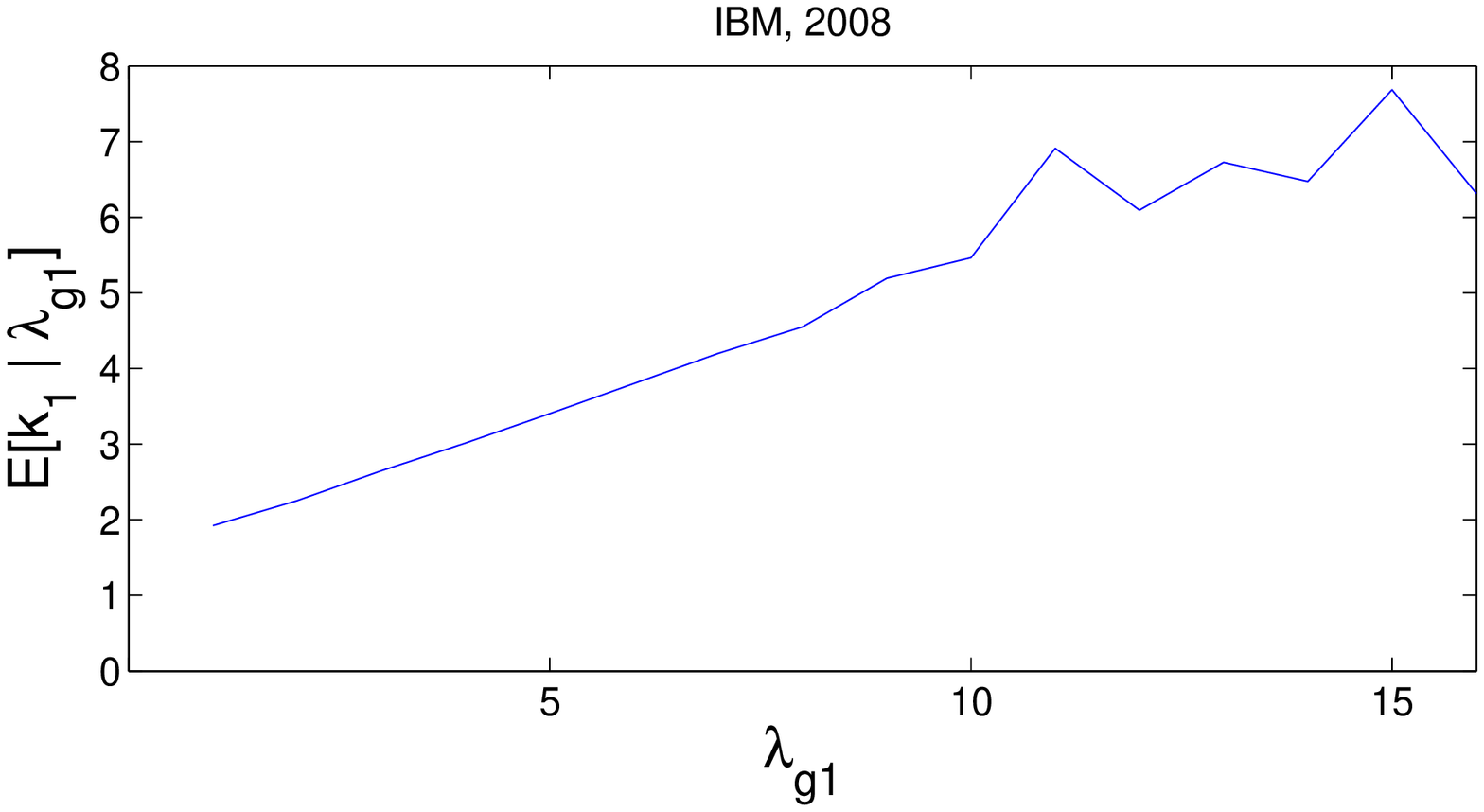}
		\caption{IBM, 2008}\label{CE_IBM2008}
	\end{subfigure}
	\begin{subfigure}[b]{0.45\textwidth}
		\includegraphics[width=\textwidth]{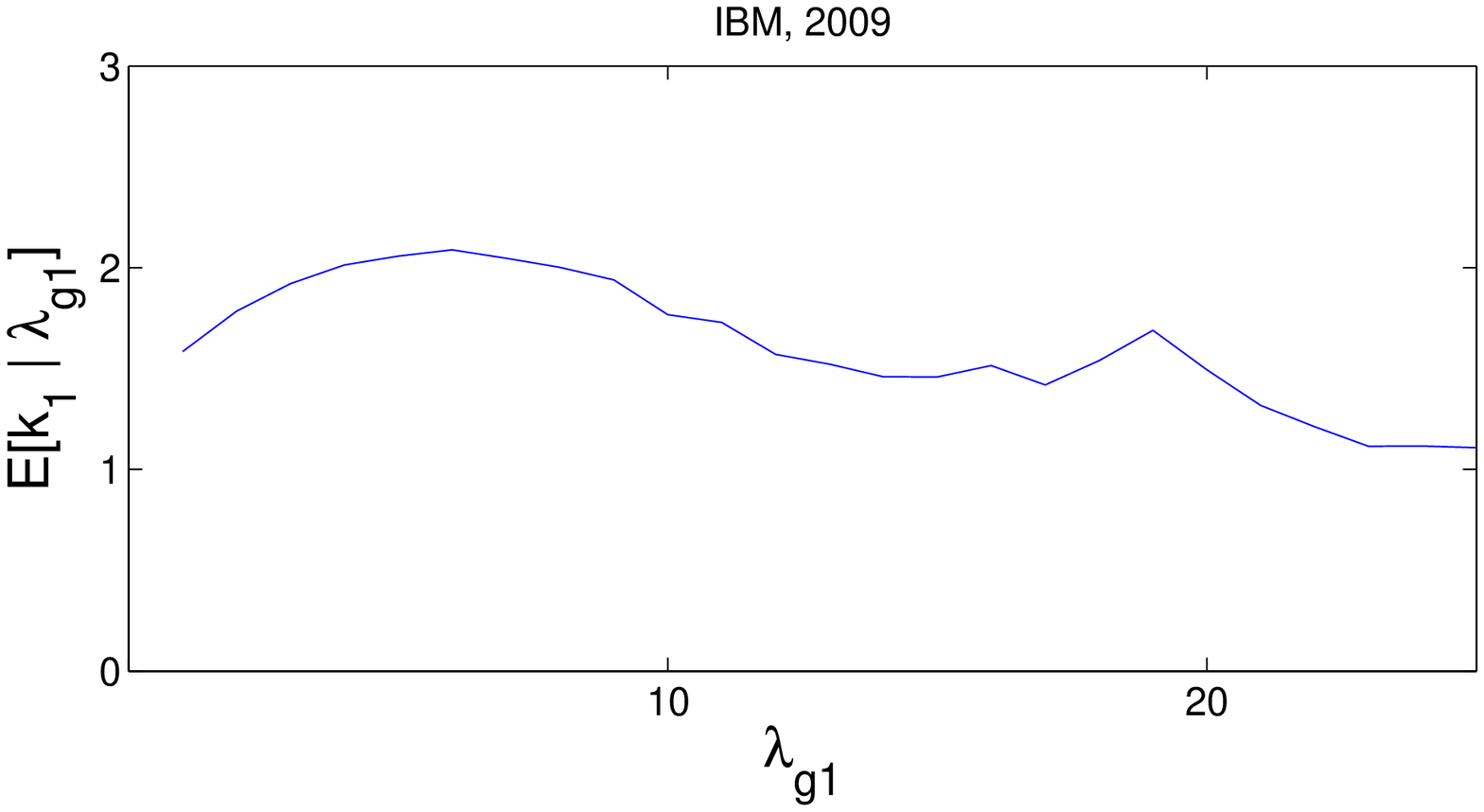}
		\caption{IBM, 2009}
	\end{subfigure}
	\begin{subfigure}[b]{0.45\textwidth}
		\includegraphics[width=\textwidth]{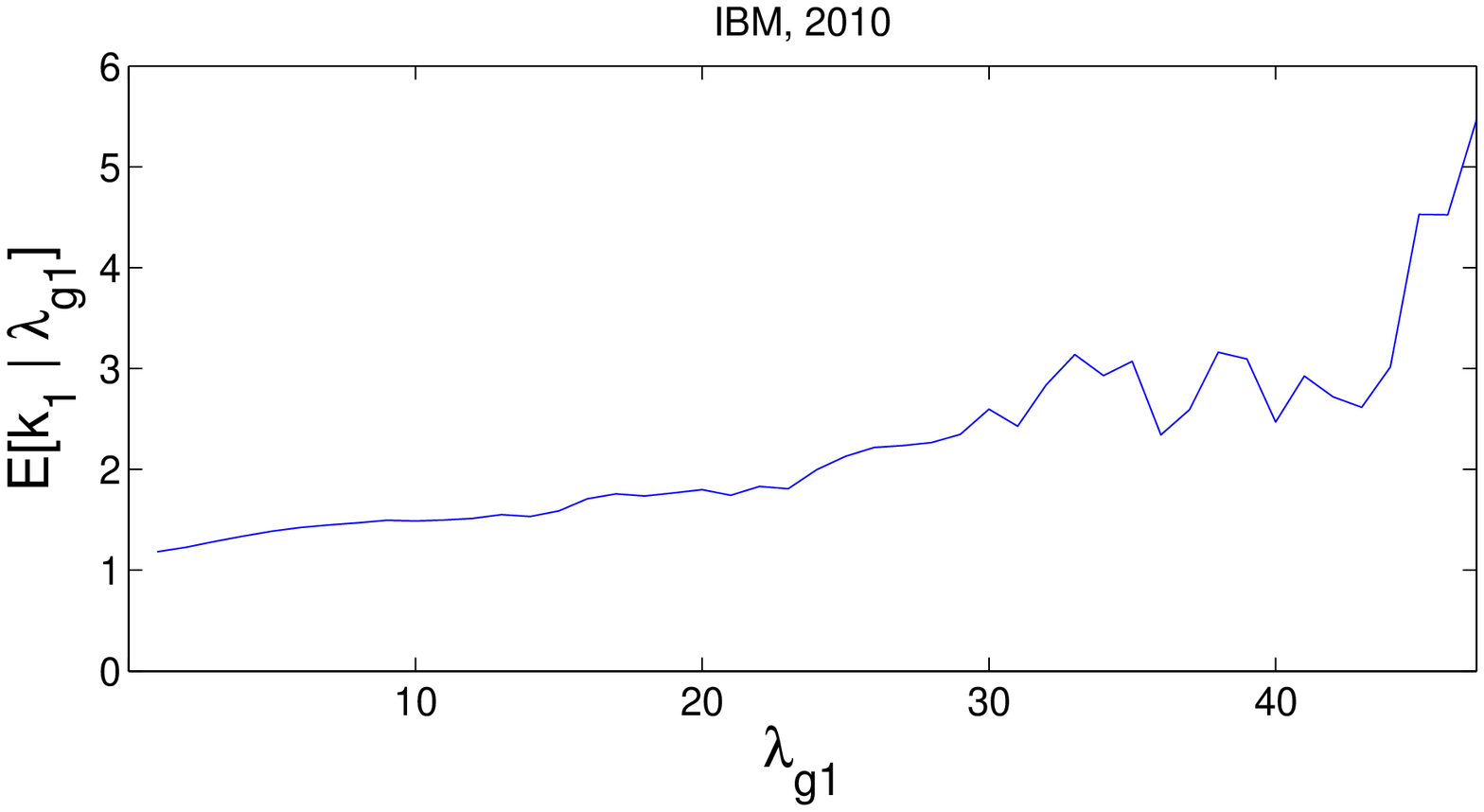}
		\caption{IBM, 2010}
	\end{subfigure}
	\begin{subfigure}[b]{0.45\textwidth}
		\includegraphics[width=\textwidth]{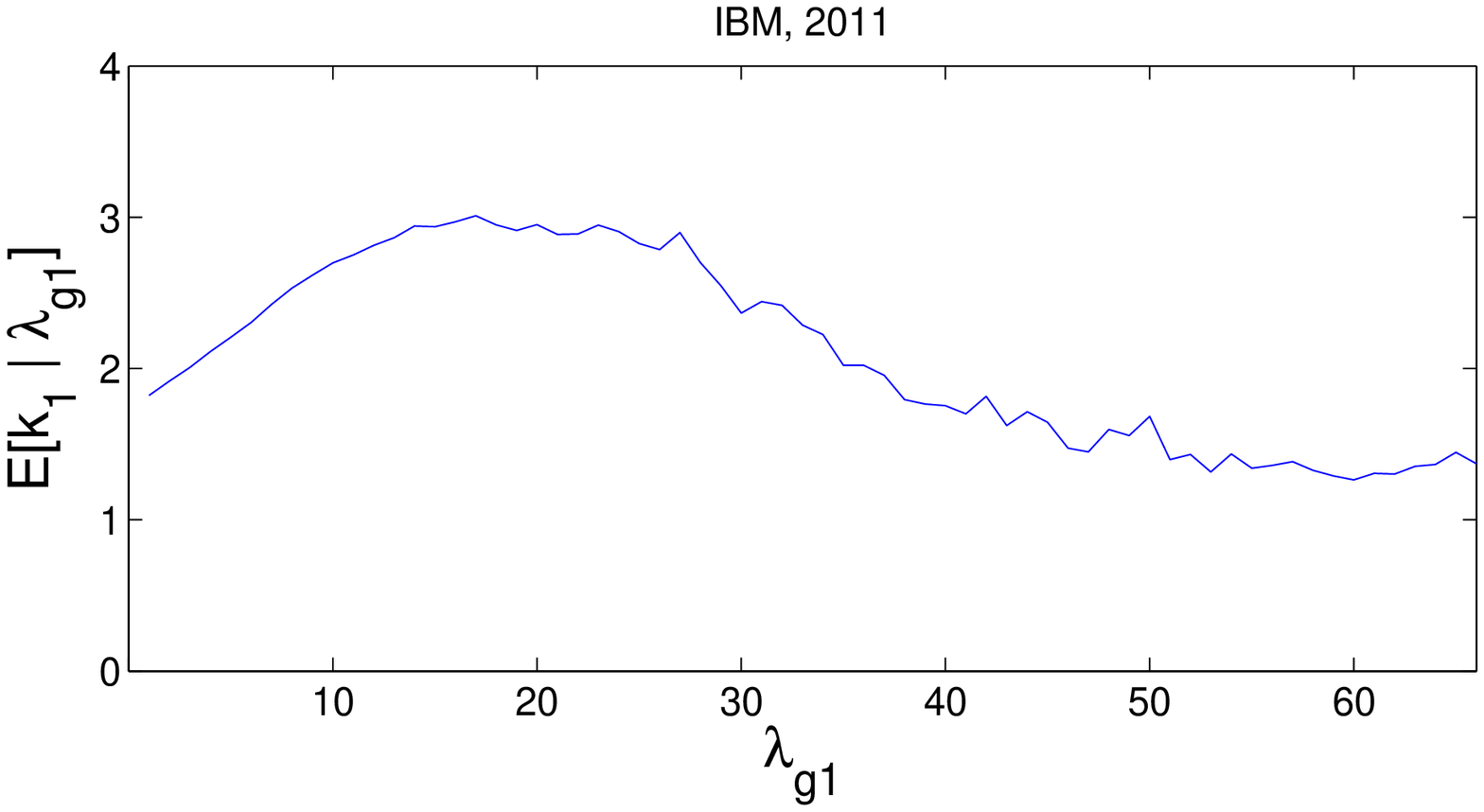}
		\caption{IBM, 2011}
	\end{subfigure}
	\caption{Conditional expectation of $k_1$ on $\lambda_{g1}$, IBM, 2008-2011}\label{Fig:CEk}
\end{figure}

\begin{figure}
	\centering
	\begin{subfigure}[b]{0.45\textwidth}
		\includegraphics[width=\textwidth]{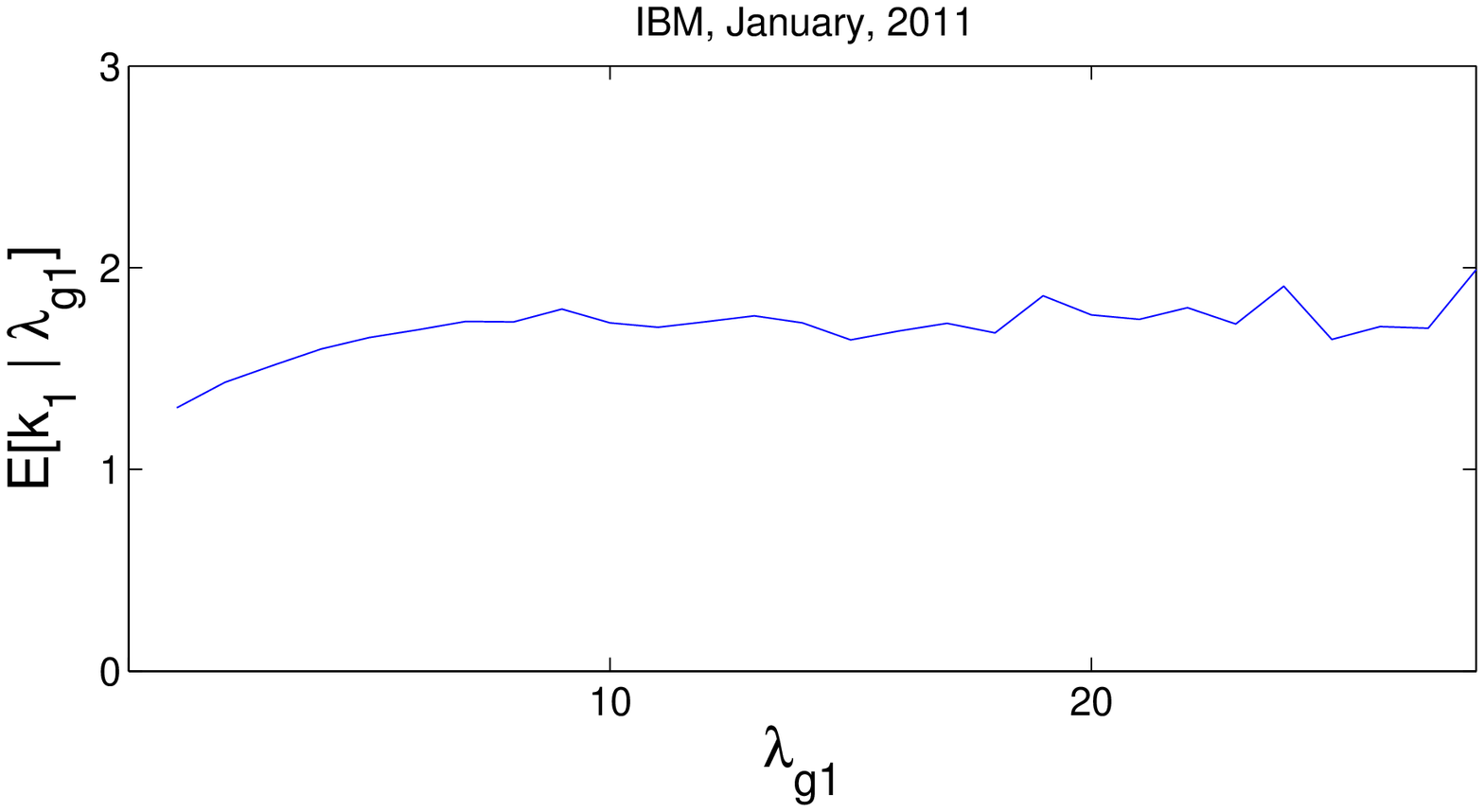}
	\end{subfigure}
	\begin{subfigure}[b]{0.45\textwidth}
		\includegraphics[width=\textwidth]{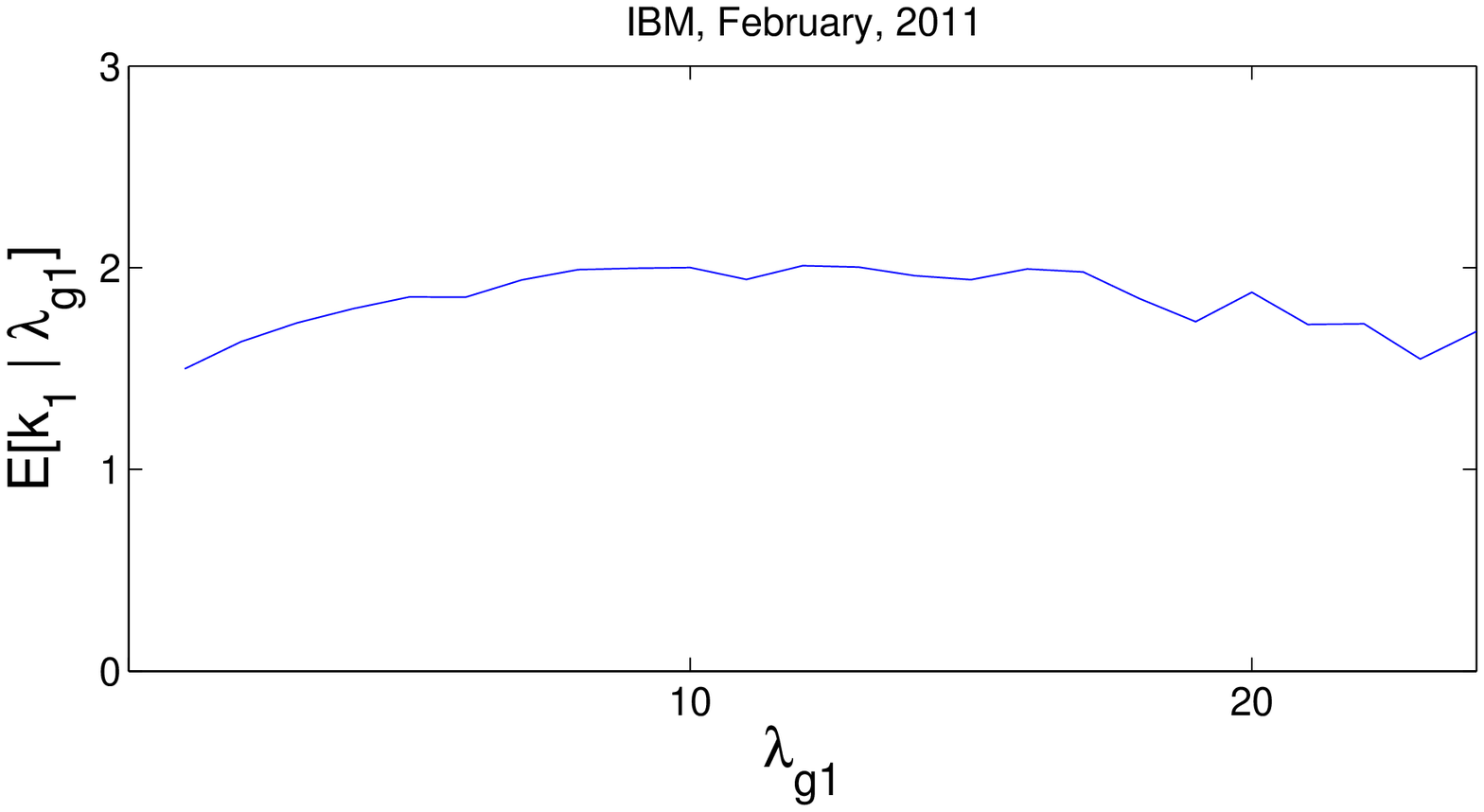}
	\end{subfigure}
	\begin{subfigure}[b]{0.45\textwidth}
		\includegraphics[width=\textwidth]{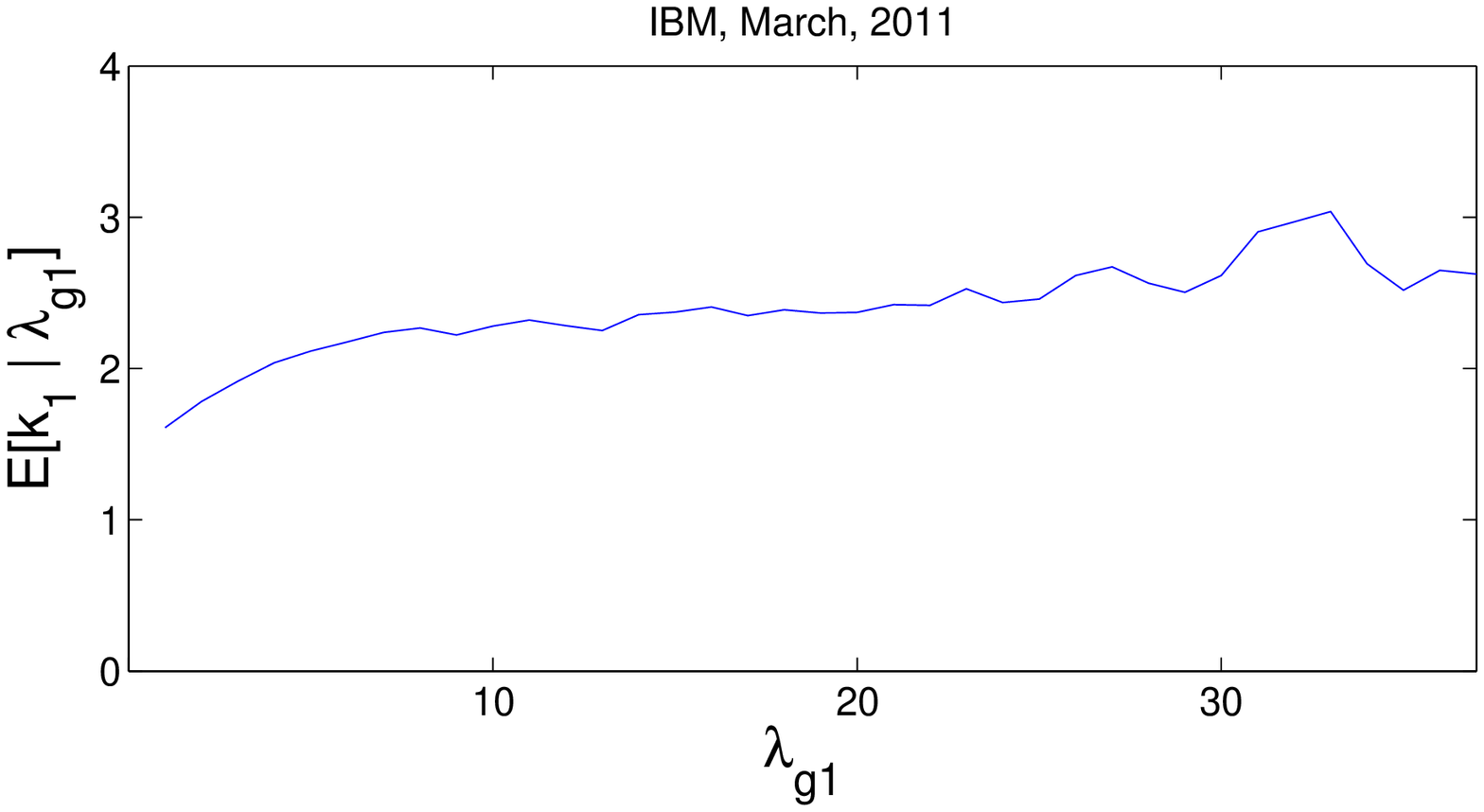}
	\end{subfigure}
	\begin{subfigure}[b]{0.45\textwidth}
		\includegraphics[width=\textwidth]{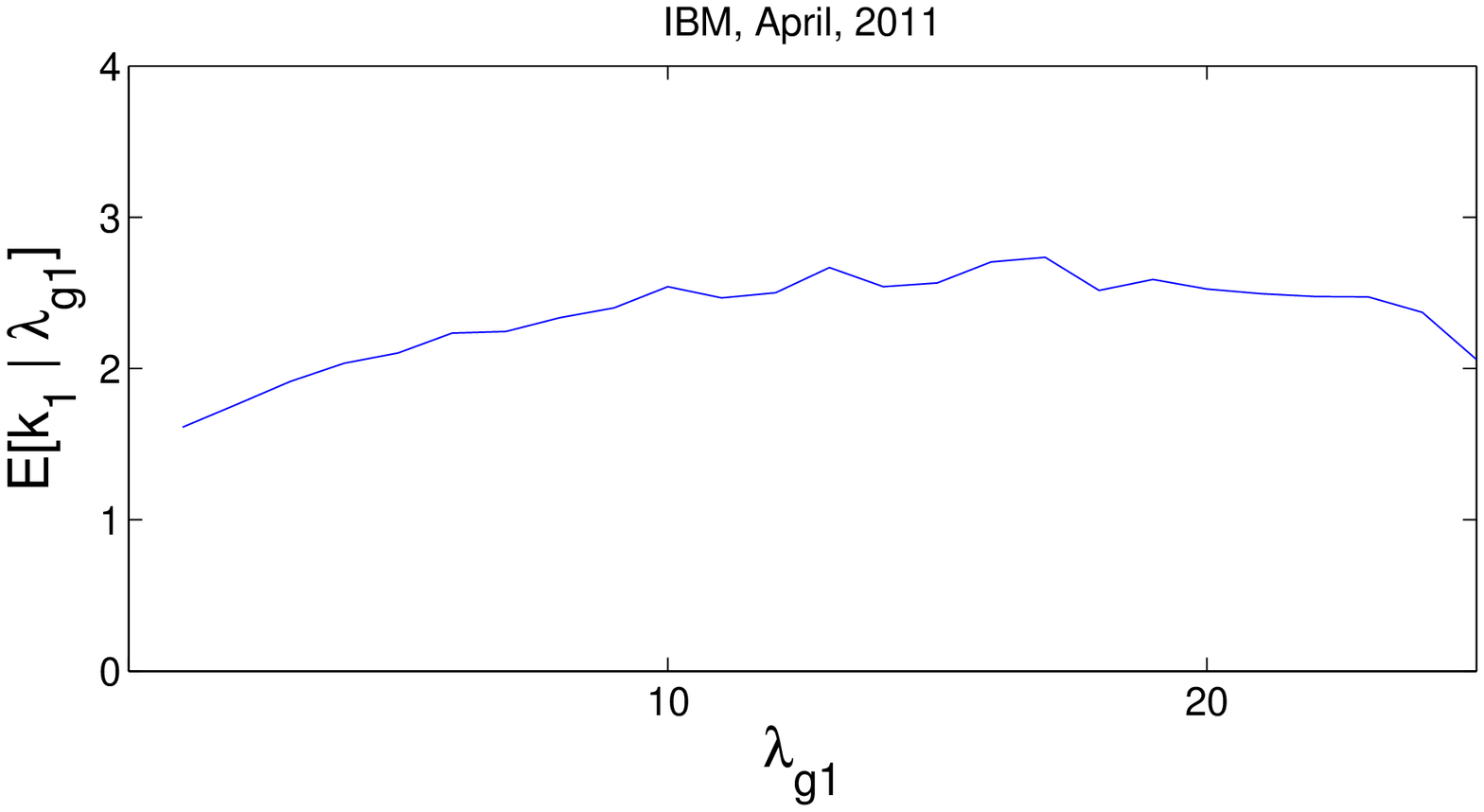}
	\end{subfigure}
	\begin{subfigure}[b]{0.45\textwidth}
		\includegraphics[width=\textwidth]{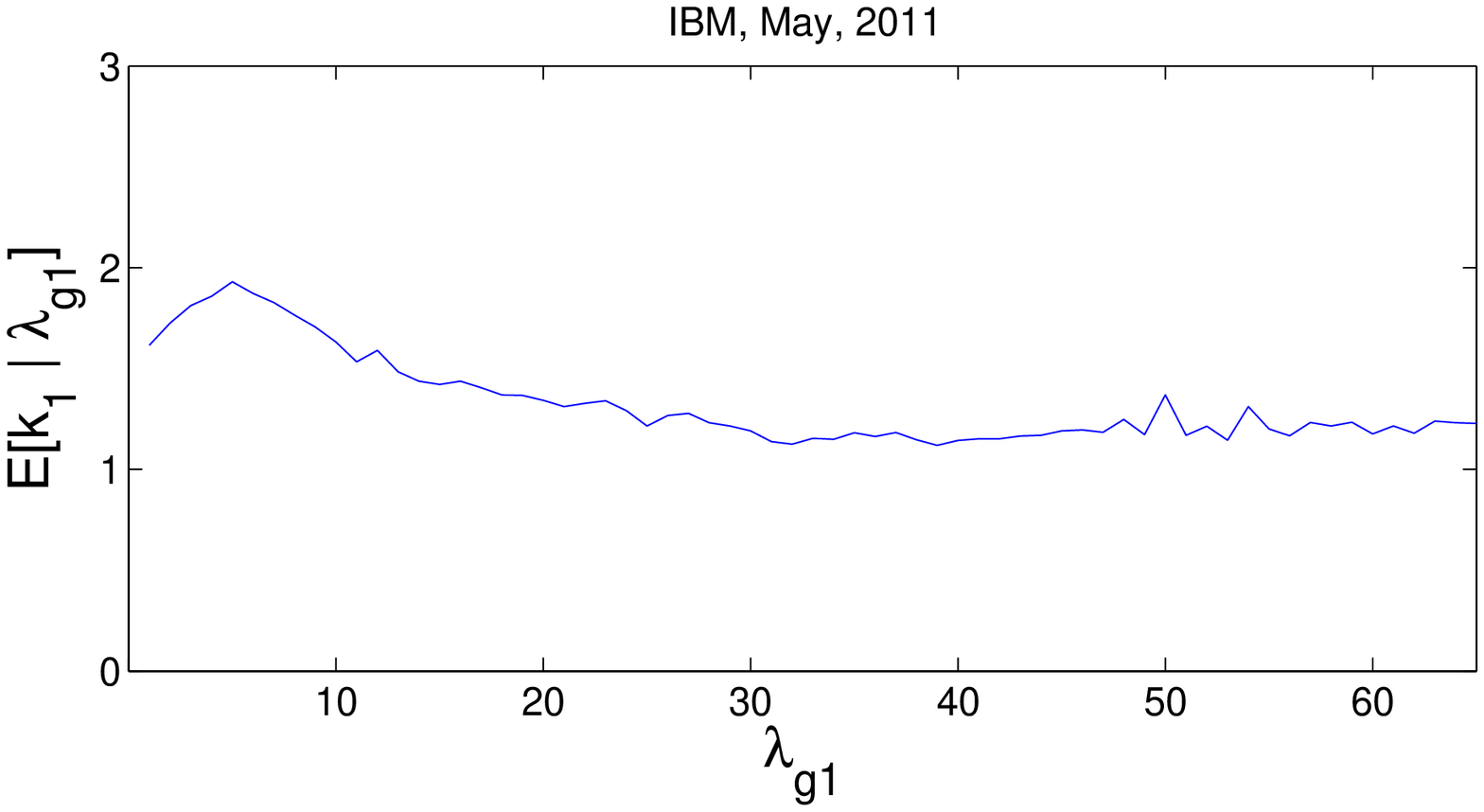}
	\end{subfigure}
	\begin{subfigure}[b]{0.45\textwidth}
		\includegraphics[width=\textwidth]{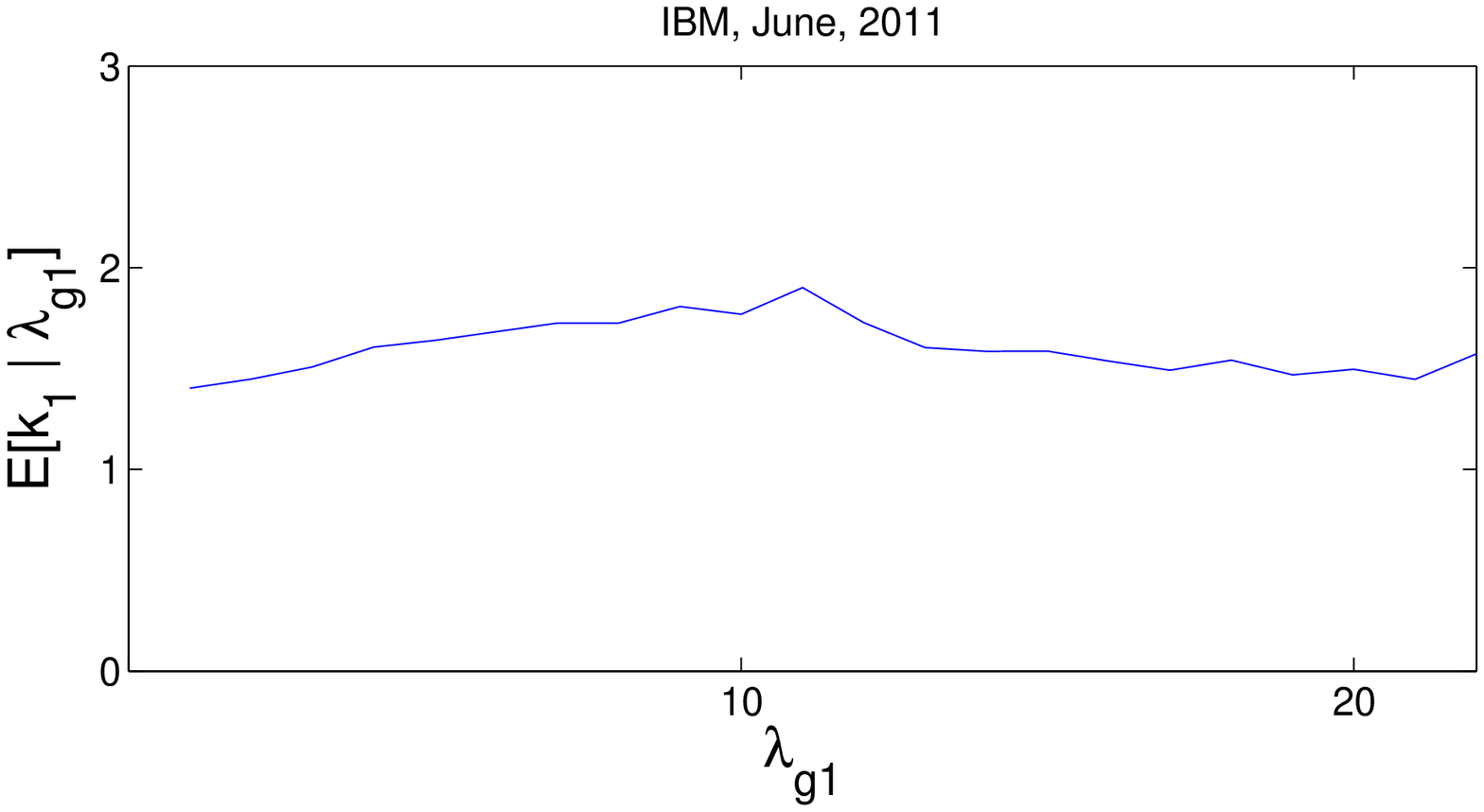}
	\end{subfigure}
	\caption{Conditional expectation of $k_1$ on $\lambda_{g1}$, IBM, 2011}\label{Fig:CEk_M}
\end{figure}

\subsection{Estimation result}

Table~\ref{Table:estimatesIBM2011} lists the likelihood estimation results of the marked Hawkes model with the tick data of IBM, January 2011, where $\log L_g$ of Eq.~\eqref{Eq:likelihood} is maximized.
The numerically computed standard errors are reported in the parentheses.
The estimations were performed on a daily basis, i.e., the estimates were recalculated in every business day.
The behaviors of $\mu, \alpha_s, \alpha_c,$ and $\beta$ in Figure~\ref{Fig:EstimationIBM2011} are similar to those estimated in the simple Hawkes model, see~\cite{LeeSeo}.

\begin{table}
\caption{Estimates of IBM, 2011 with linear impact function}\label{Table:estimatesIBM2011}
\centering
\begin{tabular}{ccccccc}
\hline
date & $\mu$ & $\alpha_s$ & $\alpha_c$ & $\beta$ & $\eta$ & $\log L_g$ \\
\hline 
0103 & 0.1080 & 0.6577 & 0.9956 & 2.2921 & 0.1241 & $-17185.5$ \\
 & (0.0021) & (0.0175) & (0.0218) & (0.0339) & (0.0187) \\
0104 & 0.1450 & 0.7033 & 1.0334 & 2.3527 & 0.2266 & $-17371.5$ \\
 & (0.0026) & (0.0157) & (0.0188) & (0.0272) & (0.0233) \\
0105 & 0.1079 & 0.9736 & 0.9414 & 2.550 & 0.1654 & $-13160.3$ \\
 & (0.0021) & (0.0207) & (0.0203) & (0.0334) & (0.0180)\\
0106 & 0.1335 & 0.8198 & 0.9475 & 2.3615 & 0.1357 & $-16603.4$\\
 & (0.0025) & (0.0173) & (0.0191) & (0.0296) & (0.0171) \\
0107 & 0.1588 & 0.8574 & 1.0145 & 2.435 & 0.1560 & $-15956.4$\\
 & (0.0029) & (0.0164) & (0.0184) & (0.0284) & (0.0176)\\
0110 & 0.1338 & 0.7423 & 1.0724 & 2.4388 & 0.1927 & $-16141.3$\\
 & (0.0025) & (0.0160) & (0.0191) & (0.0266) & (0.0163)\\
0111 & 0.1271 & 0.5855 & 1.2108 & 2.4403 & 0.1570 & $-16923.1$\\
 & (0.0024) & (0.0159) & (0.0223) & (0.0314) & (0.0155)\\
0112 & 0.1160 &	0.6517 & 0.8552 & 2.0639 & 0.3561 & $-19492.9$\\
 & (0.0023) & (0.0158) & (0.0185) & (0.0307) & (0.0481) \\
0113 & 0.1042 & 0.7245 & 1.0502 & 2.5284 & 0.2372 & $-15508.4$\\
 & (0.0020) & (0.0192) & (0.0240) & (0.0369) & (0.0261) \\
0114 & 0.1138 & 0.6702 & 0.8798 & 2.3142 & 0.2380 & $-16589.3$ \\
 & (0.0022) & (0.0175) & (0.0202) & (0.0341) & (0.0183)\\
0118 & 0.1330 &	0.5642&	1.1548&	2.5082&	0.1651 & $-16374.7$\\
 & (0.0024) & (0.0169) & (0.0239) & (0.0354) & (0.0147) \\
0119 & 0.2198 &	0.5423&	1.2631 & 2.4964 & 0.1323 & $-11223.0$\\
 & (0.0036) & (0.0133) & (0.0207) & (0.0255) & (0.0093) \\
0120 & 0.1509 & 0.7060 & 1.0017 & 2.3114 & 0.1824 & $-15709.4$\\
 & (0.0028) & (0.0158) & (0.0211) & (0.0332) & (0.0154) \\
0121 & 0.1447 & 0.4901 & 1.3356 & 2.5806 & 0.1545 & $-16524.6$ \\
 & (0.0026) & (0.0152) & (0.0247) & (0.0333) & (0.0132)\\
0124 & 0.1771 &	0.6095 & 1.2424 & 2.5658 & 0.1649 & $-14711.6$\\
 & (0.0030) & (0.0151) & (0.0205) & (0.0290) & (0.0132)\\
0125 & 0.1669 & 0.8214 & 0.9528 & 2.2982 &	0.1954 & $-13602.1$\\
 & (0.0030) & (0.0147) & (0.0167) & (0.0244) & (0.0170) \\
0126 & 0.1421 & 0.5939 & 1.3809 & 2.6146 & 0.1217 & $-9239.8$\\
 & (0.0026) & (0.0154) & (0.0231) & (0.0295) & (0.0100) \\
\hline
\end{tabular}
\end{table}

\begin{figure}
		\centering
        \begin{subfigure}[b]{0.45\textwidth}
                \includegraphics[width=\textwidth]{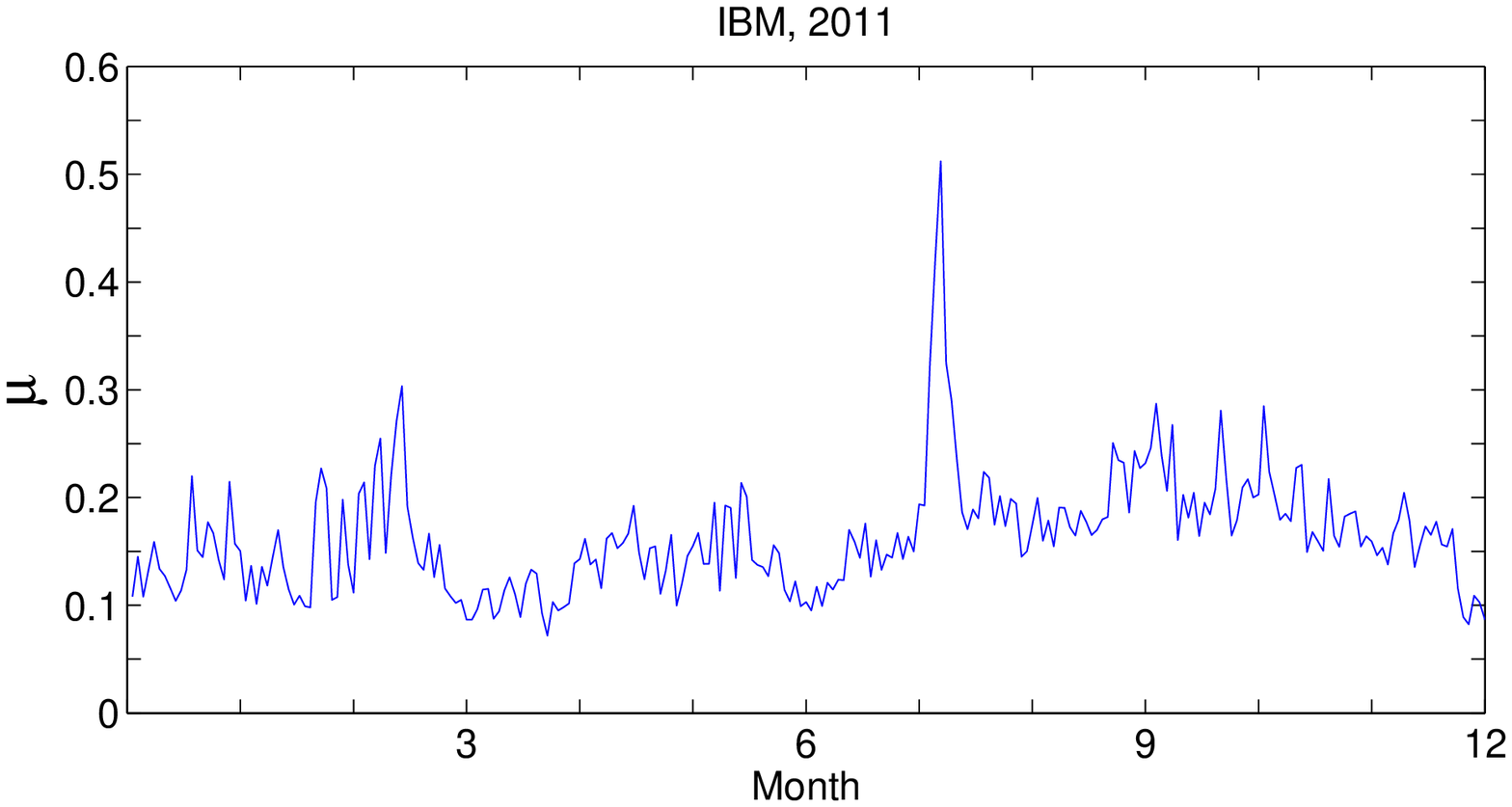}
                \caption{$\mu$}
                \label{Fig:mu_IBM2011}
        \end{subfigure}
       	\centering
        \begin{subfigure}[b]{0.45\textwidth}
                \includegraphics[width=\textwidth]{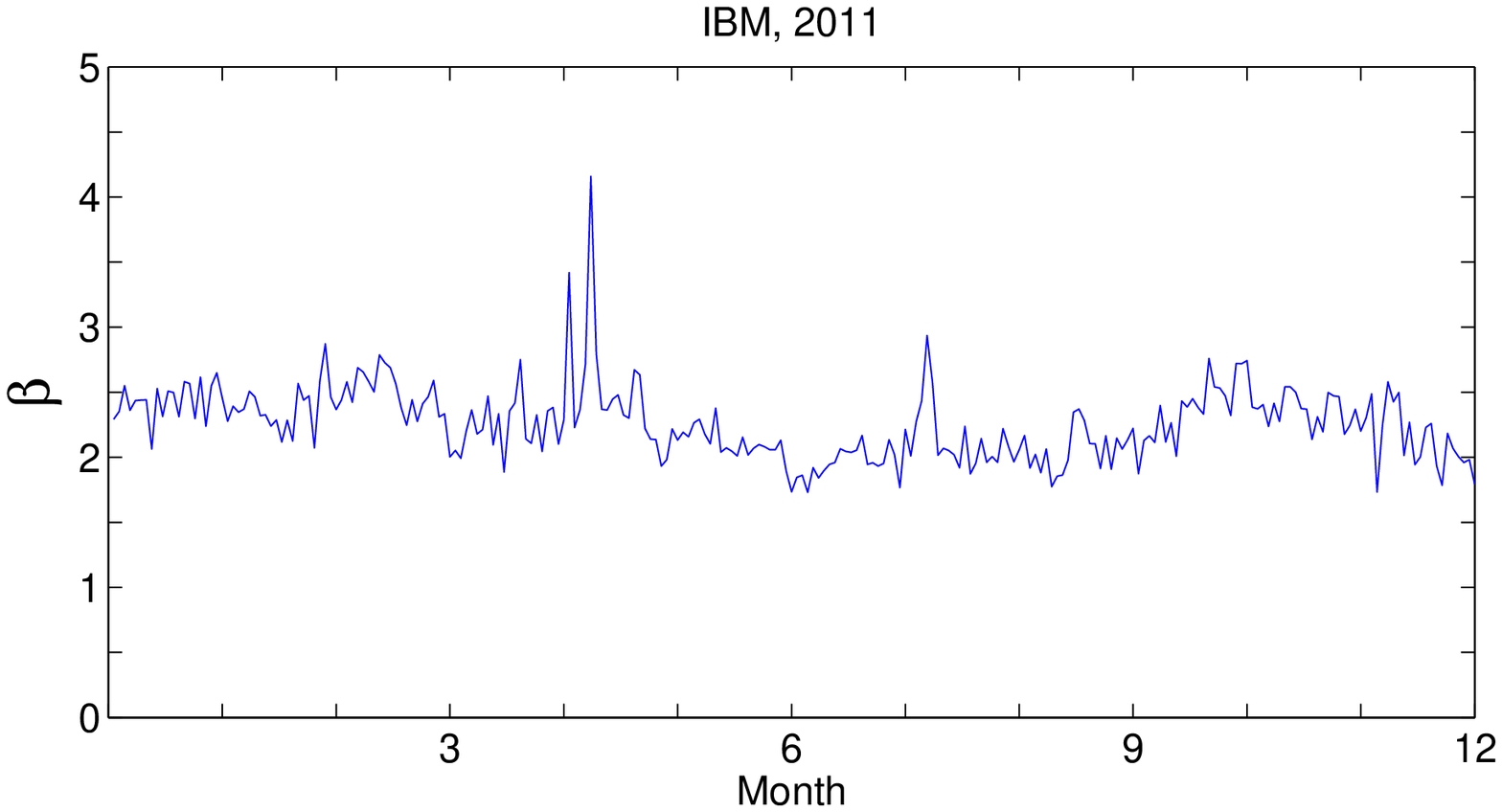}
                \caption{$\beta$}
                \label{Fig:beta_IBM2011}
        \end{subfigure}
	    \centering
        \begin{subfigure}[b]{0.45\textwidth}
                \includegraphics[width=\textwidth]{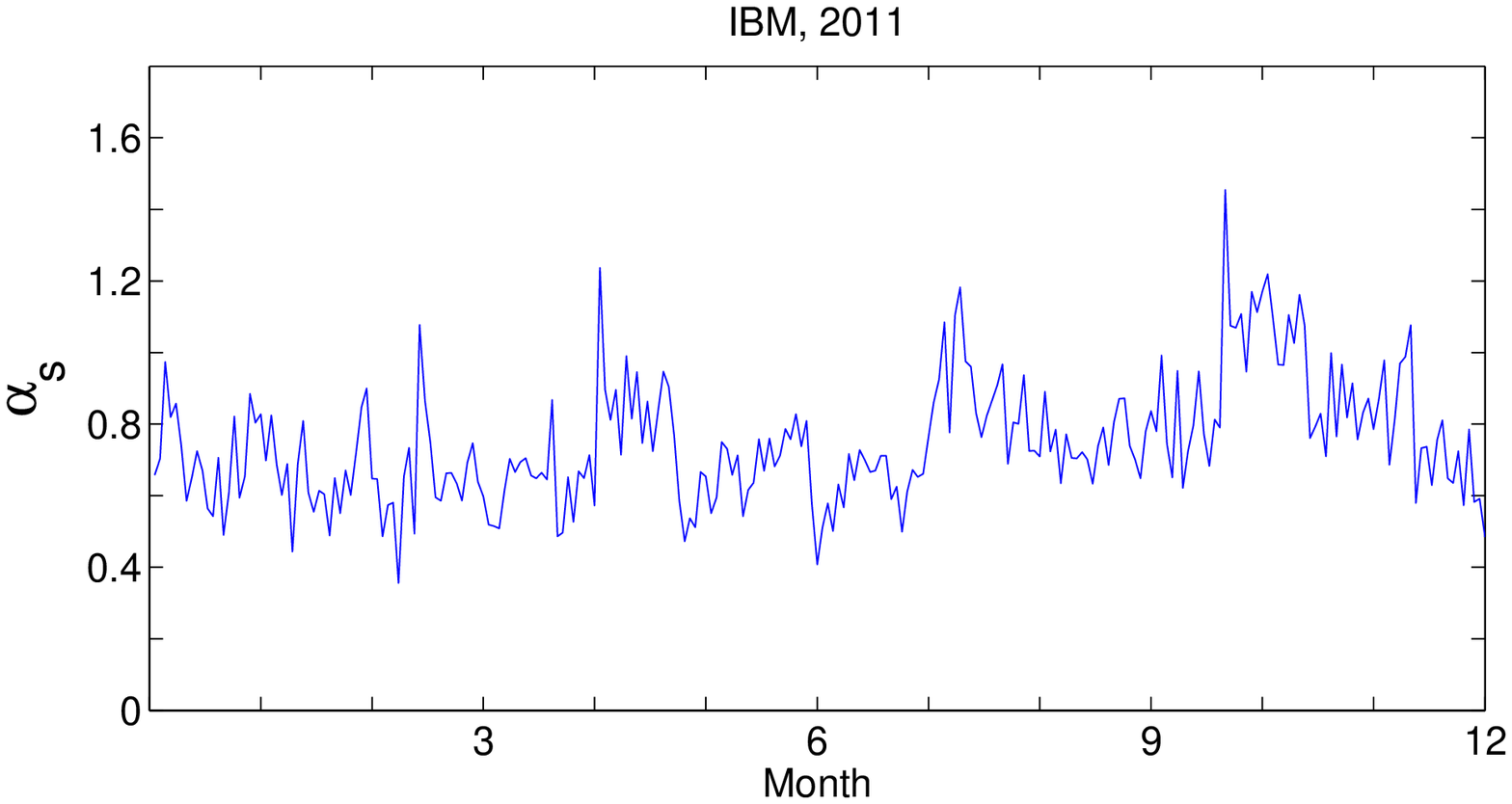}
                \caption{$\alpha_s$}
                \label{Fig:alphas_IBM2011}
        \end{subfigure}
	    \centering
        \begin{subfigure}[b]{0.45\textwidth}
                \includegraphics[width=\textwidth]{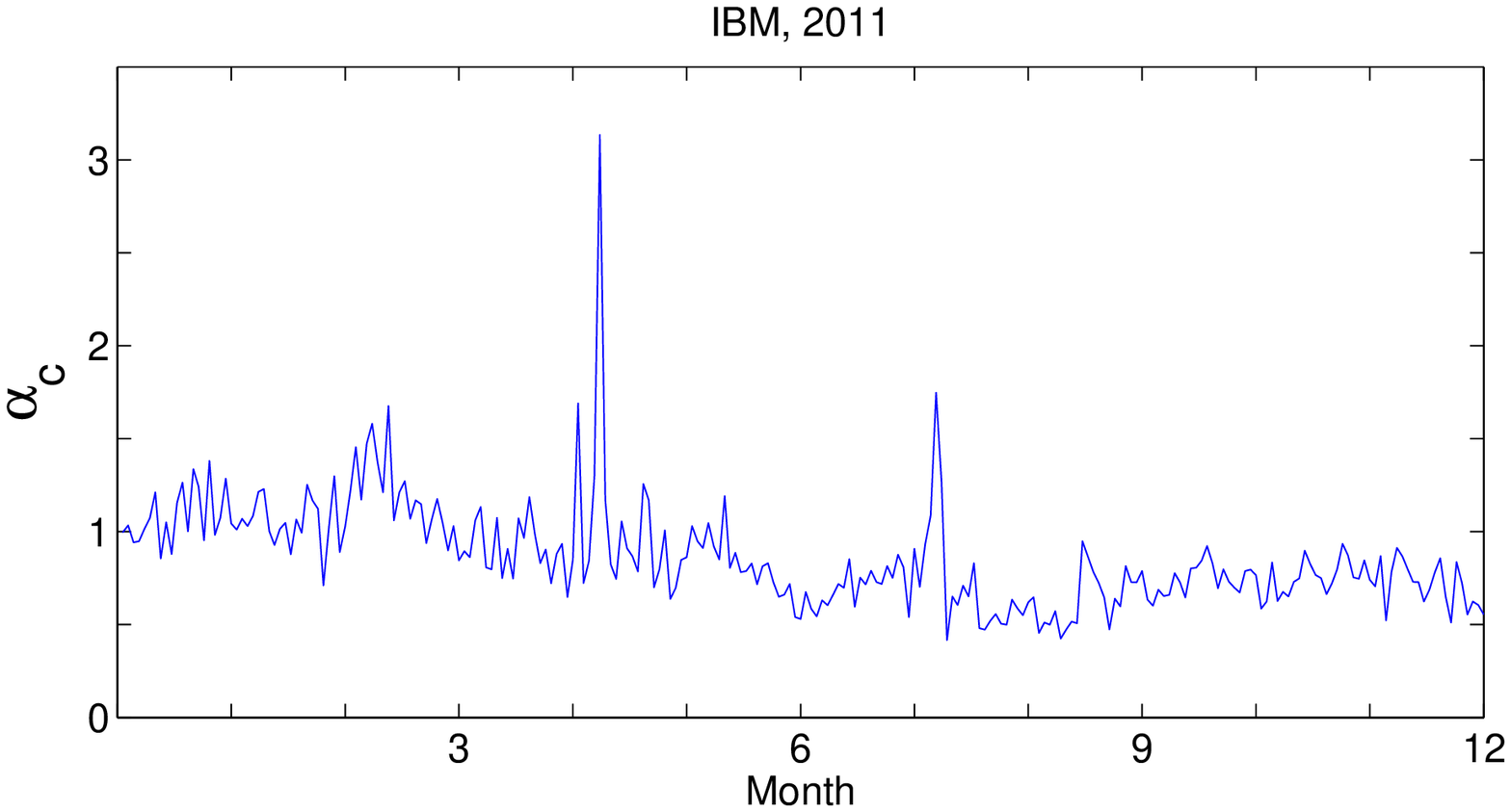}
                \caption{$\alpha_c$}
                \label{Fig:alphac_IBM2011}
        \end{subfigure}
		\caption{Marked Hawkes estimation result, IBM, 2011}\label{Fig:EstimationIBM2011}
\end{figure}

The dynamics of $\eta$ are illustrated in Figure~\ref{Fig:eta} where $\eta$ was estimated around 0.2 in general from 2008 to 2011.
The slope parameter for the impact function, $\eta$, is positive in general, which means that large mark tends to have a large impact for the future intensities.
On the other hand, $\eta$ is less than 1 and this implies that the impact of mark size 2 is generally less than the total impact of the two consecutive unit size jumps that occur over a very short time interval. 
Note that few negative $\eta$ are also observed.
Figure~\ref{Fig:etaCVX} shows the dynamics of $\eta$ estimated from CVX.
The overall behaviors of $\eta$ of IBM and CVX are similar but the $\eta$ of CVX was more volatile.

Figure~\ref{Fig:vol} compares the Hawkes volatility computed by Remark~\ref{Remark:vol} and TSRV of IBM, 2008-2011.
The trends of the Hawkes volatility and TSRV are similar but
the Hawkes volatility is generally larger than the TSRV, especially when the volatility is high.
This tendency was also found in the simple symmetric Hawkes model.
This discrepancy might be because of the restriction that the parameter settings need be symmetric and Markovian.
On the other hand, the precise reason for this is unclear at this point.
Note that in the previous simulation study, the Hawkes volatility and TSRV converges.

The empirical studies suggest that the parameter restrictions for symmetry do not perfectly meet with the real data.
Figure~\ref{Fig:EstimationIBM2011full} presents the dynamics of all parameters of the fully characterized Hawkes model of Subsection~\ref{Subsect:full} for IBM, 2011.
The result shows that for each parameter pair $(\mu_1, \mu_2)$, $(\alpha_{11}, \alpha_{22})$, $(\alpha_{12}, \alpha_{21})$, $(\beta_{11}, \beta_{22})$ and $(\beta_{12}, \beta_{21})$, 
a similar trend is observed over time but those were not exactly the same as each other. 
The summary statistics in Table~\ref{Table:IBM2011_full} show that the parameter pair has a similar mean but the mean absolute percentage error (MAPE) also shows a difference between the parameters.
The row `MAPE' presents the error between two adjacent parameters.
In addition, $\beta_{11}$ and $\beta_{12}$ are different, even in the mean.
A similar observation is found in the simple Hawkes approach.

\begin{figure}
        \centering
        \begin{subfigure}[b]{0.5\textwidth}
                \includegraphics[width=\textwidth]{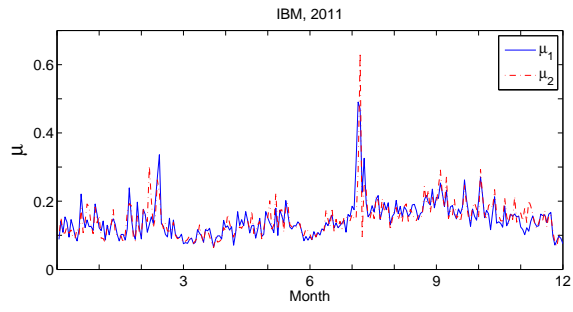}
                \caption{$\mu_1$ and $\mu_2$}
                \label{Fig:IBM2011_mu12}
        \end{subfigure}
        \centering
        \begin{subfigure}[b]{0.5\textwidth}
                \includegraphics[width=\textwidth]{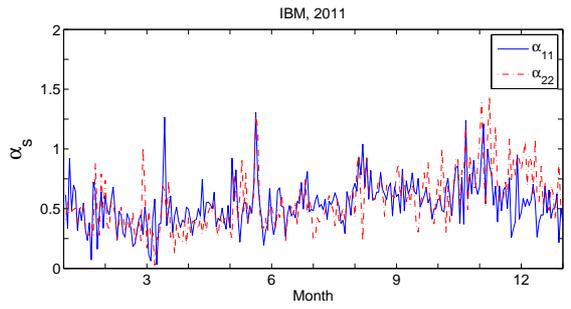}
                \caption{$\alpha_{11}$ and $\alpha_{22}$}
                \label{Fig:IBM2011_alpha1122}
        \end{subfigure}
	    \centering
        \begin{subfigure}[b]{0.5\textwidth}
                \includegraphics[width=\textwidth]{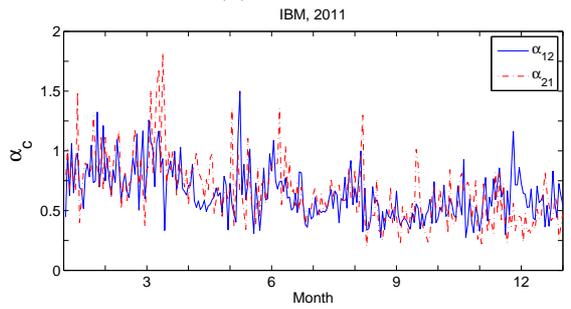}
                \caption{$\alpha_{12}$ and $\alpha_{21}$}
                \label{Fig:IBM2011_alpha1221}
        \end{subfigure}
	    \centering
        \begin{subfigure}[b]{0.5\textwidth}
                \includegraphics[width=\textwidth]{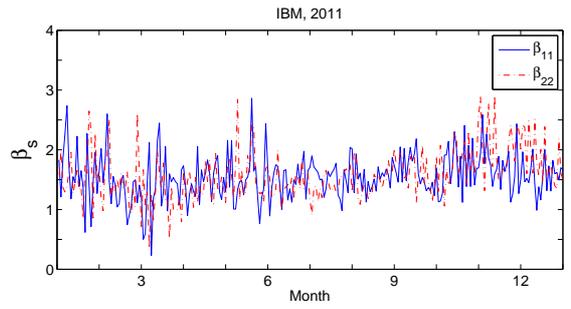}
                \caption{$\beta_{11}$ and $\beta_{22}$}
                \label{Fig:IBM2011_beta1122}
        \end{subfigure}
	    \centering
        \begin{subfigure}[b]{0.5\textwidth}
                \includegraphics[width=\textwidth]{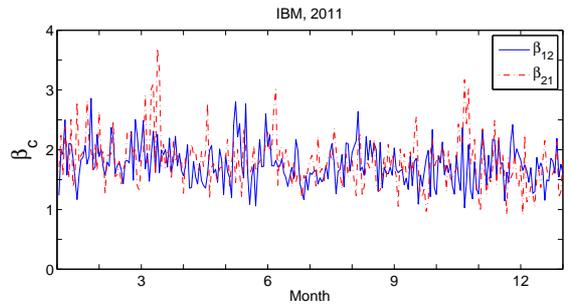}
                \caption{$\beta_{12}$ and $\beta_{21}$}
                \label{Fig:IBM2011_beta1221}
        \end{subfigure}
		\caption{Estimation result with the fully characterized marked Hawkes, IBM, 2011}\label{Fig:EstimationIBM2011full}
\end{figure}

\begin{table}
\caption{Estimation result of fully characterized self and mutually excited Hawkes process, IBM, 2011}\label{Table:IBM2011_full}
\centering
\begin{tabular}{ccccccccccccc}

\hline
 & $\mu_1$ & $\mu_2$ & $\alpha_{11}$ & $\alpha_{22}$& $\alpha_{12}$& $\alpha_{21}$ & $\beta_{11}$ &  $\beta_{22}$ & $\beta_{12}$ & $\beta_{21}$\\
\hline
IBM, 2011 \\
mean & 0.1462 & 0.1470 & 0.5404 & 0.5621 & 0.6443 & 0.6611 & 1.5499 & 1.5793 & 1.7538 & 1.8056 \\
std. & 0.0543 & 0.0545 & 0.1962 & 0.2361 & 0.2144 & 0.2717 & 0.3845 & 0.4192 & 0.3373 & 0.4282 \\
MAPE & \multicolumn{2}{c}{0.1436} & \multicolumn{2}{c}{0.3924} & \multicolumn{2}{c}{0.3179} & \multicolumn{2}{c}{0.2881} & \multicolumn{2}{c}{0.2247}\\
\hline
\end{tabular}
\end{table}

Interestingly, when the stock market is in a highly volatile state, the slope for the impact function, $\eta$ is estimated to be relatively close to zero.
For example, September 29, 2008, at the beginning of the financial crisis, the reported $\eta$ of IBM was around 0.02, which was much smaller than the annual average, 0.16, when the market is very unstable with a TSRV of 0.8766 and a Hawkes volatility 1.4530.
In the May 6, 2010 Flash Crash, the estimated $\eta$ of IBM was around 0.01 (annual average of $\eta= 0.23$) when the TSRV was 0.9656 and the Hawkes volatility was 1.3264.
Because of the statement of Federal Reserve, the stock market was highly volatile at August 9, 2011, and the $\eta$ of IBM was around 0.01 which is much smaller than the annual average 0.11.
CVX also showed a similar pattern.
The estimated $\eta$ of CVX for the above cases are around 0.04, 0.01, and 0.07, respectively, whereas the annual averages in 2008, 2010 and 2011 were 0.14, 0.41 and 0.24, respectively.
In a highly volatile market, a much larger number of large size marks are observed than in a stable market, and with those large size marks, the linear relationship between the mark size and the future impact is weak.

\begin{figure}
	\centering
	\begin{subfigure}[b]{0.45\textwidth}
		\includegraphics[width=\textwidth]{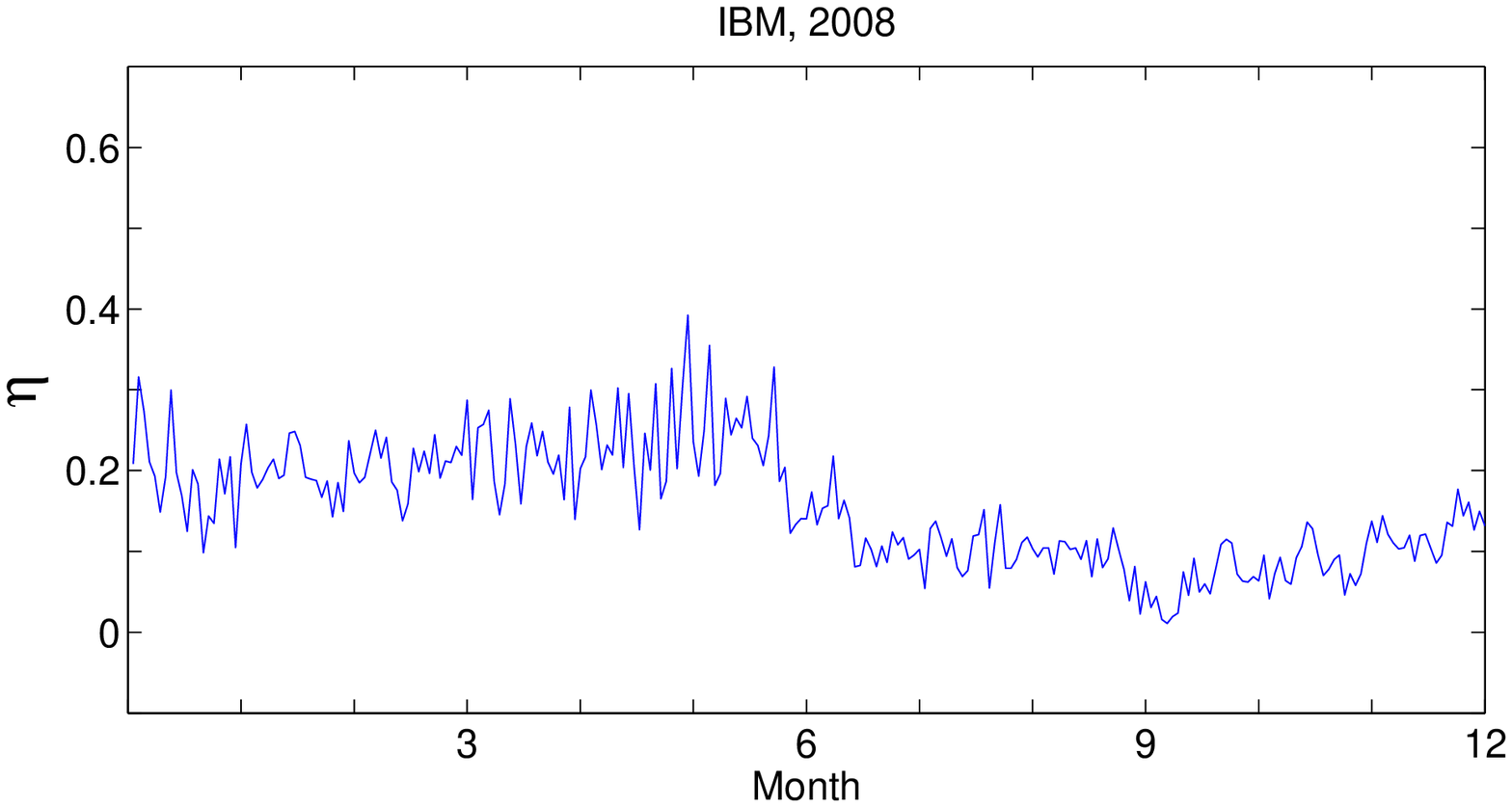}
	\end{subfigure}
	\begin{subfigure}[b]{0.45\textwidth}
		\includegraphics[width=\textwidth]{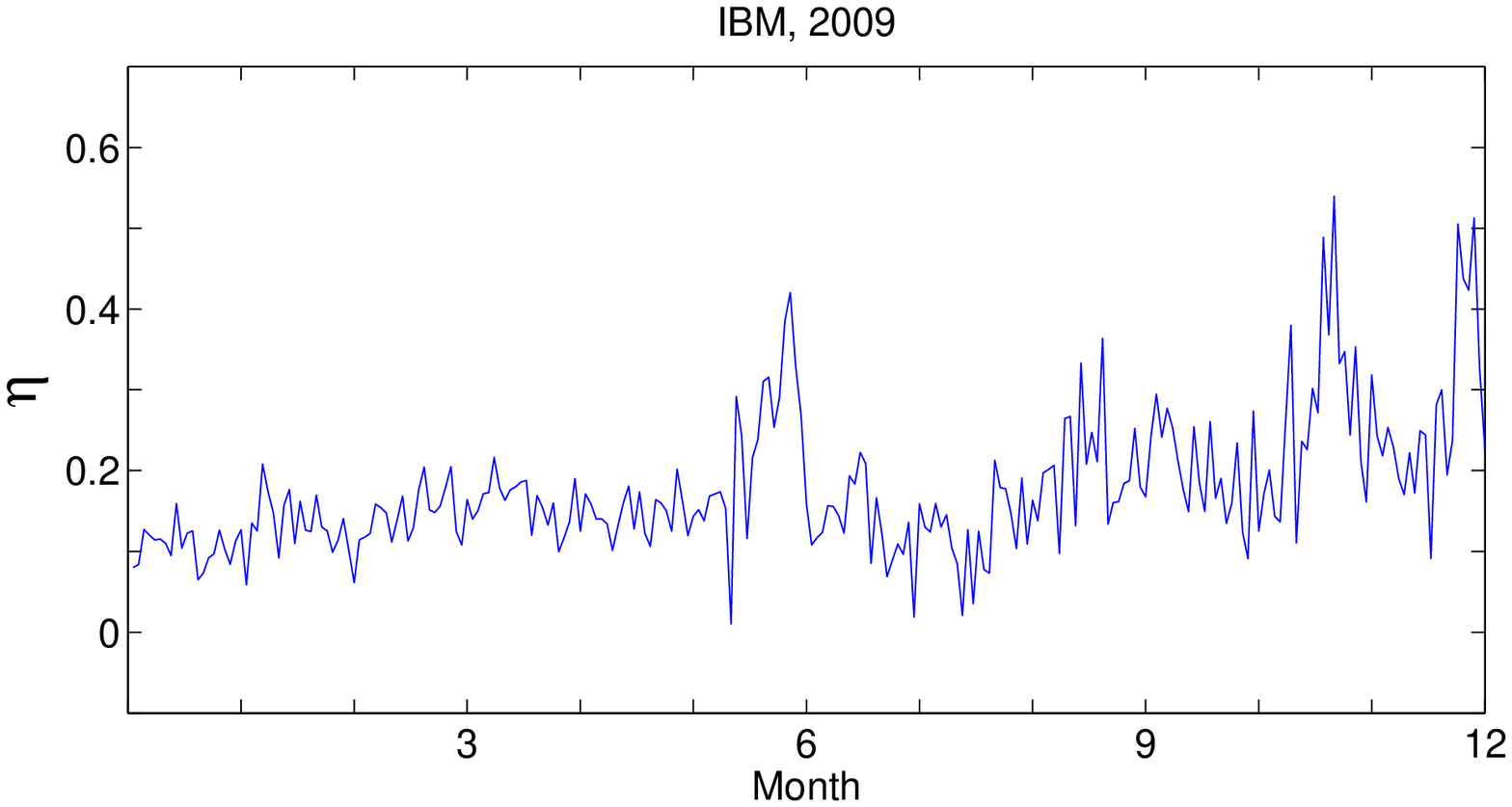}
	\end{subfigure}
	\begin{subfigure}[b]{0.45\textwidth}
		\includegraphics[width=\textwidth]{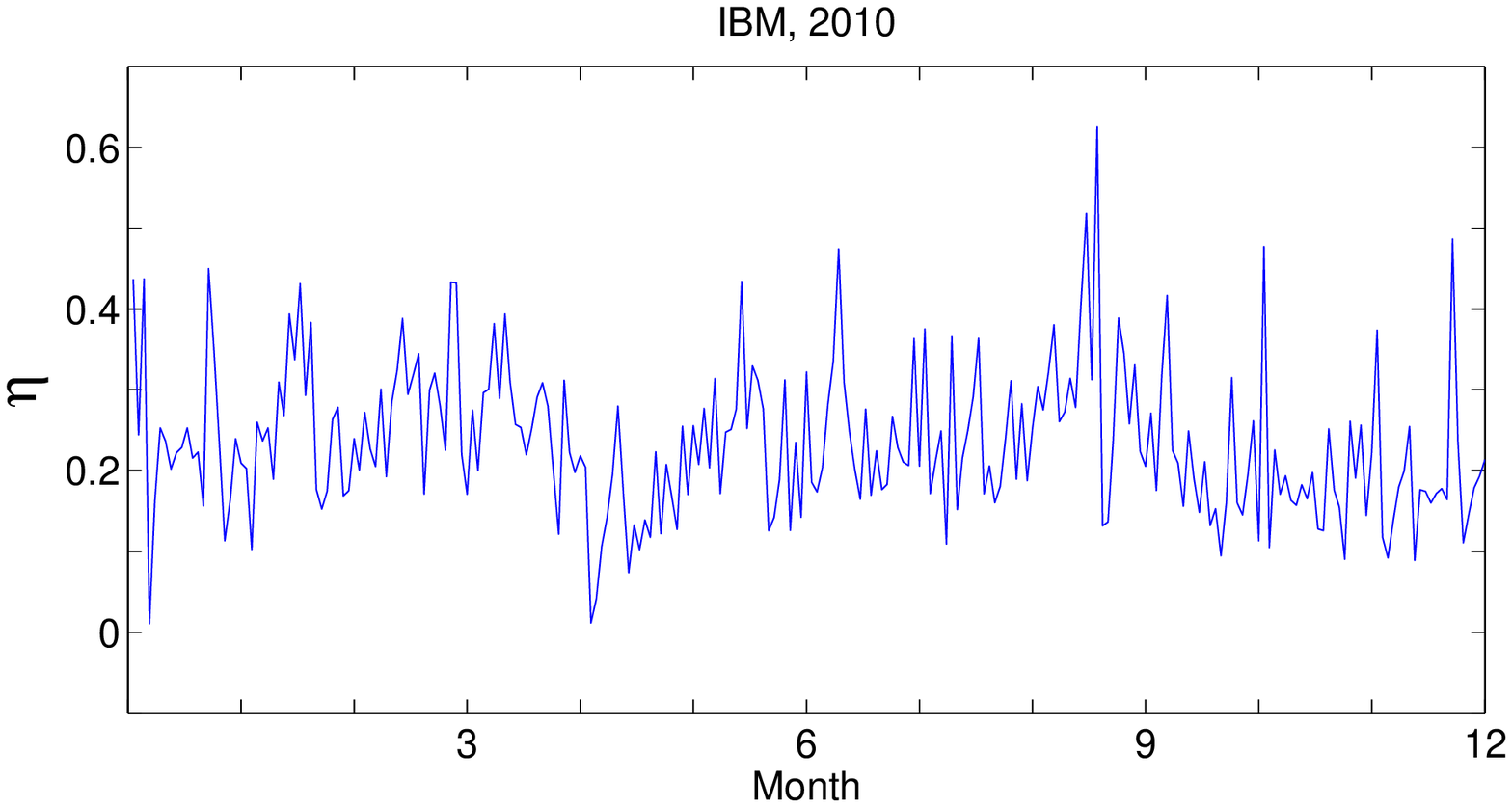}
	\end{subfigure}
	\begin{subfigure}[b]{0.45\textwidth}
		\includegraphics[width=\textwidth]{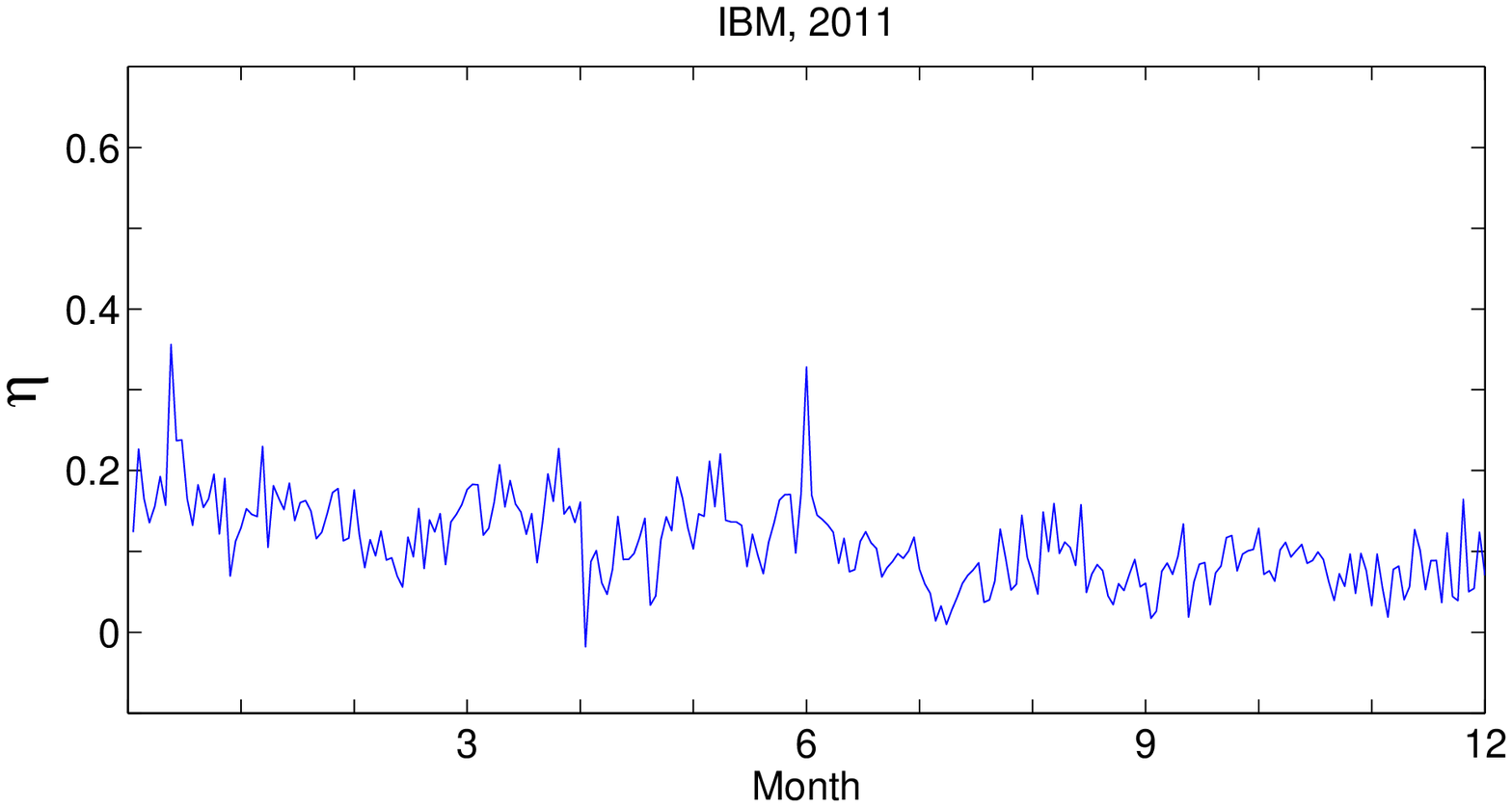}
	\end{subfigure}
	\caption{Estimation results of $\eta$, IBM, 2008-2011}\label{Fig:eta}
\end{figure}

\begin{figure}
	\centering
	\begin{subfigure}[b]{0.45\textwidth}
		\includegraphics[width=\textwidth]{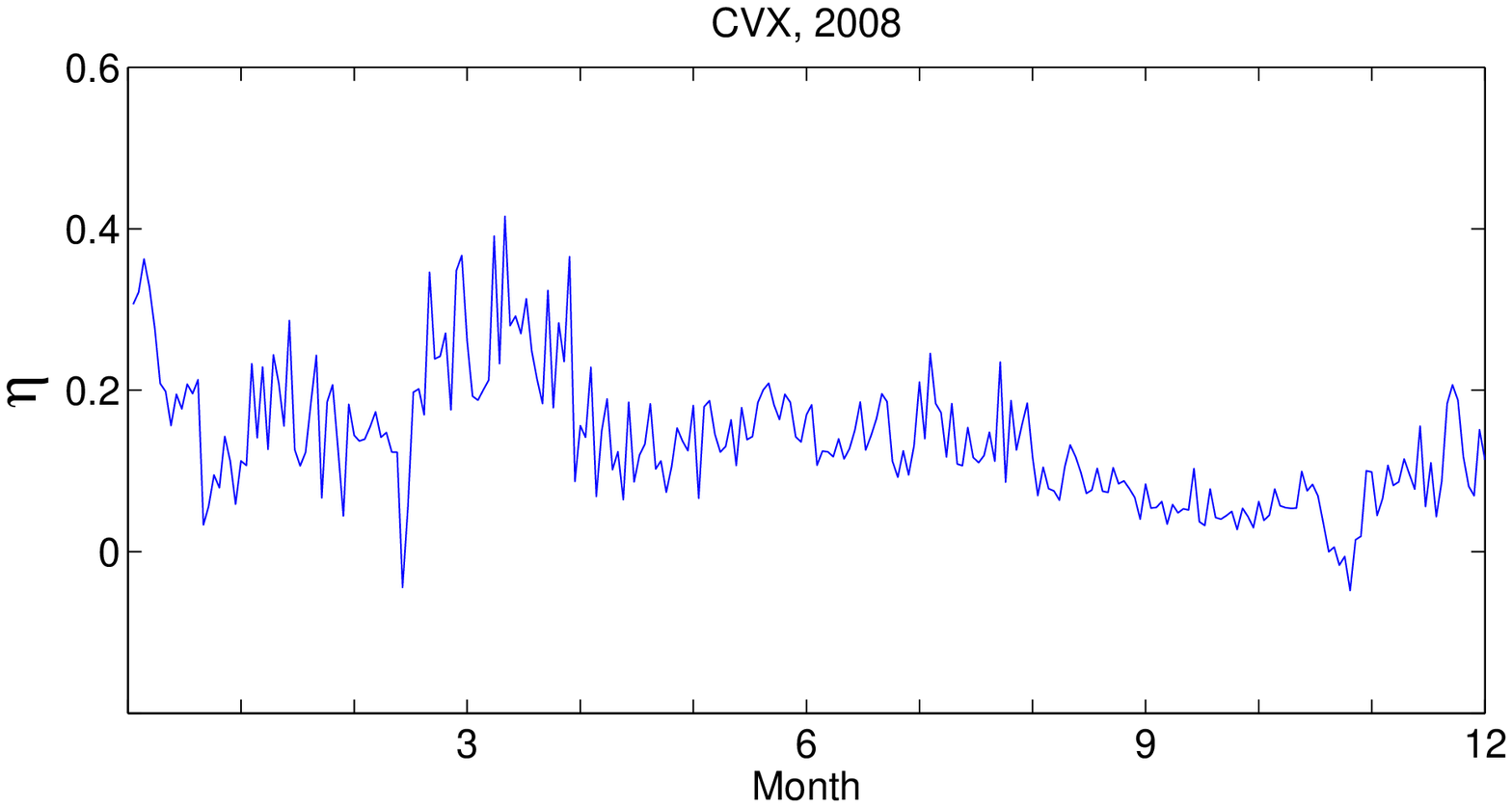}
	\end{subfigure}
	\begin{subfigure}[b]{0.45\textwidth}
		\includegraphics[width=\textwidth]{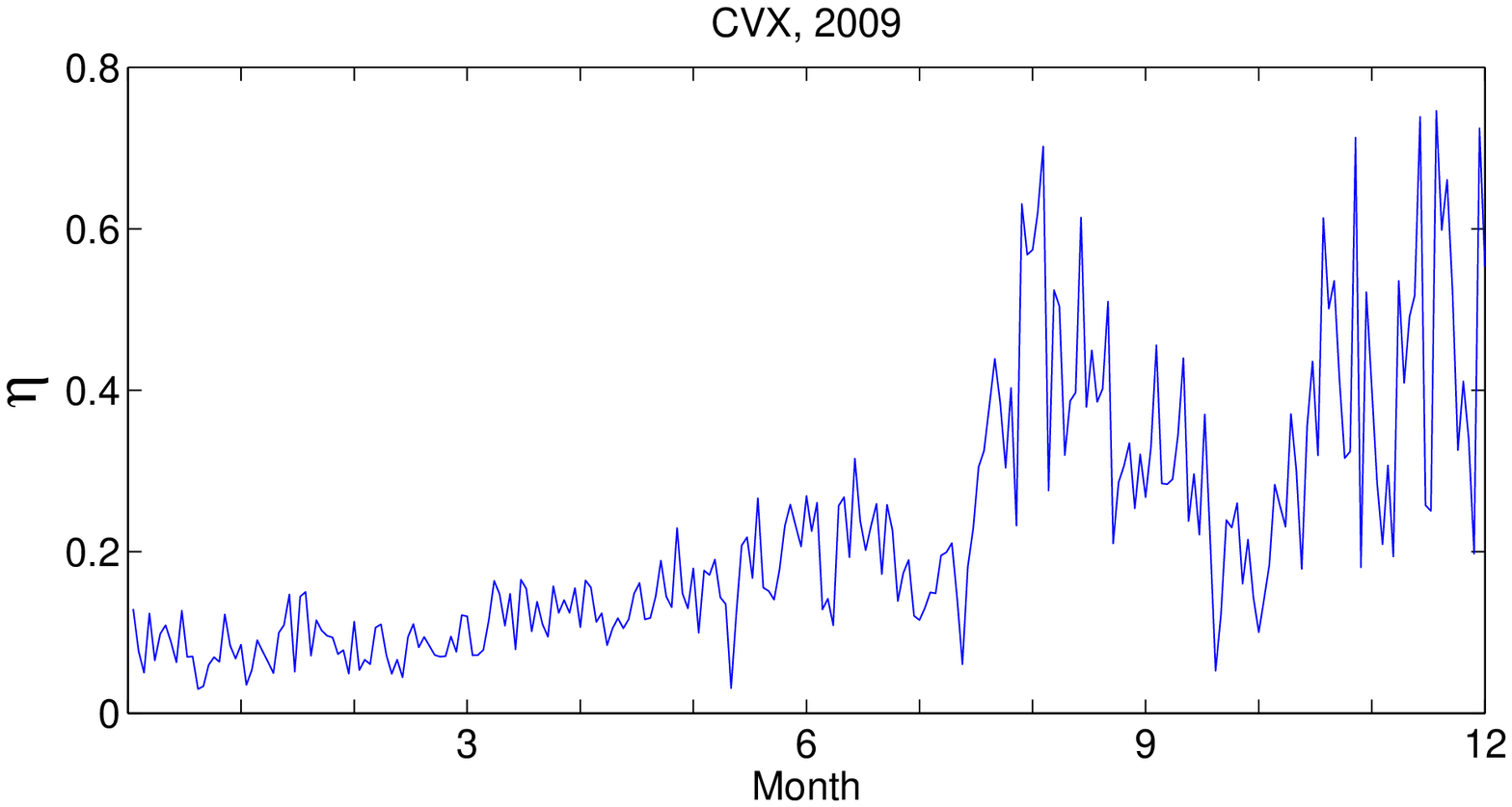}
	\end{subfigure}
	\begin{subfigure}[b]{0.45\textwidth}
		\includegraphics[width=\textwidth]{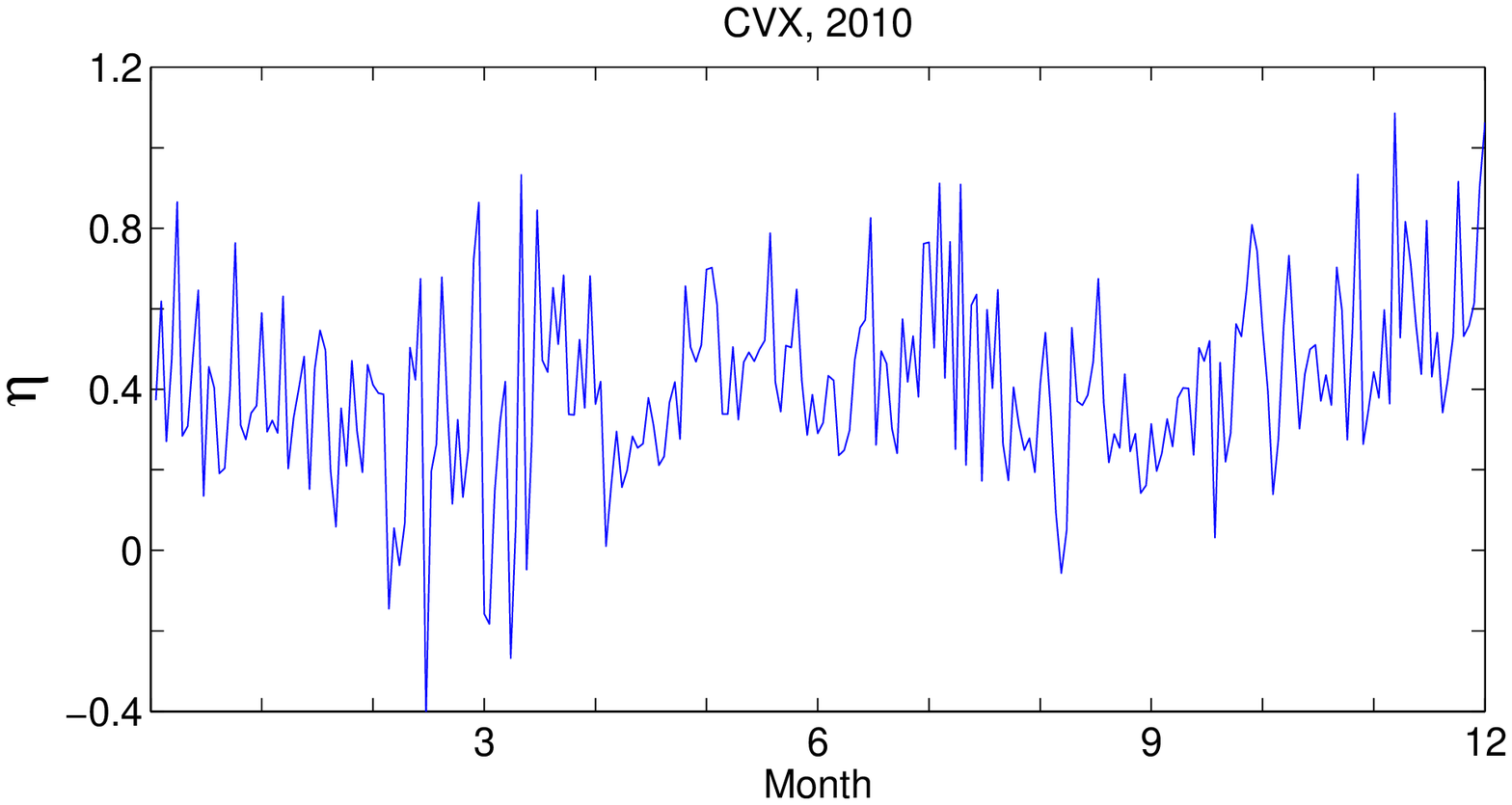}
	\end{subfigure}
	\begin{subfigure}[b]{0.45\textwidth}
		\includegraphics[width=\textwidth]{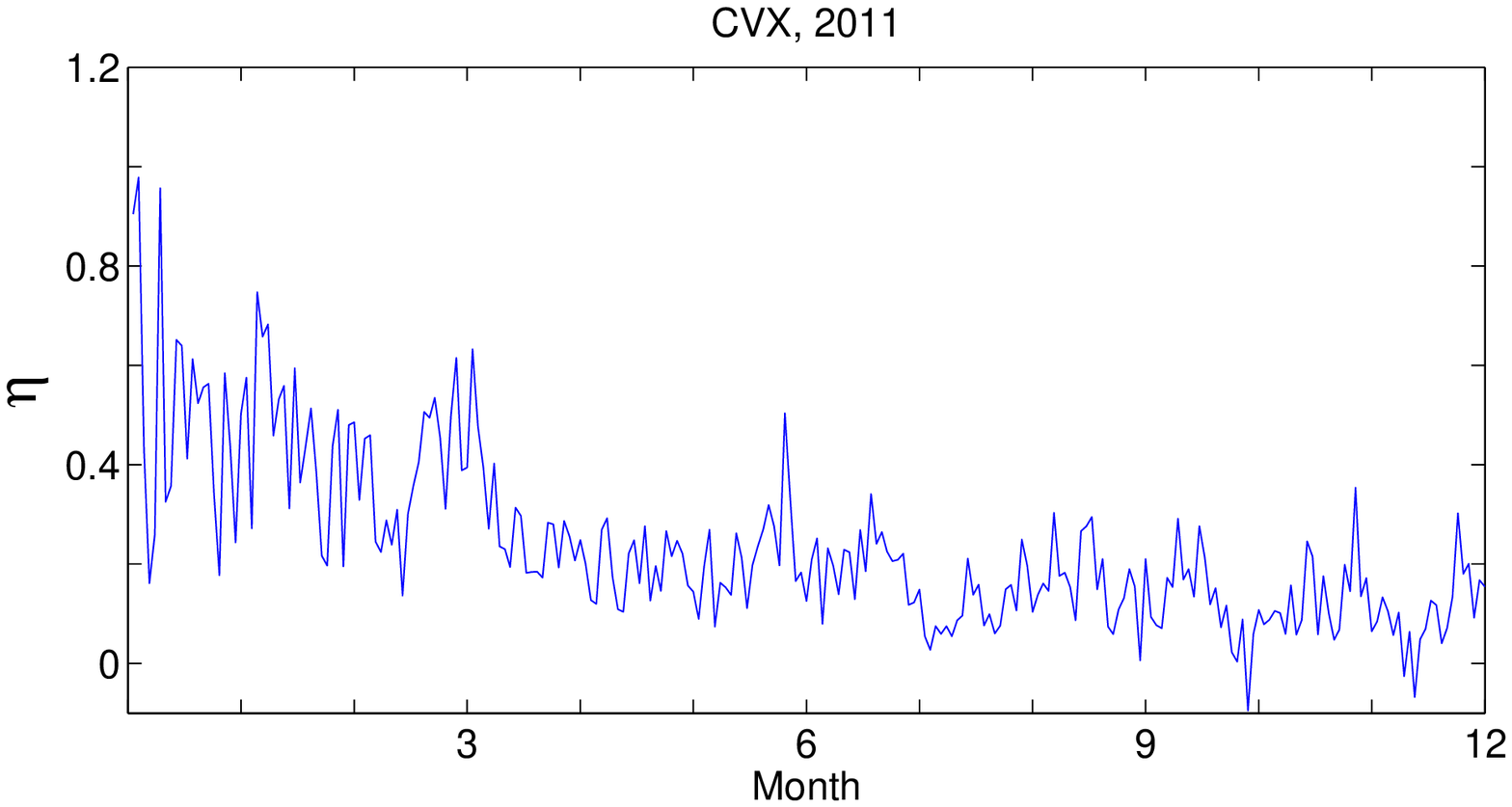}
	\end{subfigure}
	\caption{Estimation results of $\eta$, CVX, 2008-2011}\label{Fig:etaCVX}
\end{figure}

\begin{figure}
	\centering
	\begin{subfigure}[b]{0.45\textwidth}
		\includegraphics[width=\textwidth]{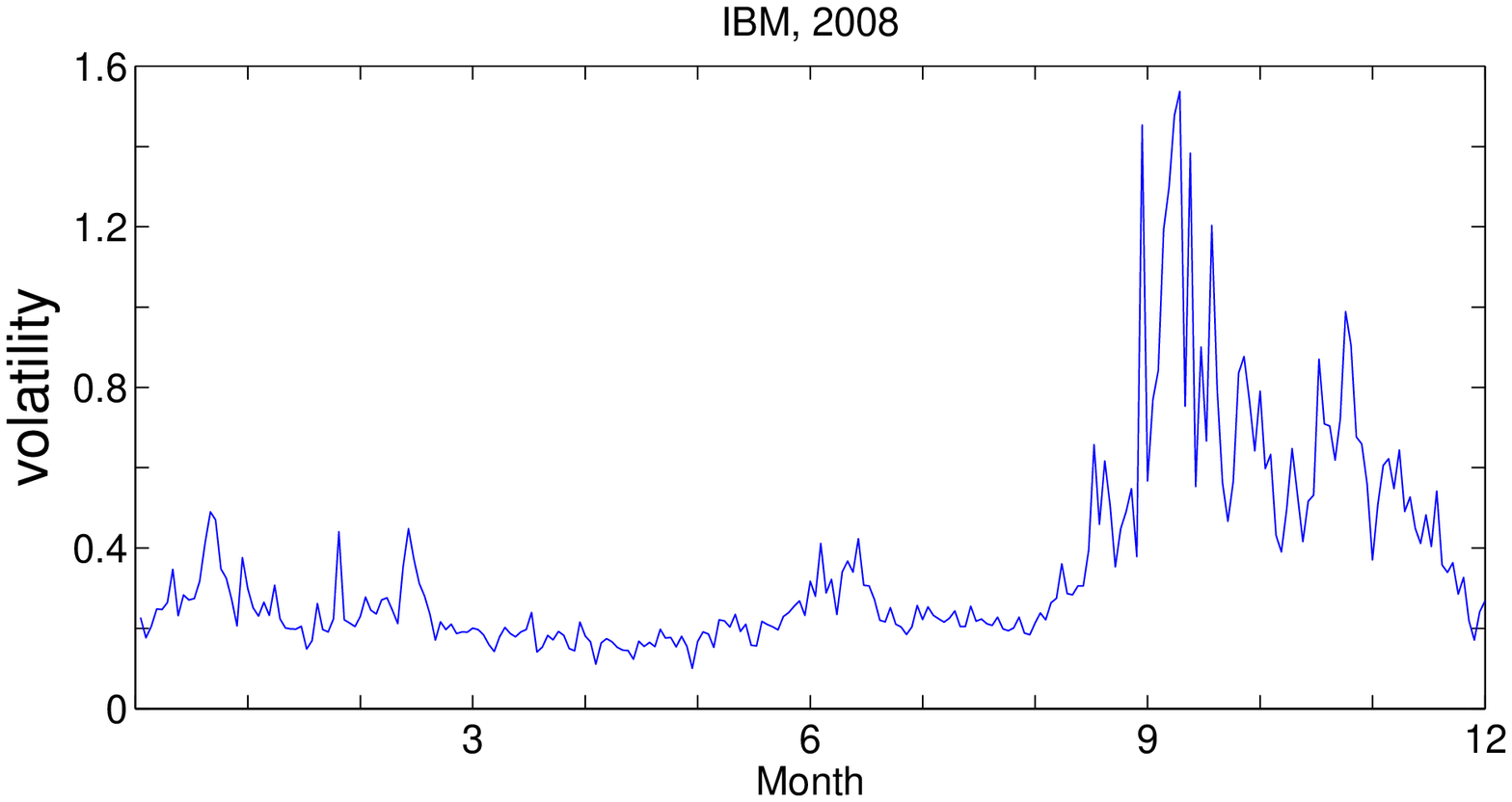}
		\caption{Hawkes volatility, 2008}
	\end{subfigure}
	\begin{subfigure}[b]{0.45\textwidth}
		\includegraphics[width=\textwidth]{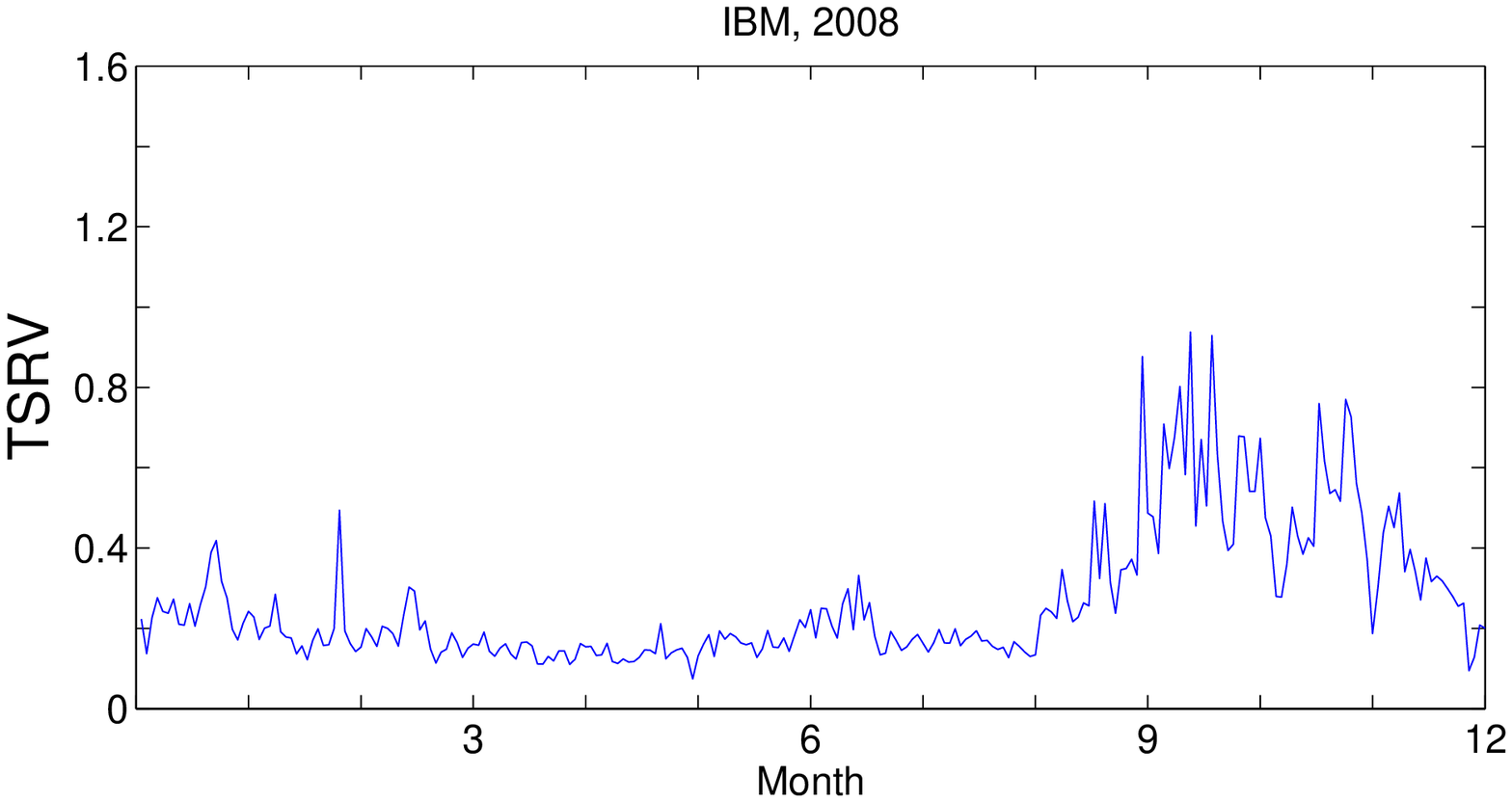}
		\caption{TSRV, 2008}
	\end{subfigure}
	\centering
	\begin{subfigure}[b]{0.45\textwidth}
		\includegraphics[width=\textwidth]{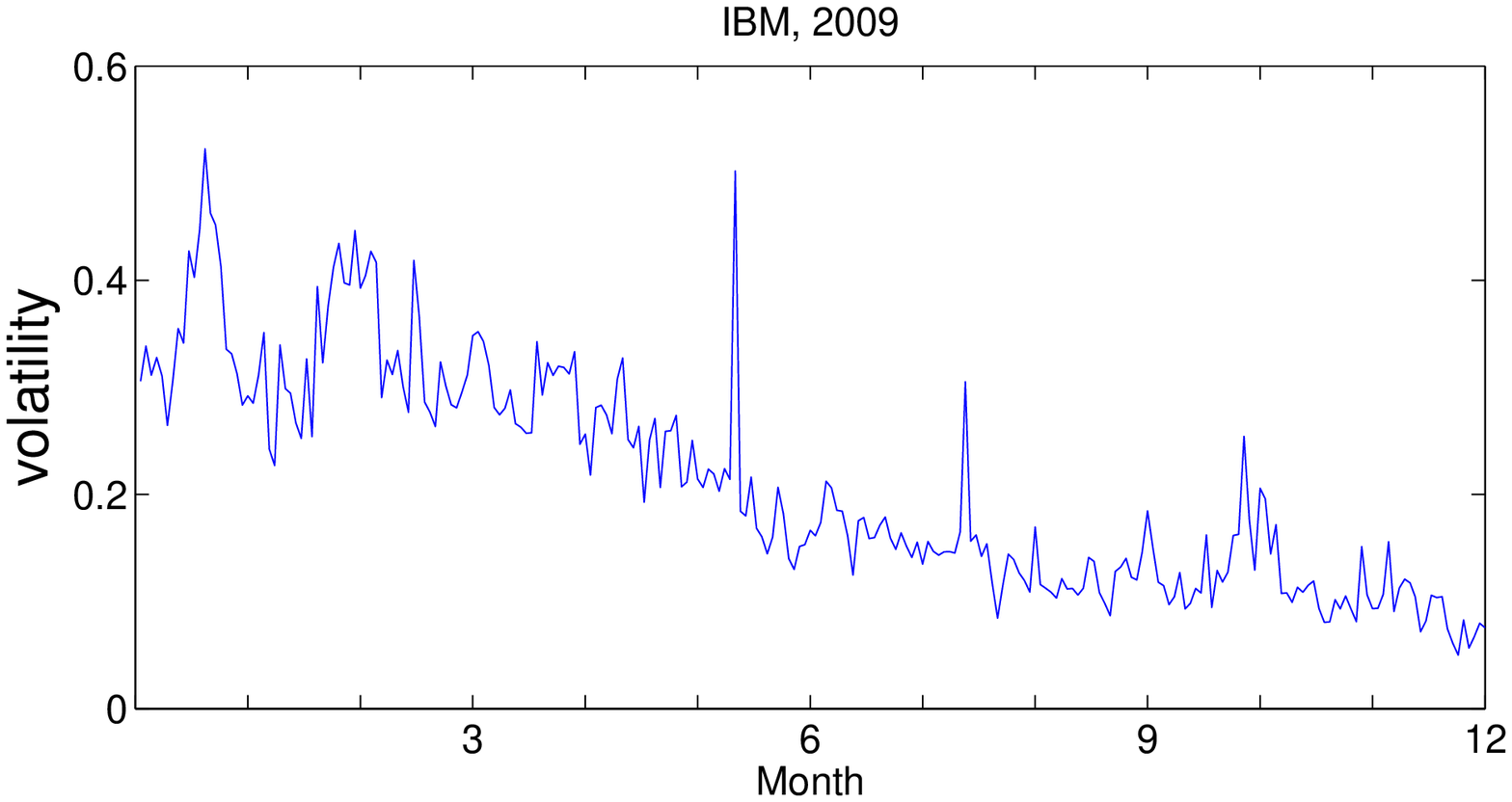}
		\caption{Hawkes volatility, 2009}
	\end{subfigure}
	\begin{subfigure}[b]{0.45\textwidth}
		\includegraphics[width=\textwidth]{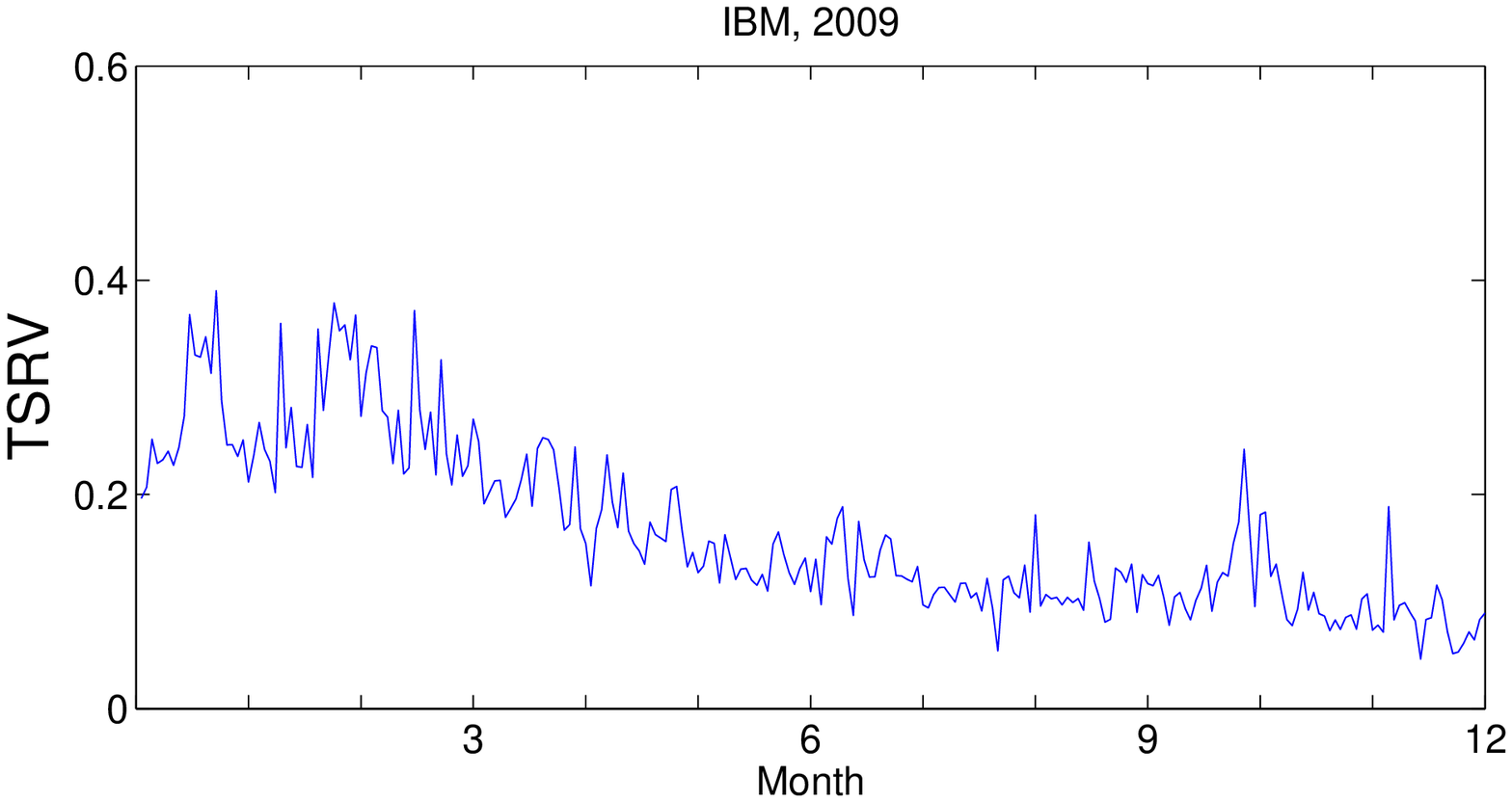}
		\caption{TSRV, 2009}
	\end{subfigure}
	\centering
	\begin{subfigure}[b]{0.45\textwidth}
		\includegraphics[width=\textwidth]{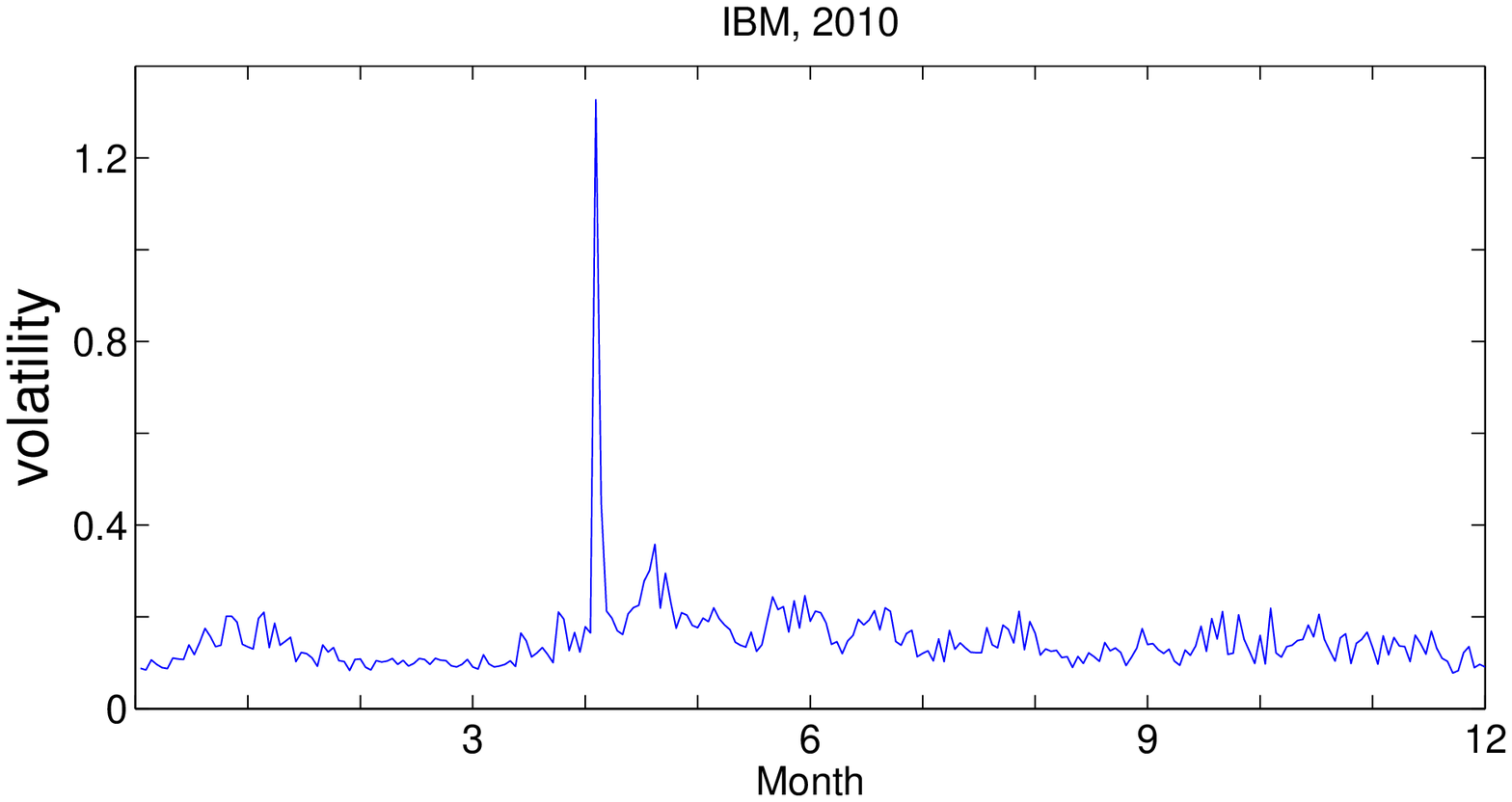}
		\caption{Hawkes volatility, 2010}
	\end{subfigure}
	\begin{subfigure}[b]{0.45\textwidth}
		\includegraphics[width=\textwidth]{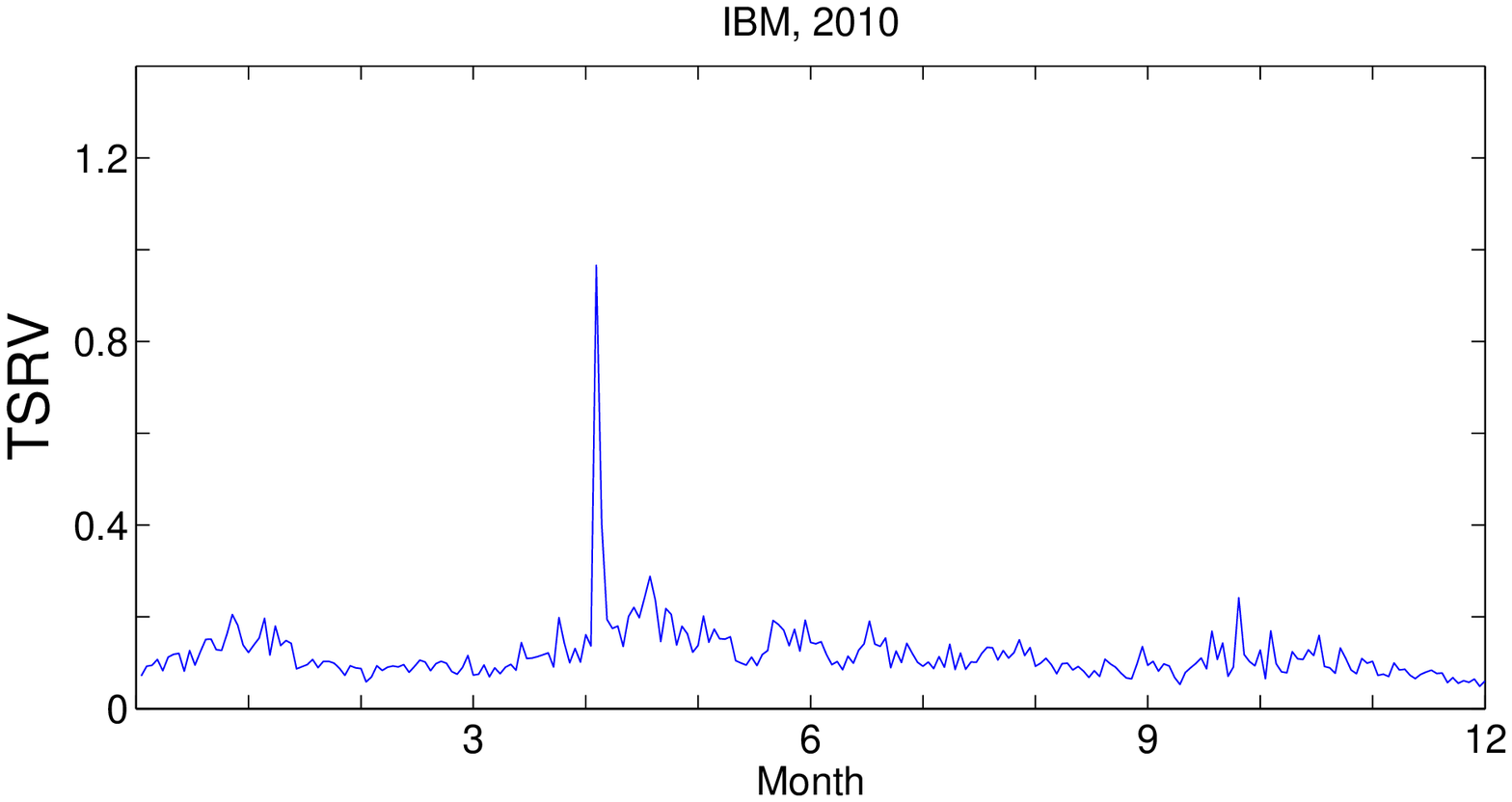}
		\caption{TSRV, 2010}
	\end{subfigure}
	\centering
	\begin{subfigure}[b]{0.45\textwidth}
		\includegraphics[width=\textwidth]{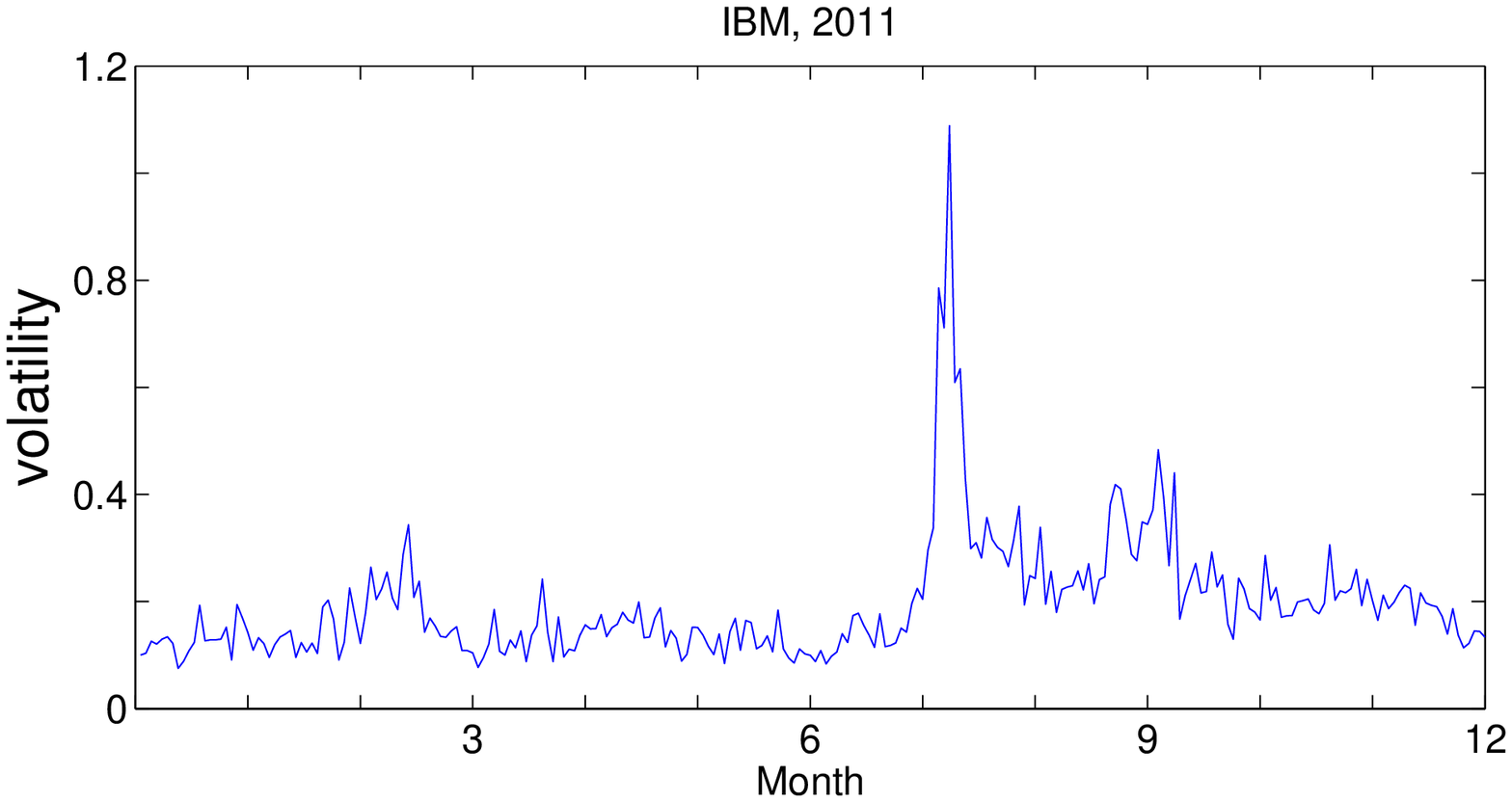}
		\caption{Hawkes volatility, 2011}
	\end{subfigure}
	\begin{subfigure}[b]{0.45\textwidth}
		\includegraphics[width=\textwidth]{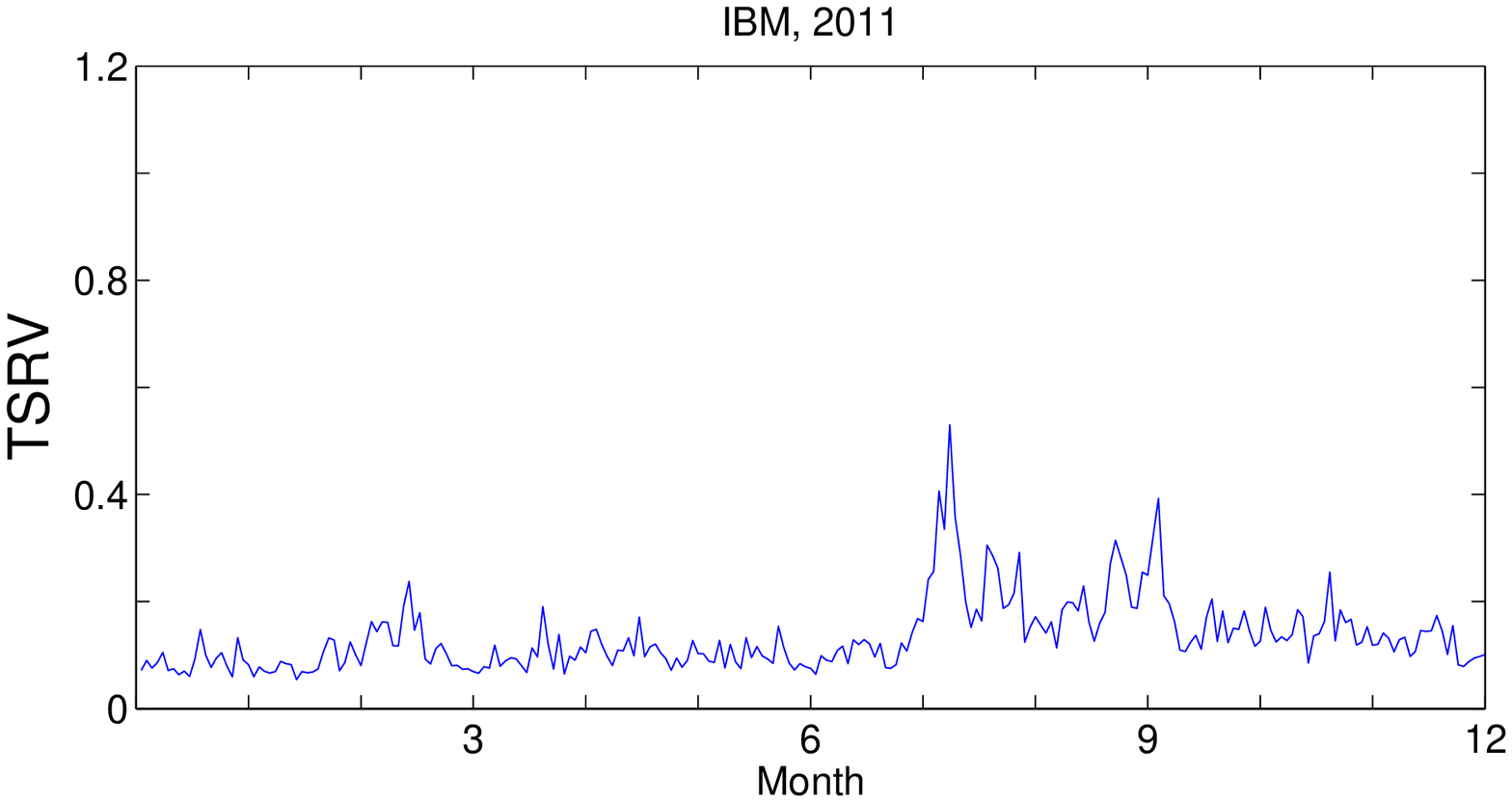}
		\caption{TSRV, 2011}
	\end{subfigure}
	\caption{Estimation results of the volatility, IBM, 2008-2011}\label{Fig:vol}
\end{figure}

In the estimation of the Hawkes models, the data of all arrival times over the sample period are used.
Even with a ten minute interval, usually more than a thousand records of arrival times are available, which is sufficient to provide a reliable result for the estimation in the aspect of the sample size
and hence adequate intraday analyses are possible.
The left of Figure~\ref{Fig:intraday} presents the intraday variation of the volatility of IBM on May 6, 2010, the day of the Flash Crash.
The estimation was performed on a ten minute basis from 10:00 to 15:30 and the cumulative volatility is plotted in the figure.
An abrupt increase in volatility is observed between 14:30 and 15:00 when the stock market crashed and the volatility was stabilized after 15:00. 
The right of the figure plots the intraday U-shape pattern of the volatility of IBM on an ordinary day of August 9, 2011.

\begin{figure}
\centering
\includegraphics[width=0.45\textwidth]{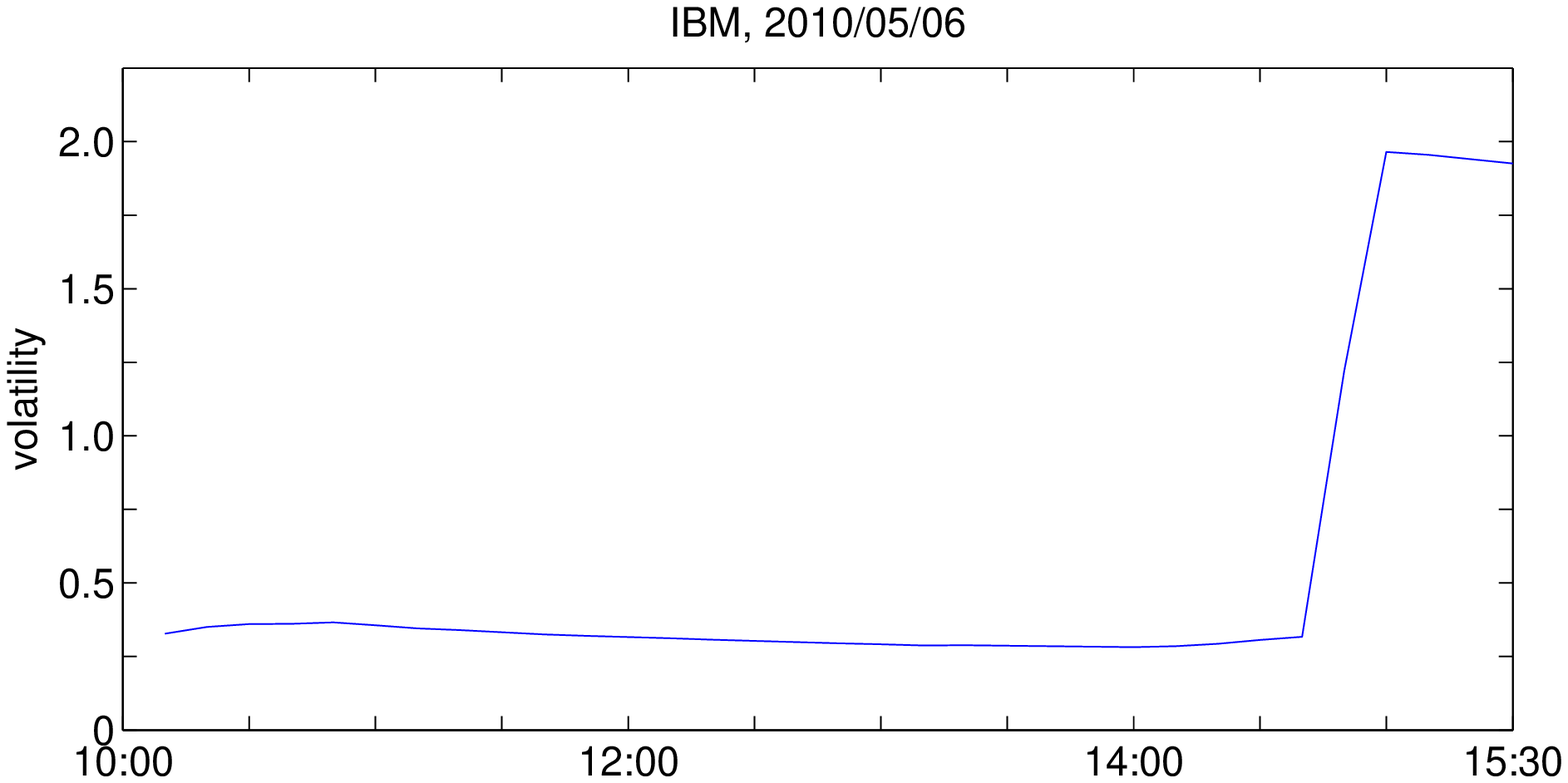}
\includegraphics[width=0.45\textwidth]{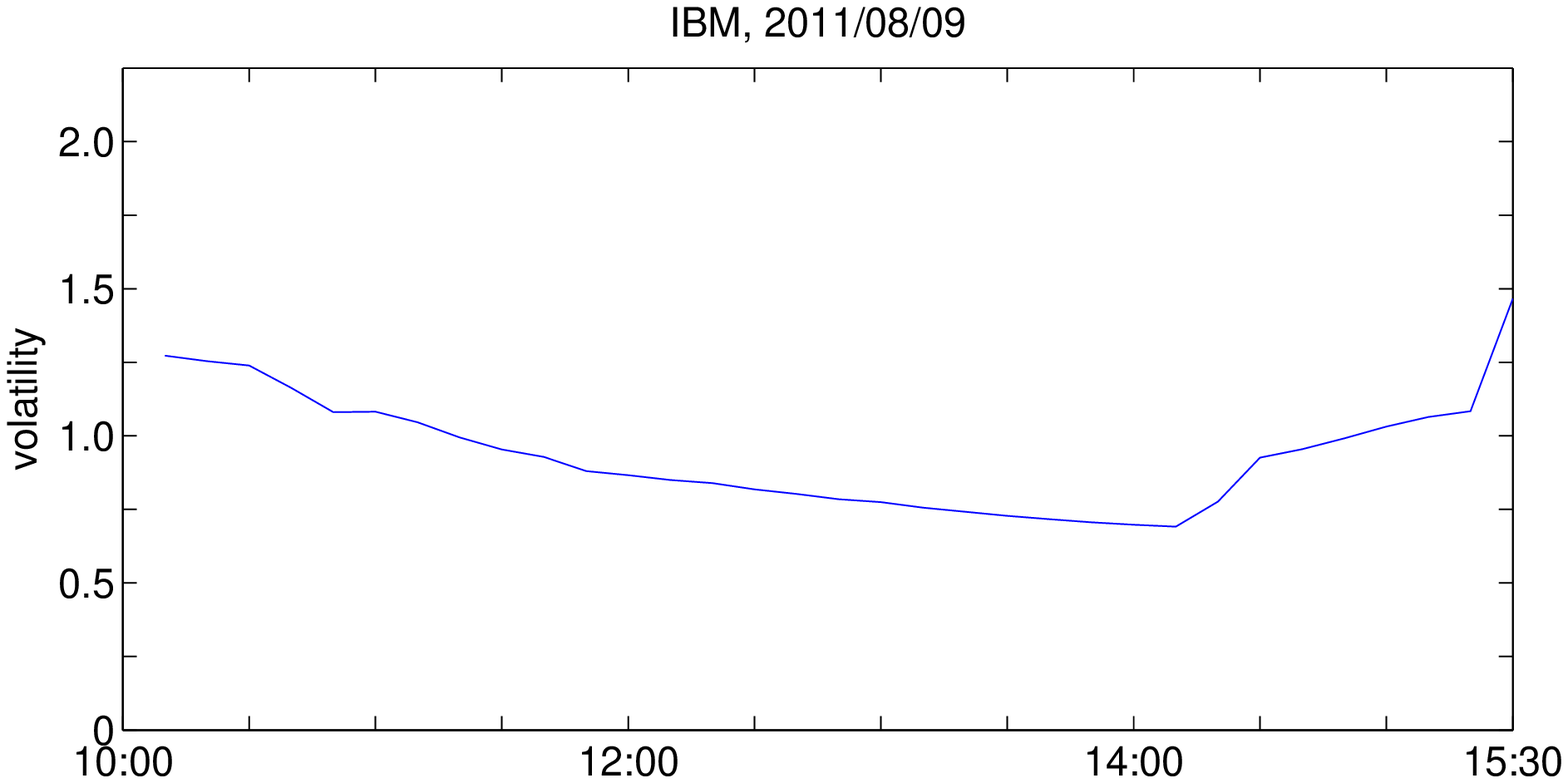}
\caption{Cumulative intraday volatitliy estimated under the marked Hawkes model with every ten minutes update}\label{Fig:intraday}
\end{figure}

In the volatility estimation perspective, the volatility formula under the i.i.d. mark in Corollary~\ref{Cor:iidvol} and the volatility formula without assuming the i.i.d. property are similar.
The empirical result also showed that the two formulas have similar values over time, as shown in Figure~\ref{Fig:iid}.
Therefore, for practical purposes, such as volatility computation, it may be sufficient to introduce an i.i.d. mark distribution.

\begin{figure}
\centering
\includegraphics[width=0.45\textwidth]{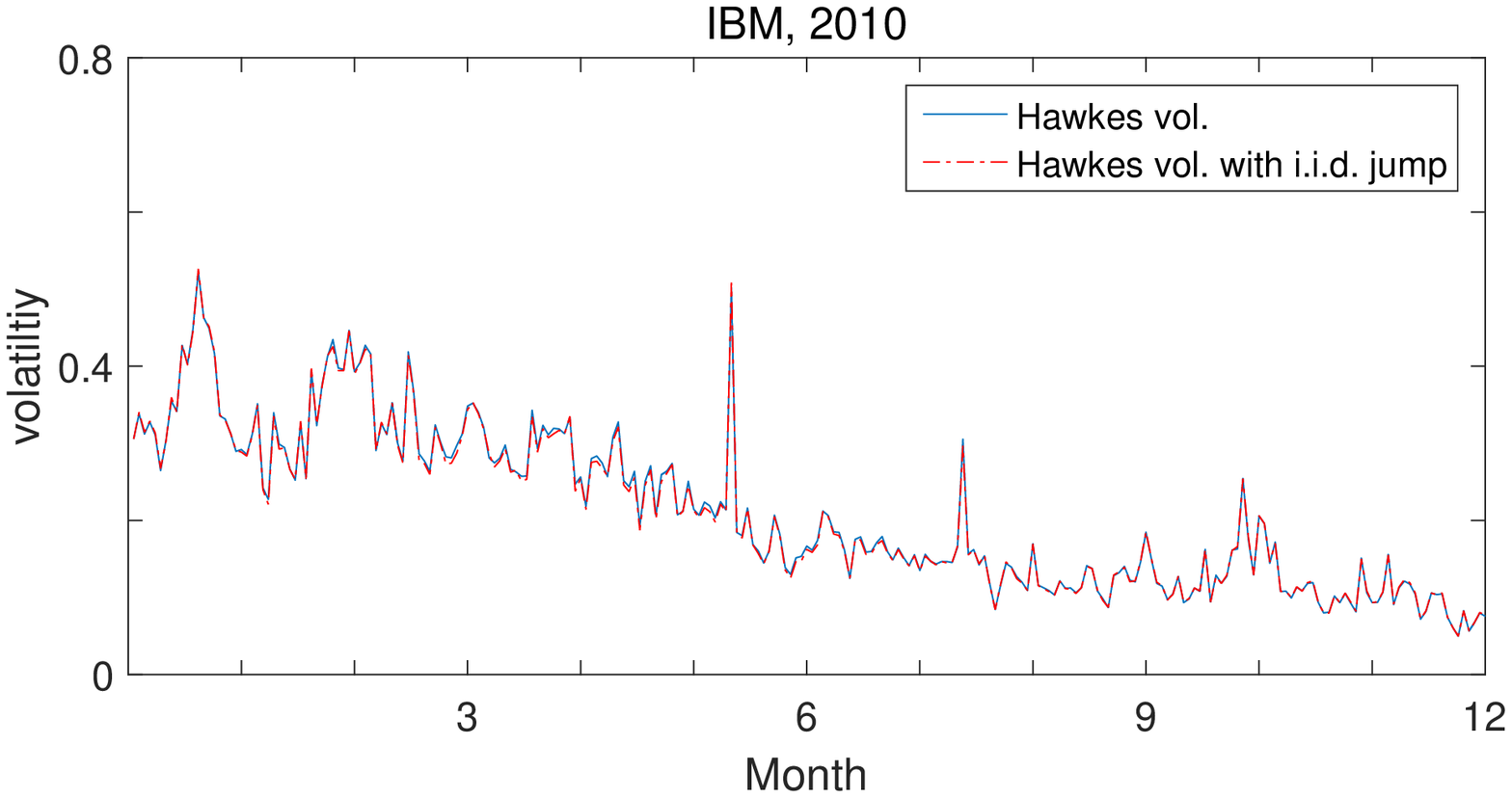}
\includegraphics[width=0.45\textwidth]{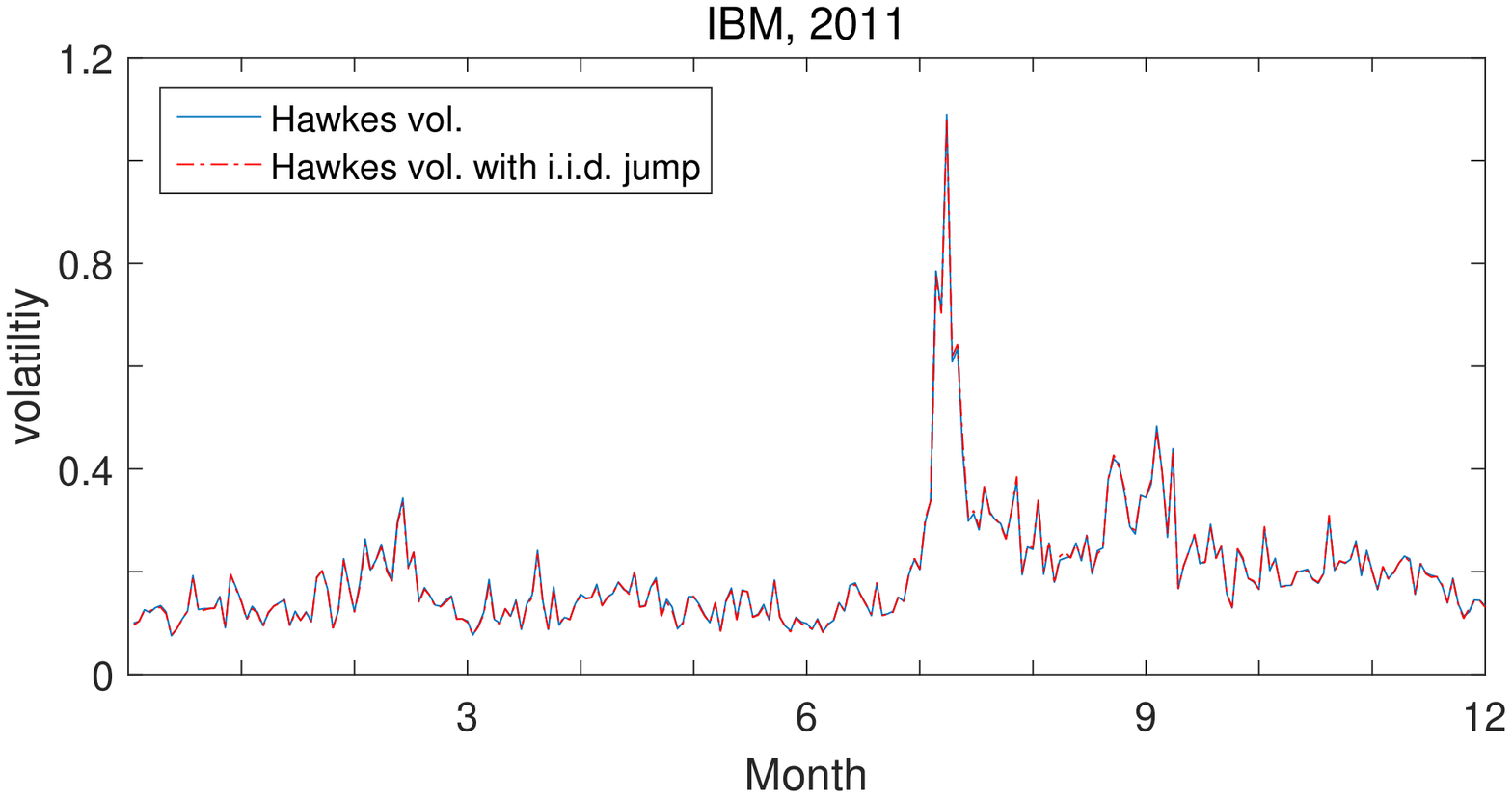}
\caption{Comparison between the daily volatility under the i.i.d. assumption and without i.i.d. assumption}\label{Fig:iid}
\end{figure}

\section{Concluding remark}\label{Sect:concl}
A marked Hawkes model was developed for price tick dynamics in equity markets.
A linear impact function was employed to describe the future effects of price jumps.
A specific distribution for the jump size was not assumed but the empirical distribution was used for the estimation.
This model is not limited to the independent mark since the empirical studies showed that the jump size depends on the ground intensities.
The volatility formula was derived based on stochastic calculus and statistical methods and the simulation studies showed that the Hawkes volatility and realized volatility are similar in the symmetric cases and the Hawkes volatility has less standard error.
On the other hand, there are biases when the underlying path is not symmetrical or the parameters vary with time.

A significant positive linear impact was observed, approximately 0.2, and various types, linear or humped shape, of the conditional mean structure of mark size in the empirical studies based on the equity prices reported in NYSE.
The Hawkes model is useful for estimating the intraday volatility particularly when the volatility is time varying.
The U shape seasonality of the changing volatility was observed and the interesting example of Flash Crash was examined.
As discussed in the simulation and empirical studies, the discrepancy between the Hawkes volatility and realized volatility and the biases from the sample volatility will be an important subject for a future study.
In the presence of the asymmetry in the price dynamics and the time varying parameters, a more robust estimation method is required for a more exact volatility computation.

One of the financial application of the marked Hawkes model is intraday volatility measurements. 
There is a discrepancy between the realized volatility and marked Hawkes volatility but the trend and the overall movements of two volatilities are consistent. 
In addition, in the marked Hawkes volatility estimation, the all available events reported in the exchange is used and hence even using the data in a relatively small interval, 
one can compute the Hawkes volatility and recognize the intraday changes of the volatility movement. 
This feature will help the traders, portfolio managers or algorithmic machines for the decision making by checking the changes in the intraday volatility trends of concerned assets.

The return model of the mark size will be considered for future work.
The distribution of the mark size depending on the current price of the underlying asset and 
it is worthwhile to examine the relationship between the mark size measured under the return process and future intensities.
In addition, it would be interesting to compare the performance of the marked Hawkes model with the ACD-GARCH model with modification of the future intensity as a function the marks.

\newpage
\bibliography{Mark}
\bibliographystyle{apalike}

\appendix
\section{Proof of Theorem~\ref{Thm:var}}\label{Proof:var}

We apply the stochastic integration theory to derive the second moment property of the two dimensional symmetric marked Hawkes process.
The quadratic (co)variation process of the semimartingales $X$ and $Y$ for $t \geq 0$ is defined as
$$ [X,Y]_t = X_t Y_t - \int_0^t X_{u-} \D Y_u - \int_0^t Y_{u-} \D X_u.$$
For the quadratic pure jump processes $X$ and $Y$, such as $\lambda_{gi}(t), N_{gi}(t)$ and $N_i (t)$ in the proposed model
\begin{equation}
[X]_t := [X,X]_t =  X_0^2 + \sum_{0\leq s<t} (\Delta X_s)^2 \label{Eq:QV}
\end{equation}
and
$$ [X,Y]_t =  X_0 Y_0 + \sum_{0\leq s<t} (\Delta X_s \Delta Y_s).$$

\subsection{Step 1}
This subsection derives the unconditional expectations of $\lambda^2_{gi}(t)$ and $\lambda_{gi}(t)\lambda_{gj}(t)$.
By the definition of the quadratic variation process, we have
\begin{align*}
\E[\lambda_{g1}^2 (t)] = \E[ [\lambda_{g1}]_t] + 2\E\left[\int_0^t \lambda_{g1}(u) \D \lambda_{g1}(u)\right].
\end{align*}
When the ground intensity $\lambda_{gi}$ is used as an integrator, we consider $\lambda_{gi}$ as the right continuous modification.
The stochastic integration is indeed a path-by-path Lebesgue-Stieltjes integral.
The integration part of the r.h.s. of the above equation is represented by
\begin{align*}
 \int_0^t \lambda_{g1}(u) \D \lambda_{g1}(u) &= \int_0^t \lambda_{g1}(u) \beta(\mu - \lambda_{g1}(u))\D u \\
&+ \int_{(0,t)\times\mathbb Z^+} q_s \beta g(k_1) \lambda_{g1}(u)N_1(\D u \times \D k_1) + \int_{(0,t)\times\mathbb Z^+} q_c \beta g(k_2) \lambda_{g1}(u) N_2(\D u \times \D k_2)
\end{align*}
or heuristically,
\begin{align*}
\D \lambda_{g1}(u) ={}\beta(\mu - \lambda_{g1}(u))\D u & + \int_{(u, u+\D u)\times\mathbb Z^+} q_s \beta g(k_1) N_1(\D u \times \D k_1) + \int_{(u, u+\D u)\times\mathbb Z^+}q_c \beta g(k_2) N_2(\D u \times \D k_2)\\
={} \beta(\mu - \lambda_{g1}(u))\D u & + \int_{(u, u+\D u)\times\mathbb Z^+} \alpha_s (1+(k_1-1)\eta) N_1(\D u \times \D k_1) \\
&+ \int_{(u, u+\D u)\times\mathbb Z^+}\alpha_c (1+(k_2-1)\eta) N_2(\D u \times \D k_2).
\end{align*}

Since the jump size of $\lambda_{g1}$ is represented by $\alpha_s (1+(k_1 -1)\eta)$ and $\alpha_c (1+(k_2 -1)\eta)$, using Eq.~\eqref{Eq:QV}, we have
\begin{align*}
\E[ [\lambda_{g1}]_t] ={}& \E [\lambda_{g1}^2(0)] + \E \left[\int_{(0,t)\times \mathbb Z^+} \alpha_s^2 (1+(k_1 -1)\eta)^2 N_1(\D u \times \D k_1) \right] \\
&+ \E \left[\int_{(0,t)\times \mathbb Z^+} \alpha_c^2 (1+(k_2-1)\eta)^2 N_2(\D u \times \D k_2) \right] \\
={}& \E [\lambda_{g1}^2(t)] + \int_0^t \alpha_s^2 \left(1+\E[ 2(k_1-1)\eta\lambda_{g1}(u)] + \E [(k_1^2 - 2k_1 + 1)\eta^2\lambda_{g1}(u)] \right) \D u  \\
&+ \int_0^t \alpha_c^2 \left(1+\E[ 2(k_2-1)\eta\lambda_{g2}(u)] + \E [(k_2^2 - 2k_2 + 1)\eta^2\lambda_{g2}(u)] \right) \D u\\
={}& \E [\lambda_{g1}^2(t)] + (\alpha_s^2 + \alpha_c^2)\{ 1+2(K_{1\lambda_{g1}}-1)\eta + (K_{1\lambda_{g1}}^{(2)} -2K_{1\lambda_{g1}} +1 )\eta^2 \}\E[\lambda_{g1}(t)]t \\
={}& \E [\lambda_{g1}^2(t)] + (\alpha_s^2 + \alpha_c^2)\bar{\bar K} \E[\lambda_{g1}(t)]t
\end{align*}
where the stationarity of $\lambda_{gi}$ and the symmetry are used.
In addition,
\begin{align*}
&\E\left[\int_0^t \lambda_{g1}(u) \D \lambda_{g1}(u)\right] =\E\left[\int_0^t \lambda_{g1}(u) \beta(\mu - \lambda_{g1}(u))\D u \right] \\
&\quad + \E\left[ \int_{(0,t)\times \mathbb Z^+}  q_s \beta g(k_1) \lambda_{g1}(u)N_1(\D u \times \D k_1) + \int_{(0,t)\times \mathbb Z^+} q_c \beta g(k_2)\lambda_{g1}(u) N_2(\D u \times \D k_2) \right]\\
&=(\beta \mu \E[\lambda_{g1}(t)] + [\alpha_s\{1+(K_{1\lambda^2_{g1}}-1)\eta\}-\beta]\E[\lambda_{g1}^2(t)] + \alpha_c\{1+(K_{2\lambda_{g1}\lambda_{g2}}-1)\eta \}\E [\lambda_{g1}(t) \lambda_{g2}(t)]) t\\
&=(\beta \mu \E[\lambda_{g1}(t)] + (\breve \alpha_s-\beta)\E[\lambda_{g1}^2(t)] + \tilde \alpha_c\E [\lambda_{g1}(t) \lambda_{g2}(t)]) t
\end{align*}
 
Similarly, 
\begin{align*}
\E[\lambda_{g1}(t)\lambda_{g2}(t)] = \E[ [\lambda_{g1}, \lambda_{g2}]_t] + \E\left[\int_0^t \lambda_{g1}(u) \D \lambda_{g2}(u)\right] +  \E\left[\int_0^t \lambda_{g2}(u) \D \lambda_{g1}(u)\right],
\end{align*}
and, for the covariation process $[\lambda_{g1}, \lambda_{g2}]_t$, we have
\begin{align*}
\E[ [\lambda_{g1},\lambda_{g2}]_t] ={}& \E [\lambda_{g1}(0)\lambda_{g2}(0)] + \E \left[\int_{(0,t)\times \mathbb Z^+} \alpha_s\alpha_c (1+(k_1-1)\eta)^2 N_1(\D u \times \D k_1) \right] \\
&+ \E \left[\int_{(0,t)\times \mathbb Z^+} \alpha_s\alpha_c (1+(k_2-1)\eta)^2 N_2(\D u \times \D k_2) \right] \\
={}& \E[\lambda_{g1}(t)\lambda_{g2}(t)] + \int_0^t\alpha_s\alpha_c\left(1+\E[ 2(k_1-1)\eta\lambda_{g1}(u)] + \E [(k_1^2 - 2k_1 + 1)\eta^2\lambda_{g1}(u)] \right) \D u\\
&+ \int_0^t\alpha_s\alpha_c \left(1+\E[ 2(k_2-1)\eta\lambda_{g2}(u)] + \E [(k_2^2 - 2k_2 + 1)\eta^2\lambda_{g2}(u)] \right) \D u\\
={}& \E[\lambda_{g1}(t)\lambda_{g2}(t)] + \alpha_s\alpha_c \left\{ 1+2(K_{1\lambda_{g1}}-1)\eta+ (K_{1\lambda_{g1}}^{(2)} -2K_{1\lambda_{g1}} +1 )\eta^2 \right\}\E[\lambda_{g1}(t)]t\\
&+ \alpha_s\alpha_c \left\{ 1+2(K_{2\lambda_{g2}}-1)\eta+ (K_{2\lambda_{g2}}^{(2)} -2K_{2\lambda_{g2}} +1 )\eta^2 \right\}\E[\lambda_{g2}(t)]t\\
={}& \E[\lambda_{g1}(t)\lambda_{g2}(t)] +  2\alpha_s\alpha_c\{ 1+2(K_{1\lambda_{g1}}-1)\eta + (K_{1\lambda_{g1}}^{(2)} -2K_{1\lambda_{g1}} +1 )\eta^2 \}\E[\lambda_{g1}(t)]t \\
={}& \E[\lambda_{g1}(t)\lambda_{g2}(t)] +  2\alpha_s\alpha_c \bar{\bar K} \E[\lambda_{g1}(t)]t
\end{align*}
and
\begin{align*}
&\E\left[\int_{(0,t)\times \mathbb Z^+} \lambda_{g1}(u) \D \lambda_{g2}(u)\right] =\E\left[\int_0^t \lambda_{g1}(u) \beta(\mu - \lambda_{g2}(u))\D u \right] \\
& \quad + \E \left[ \int_{(0,t)\times \mathbb Z^+}  q_c \beta g(k_1) \lambda_{g1}(u) N_1(\D u \times \D k_1) + \int_{(0,t)\times \mathbb Z^+}  q_s \beta g(k_2)\lambda_{g1}(u) N_2(\D u \times \D k_2) \right]\\
&=(\beta \mu \E[\lambda_{g1}(t)] + \alpha_c\{1+(K_{1\lambda^2_{g1}}-1)\eta\}\E[\lambda_{g1}^2(t)] + [\alpha_s\{1+(K_{2\lambda_{g1}\lambda_{g2}}-1)\eta \}-\beta]\E [\lambda_{g1}(t) \lambda_{g2}(t)]) t\\
&=(\beta \mu \E[\lambda_{g1}(t)] + \breve \alpha_c \E[\lambda_{g1}^2(t)] + (\tilde \alpha_s-\beta)\E [\lambda_{g1}(t) \lambda_{g2}(t)]) t.
\end{align*}

Combining the above results, we have
\begin{align*}
0 &= (\alpha_s^2 + \alpha_c^2)\bar{\bar K} \E[\lambda_{g1}(t)]t + 2\left\{ \beta \mu \E[\lambda_{g1}(t)] + (\breve \alpha_s-\beta)\E[\lambda_{g1}^2(t)] + \tilde \alpha_c\E [\lambda_{g1}(t) \lambda_{g2}(t)] \right\} t \\
0 &= 2\alpha_s\alpha_c \bar{\bar K} \E[\lambda_{g1}(t)]t + 2\left\{ \beta \mu \E[\lambda_{g1}(t)] + \breve \alpha_c \E[\lambda_{g1}^2(t)] + (\tilde \alpha_s-\beta)\E [\lambda_{g1}(t) \lambda_{g2}(t)] \right \}t
\end{align*}
and, in a matrix form,
\begin{align*}
\begin{bmatrix} \E[\lambda_{g1}^2 (t)] \\ \E[\lambda_{g1}(t)\lambda_{g2}(t)] \end{bmatrix}
= -\frac{1}{2}\E[\lambda_{g1}(t)] \bf M^{-1} 
\begin{bmatrix} (\alpha_s^2 + \alpha_c^2) \bar{\bar K} + 2\beta\mu \\ 2(\alpha_s\alpha_c\bar{\bar K}+\beta\mu) \end{bmatrix}
\end{align*}
where 
$$ \mathbf{M} = \begin{bmatrix} \breve \alpha_s - \beta & \tilde \alpha_c \\ \breve \alpha_c & \tilde \alpha_s - \beta \end{bmatrix}.
$$

\subsection{Step 2}
In this step, the goal is to calculate $\E[\lambda_{g1}(t)N_1(t)]$ and $\E[\lambda_{g1}(t)N_2(t)]$.
We have
\begin{align*}
\E[\lambda_{g1}(t)N_1(t)] = \E[[\lambda_{g1}, N_1]_t] + \E\left[\int_{(0,t)\times \mathbb Z^+}\lambda_{g1}(u) k_1 N_1(\D u\times \D k_1)\right] + \E\left[\int_0^t N_1(u-) \D \lambda_{g1}(u)\right].
\end{align*}
For each component in the r.h.s., 
\begin{align*}
\E[[\lambda_{g1}, N_1]_t] &= \E\left[\int_{(0,t)\times \mathbb Z^+} \alpha_s(1+(k_1-1)\eta)k_1 N_1(\D u\times \D k_1) \right]\\
&= \alpha_s (K_{1\lambda_{g1}}+(K_{1\lambda_{g1}}^{(2)}-K_{1\lambda_{g1}})\eta)\E[\lambda_{g1}(t)]t\\
&= \alpha_s \bar K \E[\lambda_{g1}(t)]t
\end{align*}
and
\begin{align*}
\E\left[\int_{(0,t)\times \mathbb Z^+} \lambda_{g1}(u) k_1 N_1(\D u\times \D k_1)\right] = \int_0^t \E[k_1\lambda_{g1}^2(u)]\D u = K_{1\lambda^2_{g1}} \E[\lambda_{g1}^2(t)]t
\end{align*}
and
\begin{align*}
\E\left[\int_0^t N_1(u-) \D \lambda_{g1}(u)\right] ={}& \E\left[ \int_0^t N_1(u) \beta(\mu - \lambda_{g1}(u))\D u \right] \\
&+ \E \left[ \int_{(0,t)\times \mathbb Z^+} N_1(u-) \alpha_s(1+(k_1-1)\eta) N_1(\D u \times \D k_1) \right] \\
&+ \E \left[ \int_{(0,t)\times \mathbb Z^+} N_1(u-) \alpha_c(1+(k_2-1)\eta) N_2(\D u \times \D k_2) \right] \\
={}& \int_0^t \{ \beta \mu K \E [\lambda_{g1}(u)]u + (\acute \alpha_s - \beta)\E[\lambda_{g1}(u)N_1(u)] + \grave \alpha_c\E[\lambda_{g2}(u)N_1(u)] \} \D u\\
={}& \int_0^t \{ \beta \mu K \E [\lambda_{g1}(u)]u + (\acute \alpha_s - \beta)\E[\lambda_{g1}(u)N_1(u)] + \grave \alpha_c\E[\lambda_{g1}(u)N_2(u)] \} \D u.
\end{align*}

Similarly,
\begin{align*}
\E[\lambda_{g1}(t)N_2(t)] = \E[[\lambda_{g1}, N_2]_t] + \E\left[\int_{(0,t)\times \mathbb Z^+}\lambda_{g1}(u) k_2 N_2(\D u\times \D k_2)\right] + \E\left[\int_0^t N_2(u-) \D \lambda_{g1}(u)\right]
\end{align*}
and
\begin{align*}
\E[[\lambda_{g1}, N_2]_t] &= \E\left[\int_{(0,t)\times \mathbb Z^+} \alpha_c(1+(k_2-1)\eta)k_2 N_2(\D u\times \D k_2) \right]\\
&=\alpha_c (K_{1\lambda_{g1}}+(K^{(2)}-K_{1\lambda_{g1}})\eta)\E[\lambda_{g1}(t)]t\\
&=\alpha_c \bar K \E[\lambda_{g1}(t)]t
\end{align*}
and
\begin{align*}
\E\left[\int_{(0,t)\times \mathbb Z^+} \lambda_{g1}(u) k_2 N_2(\D u\times \D k_2)\right] = \int_0^t \E[k_2\lambda_{g1}(u)\lambda_{g2}(u)]\D u = K_{2\lambda_{g1}\lambda_{g2}} \E[\lambda_{g1}(t)\lambda_{g2}(t)]t
\end{align*}
and
\begin{align*}
\E\left[\int_0^t N_2(u-) \D \lambda_{g1}(u)\right] ={}& \E\left[ \int_0^t N_2(u) \beta(\mu - \lambda_{g1}(u))\D u \right] \\
&+ \E \left[ \int_{(0,t)\times \mathbb Z^+} N_2(u-) \alpha_s(1+(k_1-1)\eta) N_1(\D u \times \D k_1) \right] \\
&+ \E \left[ \int_{(0,t)\times \mathbb Z^+} N_2(u-) \alpha_c(1+(k_2-1)\eta) N_2(\D u \times \D k_2) \right] \\
={}& \int_0^t \{ \beta \mu K_{1\lambda_{g1}} \E [\lambda_{g1}(u)]u + (\grave \alpha_s - \beta)\E[\lambda_{g1}(u)N_2(u)] + \acute \alpha_c\E[\lambda_{g2}(u)N_2(u)] \} \D u\\
={}& \int_0^t \{ \beta \mu K_{1\lambda_{g1}} \E [\lambda_{g1}(u)]u + (\grave \alpha_s - \beta)\E[\lambda_{g1}(u)N_2(u)] + \acute \alpha_c\E[\lambda_{g1}(u)N_1(u)] \} \D u.
\end{align*}

Combing the above results, we have
\begin{align*}
\E[\lambda_{g1}(t)N_1(t)] ={}& \alpha_s \bar K \E[\lambda_{g1}(t)]t + K_{1\lambda^2_{g1}} \E[\lambda_{g1}^2(t)]t \\
&+ \int_0^t \{ \beta \mu K \E [\lambda_{g1}(u)]u + (\acute \alpha_s - \beta)\E[\lambda_{g1}(u)N_1(u)] + \grave \alpha_c\E[\lambda_{g1}(u)N_2(u)] \} \D u\\
\E[\lambda_{g1}(t)N_2(t)] = {}& \alpha_c \bar K \E[\lambda_{g1}(t)]t +  K_{2\lambda_{g1}\lambda_{g2}} \E[\lambda_{g1}(t)\lambda_{g2}(t)]t \\
&+ \int_0^t \{ \beta \mu K \E [\lambda_{g1}(u)]u + (\grave \alpha_s - \beta)\E[\lambda_{g1}(u)N_2(u)] + \acute \alpha_c\E[\lambda_{g1}(u)N_1(u)] \} \D u.
\end{align*}
Let
$$ \mathbf{M}_2 = \begin{bmatrix} \acute \alpha_s - \beta & \grave \alpha_c \\ \acute \alpha_c & \grave \alpha_s - \beta \end{bmatrix}, \quad \mathbf{K}_2 = \begin{bmatrix} K_{1\lambda_{g1}^2} & 0 \\ 0 & K_{2\lambda_{g1}\lambda_{g2}} \end{bmatrix}.$$
Then,
\begin{align*}
\begin{bmatrix}\dfrac{\D \E[\lambda_{g1}(t)N_1(t)]}{\D t} \\ \dfrac{\D \E[\lambda_{g1}(t)N_2(t)]}{\D t}\end{bmatrix}
= \mathbf{M}_2 \begin{bmatrix} \E[\lambda_{g1}(t)N_1(t)] \\ \E[\lambda_{g1}(t)N_2(t)] \end{bmatrix} + 
\begin{bmatrix}\alpha_s \bar K \E[\lambda_{g1}(t)] + K_{1\lambda^2_{g1}}\E[\lambda_{g1}^2(t)]+ \beta\mu K\E[\lambda_{g1}(t)]t \\ \alpha_c \bar K \E[\lambda_{g1}(t)] + K_{2\lambda_{g1}\lambda_{g2}}\E[\lambda_{g1}(t)\lambda_{g2}(t)]+ \beta\mu K \E[\lambda_{g1}(t)]t \end{bmatrix}\\
= \mathbf{M}_2 \begin{bmatrix} E[\lambda_{g1}(t)N_1(t)] \\ \E[\lambda_{g1}(t)N_2(t)] \end{bmatrix} 
+ \E[\lambda_{g1}(t)] \left(\begin{bmatrix}1 \\1\end{bmatrix}\beta \mu K t + \begin{bmatrix} \alpha_s \bar K \\ \alpha_c \bar K\end{bmatrix}
-\frac{1}{2} \mathbf{K}_2 \mathbf{M}^{-1} 
\begin{bmatrix}  (\alpha_s^2 + \alpha_c^2) \bar{\bar K} + 2\beta\mu  \\ 2(\alpha_s\alpha_c\bar{\bar K}+\beta\mu) \end{bmatrix} \right)
\end{align*}
and the particular solution is
\begin{align*}
\begin{bmatrix} \E[\lambda_{g1}(t)N_1(t)] \\ \E[\lambda_{g1}(t)N_2(t)] \end{bmatrix}
= -\E[\lambda_{g1}(t)] & \left\{  \beta \mu \mathbf{K} \mathbf{M}_2^{-1} \begin{bmatrix}1\\1 \end{bmatrix}t + \mathbf{M}_2^{-1} \begin{bmatrix} \alpha_s \bar K \\ \alpha_c \bar K\end{bmatrix} 
-\frac{1}{2} \mathbf{M}_2^{-1} \mathbf{K}_2 \mathbf{M}^{-1} 
\begin{bmatrix}  (\alpha_s^2 + \alpha_c^2) \bar{\bar K}  \\ 2\alpha_s\alpha_c\bar{\bar K} \end{bmatrix} \right.\\
{}&\left. + \mathbf{M}_2^{-1}(\mathbf{K}\mathbf{M}_2^{-1}- \mathbf{K}_2 \mathbf{M}^{-1}) \begin{bmatrix} \beta\mu \\ \beta\mu \end{bmatrix} \right\}.
\end{align*}
Note that if $\mathbf{K}\mathbf{M}_2^{-1}- \mathbf{K}_2 \mathbf{M}^{-1}$ is close to zero, then the last term of the above equation is also close to zero and the formula is similar to the one in the simple Hawkes model.
\subsection{Step 3}
In the final step, $\E[N_1^2(t)]$ and $\E[N_1(t)N_2(t)]$ are derived.
We have
\begin{align*}
\E[N_1^2(t)] &= \E [[N_1]_t] + 2\E\left[\int_{(0,t)\times \mathbb Z^+} k_1 N_1(u) N_1(\D u \times \D k_1)  \right] \\
&= \E [k_1^2 \lambda_{g1}(t)]t + 2\int_0^t \E[k_1 \lambda_{g1}(u) N_1(u) ] \D u \\
&= K^{(2)} \E [\lambda_{g1}(t)]t + 2 \int_0^t K_{1\lambda_{g1} N_1} \E[ \lambda_{g1}(u) N_1(u)] \D u
\end{align*}
and
\begin{align*}
\E[N_1(t)N_2(t)] &= \E[[N_1,N_2]_t] + \E\left[\int_{(0,t)\times \mathbb Z^+}k_2 N_1(u-) N_2(\D u\times \D k_2)\right] + \E\left[\int_{(0,t)\times \mathbb Z^+} k_1 N_2(u-) N_1(\D u\times \D k_1)\right] \\
&= \int_0^t  K_{2\lambda_{g2} N_1} \E[ \lambda_{g2}(u) N_1(u)] \D u + \int_0^t  K_{1\lambda_{g1} N_2} \E[ \lambda_{g1}(u) N_2(u)] \D u \\
&= 2 \int_0^t  K_{1\lambda_{g1} N_2} \E[ \lambda_{g1}(u) N_2(u)] \D u.
\end{align*}
Then
\begin{align*}
\begin{bmatrix} \dfrac{ \D \E[N_1^2(t)]}{\D t} \\ \dfrac{\D \E[N_1(t)N_2(t)]}{\D t} \end{bmatrix} = 2 \mathbf{K}_3
\begin{bmatrix} \E[ \lambda_{g1}(t) N_1(t)] \\  \E[ \lambda_{g1}(t) N_2(t)] \end{bmatrix} + \begin{bmatrix} K^{(2)} \E [\lambda_{g1}(t)] \\ 0 \end{bmatrix} 
\end{align*}
where
$$ \mathbf{K}_3 = \begin{bmatrix} K_{1\lambda_{g1} N_1} & 0 \\ 0 & K_{1\lambda_{g1} N_2} \end{bmatrix}$$
and
\begin{align*}
&\begin{bmatrix} \E[N_1^2(t)] \\ \E[N_1(t)N_2(t)] \end{bmatrix} = -\E[\lambda_{g1}(t)]\mathbf{K}_3 \left\{  \beta \mu \mathbf{K} \mathbf{M}_2^{-1} \begin{bmatrix}1\\1 \end{bmatrix}t^2 \right. \\
{}&\left. + \left( 2\mathbf{M}_2^{-1} \begin{bmatrix} \alpha_s \bar K \\ \alpha_c \bar K\end{bmatrix} 
- \mathbf{M}_2^{-1} \mathbf{K}_2 \mathbf{M}^{-1} 
\begin{bmatrix}  (\alpha_s^2 + \alpha_c^2) \bar{\bar K}  \\ 2\alpha_s\alpha_c\bar{\bar K} \end{bmatrix} 
 + 2\mathbf{M}_2^{-1}(\mathbf{K}\mathbf{M}_2^{-1}- \mathbf{K}_2 \mathbf{M}^{-1}) \begin{bmatrix} \beta\mu \\ \beta\mu \end{bmatrix} - \begin{bmatrix}K^{(2)}/K_{1\lambda_{g1} N_1} \\ 0\end{bmatrix} \right)t \right\}.
\end{align*}
\end{document}